\documentclass[letterpaper,titlepage,11pt]{article}
\pdfoutput=1
\usepackage[colorlinks=true,urlcolor=blue,citecolor=magenta]{hyperref}
\usepackage{amssymb,amsmath,amsfonts}
\usepackage{epsfig}
\usepackage[sort&compress, merge]{natbib}
\usepackage{epsfig}
\usepackage{graphicx}
\usepackage{epstopdf}
\usepackage{caption}
\usepackage{subcaption}
\usepackage{amscd}
\usepackage{amsthm}
\usepackage{latexsym}
\usepackage{amsbsy}
\usepackage{bbm}
\usepackage[english]{babel}

\usepackage{tabularx}

\newcommand{\be}{\begin{eqnarray}}
\newcommand{\ee}{\end{eqnarray}}
\newcommand{\beq}{\begin{eqnarray}}
\newcommand{\eeq}{\end{eqnarray}}

\bibpunct{[}{]}{;}{n}{,}{,}

\newtheorem{definition}{Definition}[section]

\newtheorem{theorem}[definition]{Theorem}

\newtheorem{corollary}[definition]{Corollary}

\def\clock{{\count0=\time
           \divide\count0 60
           \ifnum\count0<10 0\fi\the\count0
           \multiply\count0 -60 \advance\count0 \time
           :\ifnum\count0<10 0\fi \the\count0
         }}
\newcommand{\timestamp}{{\small\vbox{\hbox{\tt\jobname.tex}
\hbox{\the\day/\the\month/\the\year, \clock}}}}

\numberwithin{equation}{section}

\begin{document}

\begin{titlepage}
\ \ \vskip 1.8cm

\centerline{\Huge \bf Blackfolds, Plane Waves and } 
\vskip 0.3cm
\centerline{\Huge \bf Minimal Surfaces}

\vskip 1.6cm \centerline{\bf Jay Armas$^{1,2}$ and Matthias Blau$^{2}$} \vskip 0.7cm

\begin{center}
\sl $^{1}$ Physique Th\'{e}orique et Math\'{e}matique \\
Universit\'{e} Libre de Bruxelles and International Solvay Institutes \\
ULB-Campus Plaine CP231, B-1050 Brussels, Belgium\\
\end{center}
\vskip 0.3cm

\begin{center}
\sl $^{2}$ Albert Einstein Center for Fundamental Physics, University of Bern,\\
Sidlerstrasse 5, 3012 Bern, Switzerland
\end{center}
\vskip 0.3cm

\centerline{\small\tt jarmas@ulb.ac.be, blau@itp.unibe.ch}

\vskip 1.3cm \centerline{\bf Abstract} \vskip 0.2cm \noindent
Minimal surfaces in Euclidean space provide examples of possible non-compact horizon geometries and topologies
in asymptotically flat space-time. On the other hand, the existence of limiting surfaces in the space-time
provides a simple mechanism for making these configurations compact. Limiting surfaces appear naturally in a
given space-time by making minimal surfaces rotate but they are also inherent to plane wave or de Sitter
space-times in which case minimal surfaces can be static and compact. We use the blackfold approach in order to
scan for possible black hole horizon geometries and topologies in asymptotically flat, plane wave and de Sitter
space-times. In the process we uncover several new configurations, such as black helicoids and catenoids, some
of which have an asymptotically flat counterpart. In particular, we find that the ultraspinning regime of
singly-spinning Myers-Perry black holes, described in terms of the simplest minimal surface (the plane), can be
obtained as a limit of a black helicoid, suggesting that these two families of black holes are connected. We also show that minimal surfaces embedded in spheres 
rather than Euclidean space can be used to construct static compact horizons in asymptotically de Sitter space-times.

\end{titlepage}

\tableofcontents

\section{Introduction}  \label{sec:intro}

Black holes in higher-dimensions are hard to classify and to construct
analytically, as Einstein equations become more intricate and complex
as the number of space-time dimensions is increased. In particular, in
asymptotically flat space-times in dimensions $D\ge6$ very few black
hole solutions are known analytically and only classification schemes,
which do not specify the solutions uniquely, based on the domain structure
are known \cite{Harmark:2009dh, Armas:2011ed}. In space-times with more
non-trivial geometry, such as plane wave and de Sitter space-times, the
problem of finding and classifying black holes only becomes aggravated.

However, recently, progress in understanding the phase structure of
black holes in $D\ge6$ has been made based on effective theories
and numerical methods. One of these effective theories, known as
the blackfold approach \cite{Emparan:2009cs, Emparan:2009at},
describing the long wavelength dynamics of black branes in a
derivative expansion including hydrodynamic and elastic degrees
of freedom \cite{Armas:2012jg, Armas:2013hsa,Armas:2013goa, Armas:2014bia,
Armas:2014rva}, has allowed to scan for non-trivial black hole horizon
topologies in asymptotically flat and (Anti)-de Sitter space-times
\cite{Emparan:2009vd,Armas:2010hz,Caldarelli:2010xz,Emparan:2011hg}. These
new black hole topologies include black rings in higher-dimensions,
black odd-spheres and black cylinders, some of which have
been constructed numerically \cite{Kleihaus:2012xh,Dias:2014cia,
Kleihaus:2014pha, Figueras:2014dta}. However, these works have only scratched the surface
of the entire set of possible horizon topologies.

This paper has a two-fold purpose: on the one hand it aims at providing
evidence for more complicated black hole horizon geometries and topologies
in different space-times and, on the other hand, to show that plane
wave space-times in vacuum allow for a very rich phase structure of
higher-dimensional black holes. The key ingredient in this work is
the use of established results in classical minimal surface theory in
higher-dimensional Euclidean and spherical spaces in order to construct
new compact horizon topologies using the blackfold approach.

Regarding the first input, plane wave space-times, 
we note that for vacuum plane wave space-times no exact analytic
black hole solutions are known, though attempts to construct such
solutions using the blackfold approach have been made in the past
\cite{LeWitt:2009qx}.\footnote{In supergravity, exact black hole
solutions that are asymptotically plane wave space-times have been
found. See \cite{Hubeny:2002nq,Hubeny:2002pj,Liu:2003cta, Gimon:2003xk,
Hubeny:2003ug} for work done in the context of supergravity.} The
configurations we shall construct in this paper should be thought of as black
holes in plane wave backgrounds. However, they are not necessarily asymptotically
plane wave black holes in some strict sense, as defined e.g.\ in
\cite{LeWitt:2008zx}. We will nevertheless occasionally refer to them
simply as asymptotically plane wave black holes in the following, and we
will come back to this issue in the concluding Sec.~\ref{sec:conclusions}.

Regarding the second input, it is a well known result from classical
minimal surface theory that minimal surfaces in $\mathbb{R}^3$ must be
non-compact \cite{Nitsche:2011}, and our aim is to show how these can
nevertheless be used to construct compact black hole horizons.
To illustrate this, note that compact minimal surfaces are found
everywhere in nature, the simplest example being that of soap films. Soap
films are thin surfaces with equal pressure on each of its sides and
are characterised by a surface tension. The surface tension acts
as a force that tries to shrink the area of the surface and hence
equilibrium configurations are minimal. 

The most common example of a
soap film is that which forms on a bubble wand after dipping it in a
soapy solution. Commonly, bubble wands have a circular shape at one of
their ends and hence the soap film takes the form of disc. Indeed, this is
the simplest example of a compact minimal surface: a plane $\mathbb{R}^2$
embedded in $\mathbb{R}^3$ with a circular boundary. From this example one
draws the following conclusion for surfaces in $\mathbb{R}^3$: in order
to create compact minimal surfaces one needs to introduce boundaries in
the embedded space.

The phenomenology of soap films is rather different than that of soap
bubbles. While soap films have an equal pressure on each side of its
surface and hence its equilibrium states are minimal surfaces, soap
bubbles have an internal pressure different from the exterior pressure
and hence, due to an interplay with the surface tension, its equilibrium
states are surfaces of non-vanishing constant mean extrinsic curvature. 

As we will see, the phenomenology of certain black brane configurations
can be similar either to that of soap films or to that of soap bubbles.
In particular, black branes share one common feature with soap films:
they are also characterised by a tension. In fact, it was noted in
\cite{Emparan:2009vd} (and we will review this in Sec.~\ref{sec:theorems})
that quite generally minimal surfaces in $\mathbb{R}^3$ may provide
non-trivial geometries for static non-compact black brane horizons in
asymptotically flat space-time.\footnote{These must be subjected to
regularity constraints such as no curvature divergences.} We observe
here that rotation provides a simple mechanism for making some of these
geometries compact, at least in some directions.

More generally, 
boundaries can be created in a given embedding space-time by introducing
limiting surfaces where the brane is forced to move at the speed
of light. Introducing rotation on a geometry implies the existence
of a stationary background Killing vector field and, generically,
of an ergo-region in the ambient space-time. Rotation involves the
existence of a $U(1)$ family of isometries inherent to an $\mathbb{R}^2$
plane and hence its boundary - defined by the limiting surface - will
always be a circle on that plane. However, there are other ways of
introducing limiting surfaces. With direct analogy to the bubble wand,
one can consider embedding space-times where limiting surfaces are
naturally present such as in de Sitter space-times, where the limiting
surface is located at the cosmological horizon and its shape is always
a higher-dimensional sphere, or in plane wave space-times, where its
shape is defined by a more general quadratic function.

In order to clarify what we mean by introducing limiting surfaces in
the embedding space-time we will now review a few examples from the
literature where this point is made explicit. In the examples that follow
(and throughout this article) we denote by $\textbf{ds}^2$ the induced
line element on the surface, and in the examples below we have embedded
each of the geometries trivially along the time-like direction $t$ of
the ambient space-time such that $t=\tau$, where $\tau$ is the time-like
coordinate on the surface, thus accounting for the term $-d\tau^2$ in 
the line element.

\begin{itemize}
\item The $\mathbb{R}^2$ plane in flat space-time \\\\ 
The simplest example of a minimal surface in Euclidean space is the $\mathbb{R}^2$ plane. We can trivially embed the two-dimensional spatial plane in Minkowski space-time as in \cite{Emparan:2009vd}, giving rise to the worldvolume geometry
\beq \label{ds:plane}
\textbf{ds}^2=-d\tau^2+d\rho^2+\rho^2d\phi^2
\eeq
(in spatial polar coordinates).
The configuration, embedded in this way, is still minimal. We can then set the plane to rotate with angular velocity $\Omega$ by considering the existence of a Killing vector field $\textbf{k}^{a}\partial_a=\partial_\tau+\Omega\partial_\phi$ whose norm is given by
\beq
\textbf{k}^2=1-\Omega^2\rho^2~~.
\eeq
The brane cannot rotate faster than the speed of light and hence we see that a limiting surface appears on the circular boundary defined by $\rho=\Omega^{-1}$. The existence of this boundary renders the $\mathbb{R}^2$ plane compact. In fact, this geometry describes the ultraspinning limit of the singly-spinning Myers-Perry black hole \cite{Emparan:2009vd}. We will revisit this configuration in Sec.~\ref{sec:helicoid} and furthermore construct a more non-trivial geometry based on minimal surface embeddings, namely the helicoid, which captures this disc geometry in an appropriate limit. In App.~\ref{sec:higherhelicoid} we construct higher-dimensional versions of helicoid geometries.

\item The $\mathbb{R}^2$ plane in de Sitter space-time \\\\
We now consider embedding a plane in de Sitter space-time as done in \cite{Armas:2010hz}.\footnote{This geometry, and higher-dimensional versions, played an important role in \cite{Emparan:2011ve} in the understanding of  horizon intersections and merger transitions.} The worldvolume geometry takes now the form
\beq \label{eq:dsr2de}
\textbf{ds}^2=-R_0^2d\tau^2+R_0^{-2}d\rho^2+\rho^2d\phi^2~~,~~R_0^2=1-\frac{\rho^2}{L^2}~~,
\eeq
where $\rho=L$ is the location of the cosmological horizon. The geometry is still minimal, even though embedded in de Sitter space-time as we will show in Sec.~\ref{sec:theorems}. In this case we do not need to set the plane to rotate and can instead consider a static geometry with Killing vector field $\textbf{k}^{a}\partial_a=\partial_\tau$ whose norm is 
\beq
\textbf{k}^2=1-\frac{\rho^2}{L^2}~~.
\eeq
Again we see that there is an inherent limiting surface at the cosmological horizon $\rho=L$ where the brane
must move at the speed of light. This introduces a circular boundary in the $\mathbb{R}^2$ plane and renders it
compact. This geometry describes the intersection of the event horizon of singly spinning Kerr-de Sitter black
holes with the cosmological horizon \cite{Armas:2010hz}. In Sec.~\ref{sec:discpw} we will show that such
compact $\mathbb{R}^2$ planes can also arise as parts of black hole horizons in plane wave space-times.
Analogously to what happens in asymptotically flat space-time, we also construct in \ref{sec:helicoidpw} two
different classes of helicoid geometries in plane wave space-times, which also reduce to disc geometries in an
appropriate limit. These results are generalised to higher-dimensional versions of helicoid geometries in
App~\ref{sec:higherhelicoid}. Furthermore, we construct other non-trivial examples using minimal embeddings such as rotating black catenoids and Scherk surfaces in Sec.~\ref{sec:catenoid} and higher-dimensional versions of rotating catenoids in App.~\ref{sec:highercatenoid}.
\end{itemize}

As can be seen from the second example presented above, embedding space-times with inherent limiting surfaces can be used to construct static geometries which were not possible in asymptotically flat space-time. With this in mind, and for the purpose of explanation, we will review and generalise two simple examples in de Sitter space-time present in the literature:
\begin{itemize}
\item Static black $p$-spheres \\\\ 
Families of static black $p$-spheres with radius $R$ in de Sitter space-time were constructed (for $p$ odd) in \cite{Armas:2010hz} where the worldvolume geometry is described by\footnote{The particular case of $p=1$, describing a black ring, was treated first in \cite{Caldarelli:2008pz}.}
\beq \label{eq:dspodd}
\textbf{ds}^2=-R_0^2d\tau^2+R^2d\Omega_{(p)}^2~~,~~R_0^2=1-\frac{R^2}{L^2}~~.
\eeq
The phenomenology of these geometries is slightly different from what we have encountered for minimal surfaces
and can be thought of as being analogous to soap bubbles (rather than soap films) instead. The tension of the brane tries to shrink the $p$-sphere but the gravitational potential of de Sitter space-time acts as an internal pressure. Equilibrium is attained when \cite{Armas:2010hz}
\beq \label{b:oddds}
\textbf{R}^2=\frac{p}{D-2}~~,~~\textbf{R}=\frac{R}{L}~~,
\eeq
where $D$ is the number of space-time dimensions. We observe in this paper that this result is actually
valid for all $p\ge1$ and not only for odd $p$, the reason being that since the $p$-sphere is not rotating, there is no obstruction to solving the equations of motion. In App.~\ref{sec:blackodd}, we present the analogous configurations in plane wave space-times.

\item Static black $p+2$-balls \\\\
The existence of a limiting surface in de Sitter space-times also allows for a higher-dimensional generalisation of the simplest minimal surface, namely, the $\mathbb{R}^2$ plane described in Eq.~\eqref{eq:dsr2de}. These geometries take the simple form of a $p+2$-ball
\beq\label{eq:dsbde}
\textbf{ds}^2=-R_0^2d\tau^2+R_0^{-2}d\rho^2+ \rho^2d\Omega_{(p+1)}^2~~,~~R_0^2=1-\frac{\rho^2}{L^2}~~,
\eeq
and are minimal surfaces in higher dimensions, as we will show in Sec.~\ref{sec:theorems}. In the case where
$p$ is even, they have also been considered in \cite{Armas:2010hz}, though presented in a different way, and
they describe the intersection of Kerr-de Sitter black holes with multiple ultraspins with the cosmological
horizon \cite{Armas:2010hz}. However, these configurations are valid for all $p\ge0$ and for the case in which
$p$ is odd they describe a new type of static black holes which is not 
connected to the family of Kerr-de Sitter black holes. In Sec.~\ref{sec:discpw} we will construct the analogous configurations in plane wave space-times while in Sec.~\ref{sec:desitter} we will consider de Sitter space-times with a black hole horizon. We will use these geometries in Sec.~\ref{sec:desitter} in order to show that one can construct compact black hole horizons in de Sitter space-times from minimal surfaces in the $(p+1)$-sphere of \eqref{eq:dsbde}, in particular we will construct black hole horizons using the Clifford torus and its higher-dimensional version as the starting point.
\end{itemize}

This paper is organised as follows. In Sec.~\ref{sec:blackfolds} we review
the necessary ingredients regarding minimal surfaces and the blackfold
approach required for the purposes of this paper. We also analyse in
detail the validity of the method based on a second order effective
action and improve previous analyses in the literature. In particular,
in Sec.~\ref{sec:validity} we identify the intrinsic and extrinsic
curvature invariants that classify each blackfold configuration up to
second order in a derivative expansion. The length scale associated
with each of these invariants is required to be much larger than the
thickness of the brane. We also introduce a condition which is necessary
for dealing with intersections of multiple worldvolume geometries. Subsequently,
in Secs.~\ref{sec:pw} and \ref{sec:cembed} we specify our conventions
for the different ambient space-times that we will consider. In
Secs.~\ref{sec:theorems} and \ref{sec:sclasses}, we make several observations
about the solutions to the blackfold equations and prove a number of statements
(theorems)
regarding minimal surfaces in the relevant embedding space-times. For
instance, we show that, amongst all minimal surfaces embedded in
$\mathbb{R}^{3}$, the plane and the helicoid are the only stationary
minimal embeddings which solve the blackfold equations in flat space-time.

In Sec.~\ref{sec:examples} we construct and study several new black
hole configurations in flat, plane wave and de Sitter space-times. In
particular, in Sec.~\ref{sec:helicoid} we construct a rotating black
helicoid in asymptotically flat space-time which turns out to have a
limit in which the ultraspinning regime of Myers-Perry black holes is
captured, hence showing that these two families of black hole solutions
are connected. In Secs.~\ref{sec:discpw} and \ref{sec:helicoidpw}
we construct the analogous configurations in plane wave space-times
and show that they have valid static limits due to the presence of
inherent limiting surfaces in the space-time. In Sec.~\ref{sec:catenoid}
we study more non-trivial examples of minimal embeddings in plane wave
space-times such as rotating catenoids and Scherk surfaces, which do not
have a flat space-time counterpart.  Finally in Sec.~\ref{sec:desitter}
we find black hole configurations using the Clifford torus and its
higher-dimensional counterpart.

In Sec.~\ref{sec:conclusions} we conclude with open problems and future
research directions. We also include several appendices. In App.~\ref{sec:valanal} we give specific details regarding the validity analysis of the configurations studied in this paper. In App. \ref{sec:higherhelicoid} and
App.~\ref{sec:highercatenoid} we study higher-dimensional versions of helicoids and catenoids. While the focus of this paper is on minimal surfaces
and their relevance for black hole horizons, in App.~\ref{sec:blackodd} we
construct and study several classes of stationary geometries with non-zero
constant mean extrinsic curvature that generalise \eqref{eq:dspodd}
to plane wave space-times.

\section{The blackfold approach and minimal surfaces} \label{sec:blackfolds}
In this section we first review some of the literature and required definitions for studying minimal surfaces,
with special focus on two-dimensional surfaces embedded in three-dimensional Euclidean space, which we will use
in subsequent parts of this paper. This is followed by a review of the necessary material for applying the
blackfold approach to the cases relevant in this paper, while improving the analysis of its regime of validity
based on a second order effective action. We then specify our conventions for the ambient space-times that
we will consider and introduce various classes of embeddings that we are interested in. This is followed by
several theorems for minimal surfaces in the relevant ambient space-times as well as by an overview of
different types of solutions of the blackfold equations that will appear in the following.
\subsection{Minimal surfaces}\label{sec:minimal}
Minimal surfaces is a vast and rich topic in the mathematics literature
with applications that range from soap films to polymer physics (see
e.g. \cite{Nitsche:2011} for a historical perspective of minimal
surfaces). These surfaces, within the mathematics literature, are
defined as the critical points of the induced area functional
\beq \label{eq:area}
\mathcal{A}[X^{\mu}(\sigma^{a})]=\int_{\mathcal{W}_{p+1}}d\sigma^{p+1}\sqrt{|\gamma|}~~.
\eeq
Here $\sigma^{a}$ are coordinates on the surface, 
$X^{\mu}(\sigma^{a})$ parametrises the $(p+1)$-dimensional surface $\mathcal{W}_{p+1}$ in the ambient 
(background, embedding) space(-time) with metric $g_{\mu\nu}$, and $\gamma$ is the determinant of the
induced metric 
\beq
\gamma_{ab}= \partial_a X^\mu \partial_b X^\nu g_{\mu\nu}
\eeq
(and see e.g.\ \cite{MeeksReview} for different (but equivalent) definitions and characterisations
of minimality).
We should note here that the mathematical terminology is somewhat (and
uncharacteristically) imprecise, as these surfaces are called minimal surfaces regardless of whether or
not the area is actually a minimum (and not some other extremum or critical point) of the area functional.
It might be more appropriate to refer to them as extremal surfaces, but we will follow the standard terminology
here.

Classically, most of the work done on minimal surfaces has been on two-dimensional surfaces embedded in
Euclidean three-dimensional space $\mathbb{R}^3$, equipped with the standard Euclidean metric with line
element
\beq \label{ds:r3}
d\mathbb{E}^{2}_{(3)}=dx_1^2+dx_2^2+dx_3^2~~,
\eeq
and hence are codimension one surfaces. Finding the critical points of \eqref{eq:area} is then equivalent 
to finding solutions to the minimal surface equation $K=0$, where $K$ is the mean extrinsic curvature of the
surface. In higher dimensions, and for surfaces of arbitrary codimension, the minimal surface equation 
takes the form
\beq \label{eq:min}
K^{i}=0~~,
\eeq
where the index $i$ labels the transverse directions to the surface embedding. The mean extrinsic curvature is defined as $K^{i}=\gamma^{ab}{K_{ab}}^{i}$ where ${K_{ab}}^{i}$ is the extrinsic curvature of the embedding given by
\beq \label{eq:ext}
{K_{ab}}^{i}={n^{i}}_{\mu}\partial_a {u^{\mu}}_{b}+{n^{i}}_{\rho}{\Gamma^{\rho}}_{\mu\nu}{u^{\mu}}_{a}{u^{\nu}}_{b}~~,
\eeq
where ${\Gamma^{\rho}}_{\mu\nu}$ is the Christoffel connection associated with the ambient metric $g_{\mu\nu}$ whereas ${n^{i}}_{\mu}$ projects orthogonally to $\mathcal{W}_{p+1}$ and satisfies the relations ${n^{i}}_{\mu}{u^{\mu}}_{a}=0$ and ${n^{i}}_{\mu}{n_{j}}^{\mu}=\delta^{i}_{j}$, where ${u^{\mu}}_{a}=\partial_a X^{\mu}$ projects along the surface $\mathcal{W}_{p+1}$. 

There is a vast number of two-dimensional minimal surfaces embedded into $\mathbb{R}^{3}$ (see e.g. \cite{MeeksReview} for an overview), which can have multiple genus and self-intersections but must always be non-compact. Two of the simplest examples of minimal surfaces in $\mathbb{R}^{3}$ are the so called \emph{ruled} surfaces consisting of the plane $\mathbb{R}^{2}$ and the helicoid. Both of these can be described by the embedding 
\beq \label{e:helicoid}
X^{1}(\rho,\theta)=\rho \cos(a\theta)~~,~~X^2(\rho,\theta)=\rho \sin(a\theta)~~,~~X^3(\rho,\theta)=\lambda \theta~~,
\eeq
into $\mathbb{R}^{3}$ with metric \eqref{ds:r3}, where $\lambda/a$ is the pitch of the helicoid. If we set
$\lambda=0$ then we recover the embedding of the plane $\mathbb{R}^{2}$. This example will play a significant role when we look at black hole horizon geometries in Sec.~\ref{sec:examples}.

The problem of finding minimal surfaces in $\mathbb{R}^{3}$ has been
partially solved since finding solutions to the complicated second
order differential equation \eqref{eq:min} has been reduced to
finding holomorphic functions of one complex variable by using
the Weierstrass-Enneper representation of minimal surfaces (see
e.g. \cite{MeeksReview}). In $D$-dimensional Euclidean spaces
$\mathbb{R}^{(D)}$ or other spaces such as $D$-dimensional spheres
$\mathbb{S}^{(D)}$, this tool is not generally available and hence
finding minimal surfaces is a more complex task. In $\mathbb{R}^{(D)}$
generalisations of certain minimal surfaces are available, such
as the planes, helicoids, catenoids \cite{Barbosa1981, Jorge:1984,
Hoppe:2013}, Enneper's surface \cite{Choe:1996} and Riemann minimal
surfaces \cite{Kaabachi_riemannminimal} but very few cases are
known. In $\mathbb{S}^{3}$,  the equatorial 2-sphere and the Clifford torus
constitute the simplest examples of minimal surfaces but also more
non-trivial examples such as Lawson surfaces have been constructed
(see \cite{Brendle2013} for a recent overview of the results). Minimal
surfaces in Lorentzian space-times $\mathbb{L}^{(D)}$ have also been
considered in the mathematics literature and some examples of minimal
surfaces are known (see e.g. \cite{Mira:2003, LeeS:2008, Wook:2011}
for a selection of minimal surfaces). However, as we will explain
in Sec.~\ref{sec:conclusions}, these are not of use for the purposes of
this work. 

\subsection*{Minimal surface equation in $\mathbb{R}^{3}$}
Two-dimensional minimal surfaces embedded in
$\mathbb{R}^{3}$ can be described in terms of what is known as a Monge parametrisation, which takes the form
\beq
X^{1}(u,v)=u~~,~~X^{2}(u,v)=v~~,~~X^{3}(u,v)=f(u,v)~~.
\eeq
Therefore, these minimal surfaces are described in terms of a single function $f(u,v)$ of two variables $u,v$. Evaluating explicitly the normal vector $n_\rho$ in $\mathbb{R}^{3}$ we find\footnote{We have omitted the transverse index $i$ from ${n^{i}}_{\rho}$ since the surface is of codimension one.}
\beq
n_\rho=\frac{1}{\sqrt{1+f_u^2+f_v^2}}\left(-f_u,-f_v,1\right)~~,
\eeq
where $f_u=\partial f/\partial u$ and $f_v=\partial f/\partial v$. Using this in \eqref{eq:min} one finds the minimal surface equation
\beq \label{eq:minm}
f_{uu}(1+f_{v}^2)+f_{vv}(1+f_u^2)-2f_u f_v f_{uv}=0~~,
\eeq
where $f_{uu}=\partial_u f_u$, $f_{vv}=\partial_v f_v$ and $f_{uv}=\partial_u f_v$. Eq.~\eqref{eq:minm} will
play an important role when we explore minimal surfaces solutions in non-trivial ambient space-times.

\subsection{The blackfold approach} \label{sec:validity}
The blackfold approach describes the effective dynamics of long-wavelength perturbations of black branes \cite{Emparan:2009cs, Emparan:2009at}. It consists of wrapping asymptotically flat black $p$-branes on an arbitrary $(p+1)$-dimensional submanifold $\mathcal{W}_{p+1}$ placed in a background space-time with metric $g_{\mu\nu}$. In this work we are interested in patching $\mathcal{W}_{p+1}$ with neutral vacuum black $p$-branes endowed with the metric \cite{Emparan:2009at}
\beq \label{ds:blackp}
ds^2_p=\left(\gamma_{ab}(\sigma^{c})+\frac{r_0^n(\sigma^{c})}{r^n}u_{a}(\sigma^{c})u_{b}(\sigma_{c})\right)d\sigma^{a}d\sigma^{b}+\frac{dr^2}{1-\frac{r_0^n(\sigma^{c})}{r^n}}+r^2d\Omega_{(n+1)}^2+\dots~.
\eeq
Here $\sigma^{c}$ denotes the set of coordinates on $\mathcal{W}_{p+1}$ and the worldvolume indices $a,b,c..$
run over $a=1,...,p+1$. The black brane metric \eqref{ds:blackp} is characterised by a set of fields
$\gamma_{ab},u^{a},r_0$ to leading order that vary slowly over $\mathcal{W}_{p+1}$ while higher-order
corrections - represented by the 'dots' - involve derivatives of $\gamma_{ab},u^{a},r_0$. As in our discussion
of minimal surfaces above, the worldvolume tensor $\gamma_{ab}=g_{\mu\nu}(X^{\alpha})\partial_a X^{\mu}\partial_b X^{\nu}$ is the induced metric on $\mathcal{W}_{p+1}$ and $X^{\alpha}(\sigma^{c})$ the set of mapping functions describing the location of the submanifold in the ambient space-time. The vector $u^{a}$ denotes the local boost velocity of the brane and is normalised such that $u^{a}u_{a}=-1$, while $r_0$ is the local brane thickness, i.e., the horizon size of the transverse $(n+2)$-dimensional part of the metric \eqref{ds:blackp}. We have chosen to parameterise the number of space-time dimensions such that $D=n+p+3$.

\subsubsection*{Effective dynamics}
In this work we are interested in stationary configurations, embedded in a stationary background with Killing vector field $k^{\mu}$, rotating with angular velocity $\Omega^a$ in each of the worldvolume rotational isometry directions $\phi_a$. These are characterised by a worldvolume Killing vector field of the form
\beq \label{eq:kvf}
\textbf{k}^{a}\partial_{a}=\partial_\tau+\Omega^{a}\partial_{\phi_a}~~,
\eeq
where $\tau$ labels the worldvolume time-like direction, and respective moduli on $\mathcal{W}_{p+1}$ given by $-|\partial_t|^2=-|\partial_\tau|^2=R_0^2$ and $|\partial_{\phi_a}|^2=R_a^2$. This worldvolume Killing vector field is required to map to the background Killing vector field $k^{\mu}$ and hence one must have that $k^{\mu}={u^{\mu}}_{a}\textbf{k}^{a}$ where ${u^{\mu}}_{a}=\partial_{a}X^{\mu}$ projects onto worldvolume directions.\footnote{This is required in order for the local thermodynamics of the brane \eqref{ds:blackp} to be well defined on the worldvolume, see \cite{Emparan:2009at} for a discussion of this point.} We note that we have assumed that the modulus of the time-like Killing vector field of the background space-time $\partial_t$ has the same norm as the worldvolume time-like Killing vector field $\partial_\tau$. This will be the case for all configurations presented in this paper, examples where this is not the case can be found in \cite{Armas:2012bk, Armas:2013ota}.

For stationary configurations, the effective dynamics of blackfolds, to second order in derivatives, is described in terms of a free energy functional of the form \cite{Armas:2013hsa, Armas:2013goa}\footnote{Here we have ignored backreaction corrections and also corrections due to spin in transverse directions to the worldvolume. See \cite{Armas:2011uf} for a discussion of backreaction corrections and \cite{Armas:2013hsa, Armas:2014rva} where spin corrections are included into the free energy.}
\beq\label{eq:free}
\begin{split}
\mathcal{F}[X^{i}]=-\int \sqrt{-\gamma}d^{p}\sigma\Big(P&+\upsilon_1\textbf{k}^{-1}\nabla_{a}\nabla^a\textbf{k}+\upsilon_2 \mathcal{R}+\upsilon_{3}u^{a}u^{b}\mathcal{R}_{ab} \\
&+\lambda_1 K^{i}K_{i}+\lambda_2 K^{abi}K_{abi}+\lambda_3 u^{a}u^{b}{K_{a}}^{ci}K_{bci}\Big)~~.
\end{split}
\eeq
Here we have introduced the indices $i,j,k...$ that run over $i=1,..,n+2$ to label orthogonal directions to $\mathcal{W}_{p+1}$. Furthermore, if the worldvolume Killing vector field is hypersurface orthogonal with respect to the spatial metric we can write $\sqrt{-\gamma}d\sigma^{p}=R_0dV_{(p)}$ where $R_0$ is the norm on $\mathcal{W}_{p+1}$ of the time-like Killing vector field $\partial_t$ associated with the time-like direction $t$ of the ambient space-time and $dV_{(p)}$ is the volume form on the spatial part of the worldvolume. 

The leading order term in \eqref{eq:free} is the pressure $P$ of the
effective fluid living on the worldvolume part of \eqref{ds:blackp}
and is responsible for the modified dynamics of blackfolds when
compared with the area-minimising action \eqref{eq:area} for minimal
surfaces. The other second order contributions are proportional to
the gradient of $\textbf{k}$, which measures variations of the local fluid temperature,\footnote{This term can be exchanged by a term proportional to the square of the fluid acceleration or the fluid vorticity $\omega_{ab}\omega^{ab}$, with the fluid velocity being given by $u^{a}$, if no boundaries are present \cite{Armas:2013hsa}. If the worldvolume has boundaries $\omega_{ab}\omega^{ab}$ may be independent. However, for all configurations we consider here $\omega_{ab}=0$ and hence we do not need to consider it.} the induced Ricci scalar $\mathcal{R}$
and the induced Ricci tensor $\mathcal{R}_{ab}$ on $\mathcal{W}_{p+1}$,
the mean extrinsic curvature $K^{i}=\gamma^{ab}{K_{ab}}^{i}$ and other
contractions with the extrinsic curvature tensor ${K_{ab}}^{i}$ defined
in \eqref{eq:ext}.
The set of scalars $P,\upsilon_i,\lambda_i$
depend only on the local temperature $\mathcal{T}$ of the brane
\eqref{ds:blackp}, which is related to the global temperature $T$
of the configuration via a local redshift $T=\textbf{k}\mathcal{T}$
where 
\beq
\textbf{k}=\sqrt{-\gamma_{ab}\textbf{k}^{a}\textbf{k}^{b}}
\eeq
is the modulus of the Killing vector field \eqref{eq:kvf}.

\subsubsection*{Equations of motion to leading order}
In order to scan for possible horizon topologies it is not necessary to consider more than the leading order term in \eqref{eq:free}, since all the other terms in \eqref{eq:free} are correction terms to the leading order dynamics. However, as we will see later in this section, they are necessary for understanding the regime of validity of this approach. Focusing on the leading order term, the equations of motion and boundary conditions that arise from varying \eqref{eq:free} take the form
\beq \label{eq:bfeom}
\nabla_{a}T^{ab}=0~~,~~T^{ab}{K_{ab}}^{i}=0~~,~~T^{ab}\eta_{b}|_{\partial\mathcal{W}_{p+1}}=0~~,
\eeq
where $\eta_{b}$ is a unit normalised normal vector to the worldvolume boundary $\partial\mathcal{W}_{p+1}$ and the effective stress-energy tensor $T^{ab}$ takes the perfect fluid form
\beq \label{eq:stress}
T^{ab}=P\gamma^{ab}-P'\textbf{k}u^{a}u^{b}~~,~~u^{a}=\frac{\textbf{k}^{a}}{\textbf{k}}~~,
\eeq
where the prime denotes a derivative with respect to $\textbf{k}$ and the fluid velocity $u^{a}$ is aligned
with the worldvolume Killing vector field $\textbf{k}^a$. Because of worldvolume general covariance,
the stress tensor \eqref{eq:stress} automatically solves the conservation equation in \eqref{eq:bfeom}. Therefore only the extrinsic equation $T^{ab}{K_{ab}}^{i}=0$ and the boundary condition are non-trivial. For the black branes \eqref{ds:blackp} the pressure $P$ takes the form \cite{Emparan:2009at}
\beq \label{eq:pressure}
P=-\frac{\Omega_{(n+1)}}{16\pi G}r_0^{n}~~,~~r_0=\frac{n}{4\pi T}\textbf{k}~~,
\eeq
where $\Omega_{(n+1)}$ is the volume of the unit $(n+1)$-sphere and $G$ is Newton's gravitational constant. In this case Eqs.~\eqref{eq:bfeom} reduce to
\beq \label{eq:bfeom1}
K^{i}=nu^{a}u^{b}{K_{ab}}^{i}~~,~~\textbf{k}|_{\partial\mathcal{W}_{p+1}}=0~~.
\eeq
These equations are known as the blackfold equations \cite{Emparan:2009cs, Emparan:2009at}. The first equation in \eqref{eq:bfeom1} exhibits the difference between blackfold dynamics and area-minimising actions \eqref{eq:area} due to the presence of the generically non-vanishing contraction $u^{a}u^{b}{K_{ab}}^{i}$. If the brane is rotating this contraction can be thought of as a repulsive centrifugal force, while if the brane is static but embedded in a space-time with a limiting surface it can be thought of as a force due to the non-trivial background gravitational potential. The second equation in \eqref{eq:bfeom1} expresses the fact that at the boundary the brane must be moving with the speed of light and hence the brane thickness $r_0$ must vanish there. This is the reason why limiting surfaces can provide a mechanism for making certain geometries compact. We note that while the first equation in \eqref{eq:bfeom1} has been shown to arise as a constraint equation when solving Einstein equations for a perturbed black brane metric \eqref{ds:blackp} \cite{Emparan:2007wm,Camps:2012hw}, blackfolds with boundaries were not considered in \cite{Emparan:2007wm,Camps:2012hw} and hence recovering the second equation in \eqref{eq:bfeom1} from gravity is still an open problem.

It was shown in \cite{Camps:2012hw} that, for worldvolumes without boundaries, for every solution of the blackfold equations \eqref{eq:bfeom1} there always exists a perturbed near-horizon metric \eqref{ds:blackp} which is regular on the horizon. The blackfold method for constructing the perturbed metric for a given worldvolume geometry relies on a matched asymptotic expansion. In this expansion, the metric in the far region $r\gg r_0$, obtained by solving the linearised Einstein equations in a given background space-time with a given source, serves as a boundary condition for the metric in the near-horizon region \eqref{ds:blackp}. It is unclear at present whether horizon regularity can be reconciled with arbitrary asymptotic boundary conditions. Ref.~\cite{LeWitt:2009qx} provides an example where horizon regularity and plane wave boundary conditions in the sense defined in \cite{LeWitt:2008zx} could not be simultaneously fulfilled. We will briefly come back to this issue in Sec.~\ref{sec:conclusions}, and comment on in which way our setting differs from that considered in \cite{LeWitt:2009qx}.

\subsubsection*{Thermodynamics}
The thermodynamic properties of these configurations can be obtained directly from the free energy functional \eqref{eq:free} by noting that the free energy satisfies the relation
\beq\label{eq:freer}
\mathcal{F}=M-T S-\Omega^{a}J_{a}~~,
\eeq
where $M$ is the total mass, $S$ is the entropy and $J_a$ is the angular momentum associated with each rotational isometry direction $\phi_a$. Given the free energy $\mathcal{F}$ we can obtain the entropy and angular momenta simply by \cite{Armas:2014rva}
\beq \label{eq:thermo}
S=-\frac{\partial\mathcal{F}}{\partial T}~~,~~J_a=-\frac{\partial\mathcal{F}}{\partial \Omega^a}~~,
\eeq
and hence the mass via \eqref{eq:freer}. It is easy to show that the branes \eqref{ds:blackp} satisfy the relation $\mathcal{F}=TS/n$ and hence throughout this paper we avoid presenting expressions for $S$ as we always present the free energy $\mathcal{F}$ for all configurations. These configurations satisfy a Smarr-type relation of the form
\beq\label{eq:smarr}
(n+p)M-(n+p+1)\left(TS+\Omega^{a}J_a\right)=\boldsymbol{\mathcal{T}}~~,
\eeq
where $\boldsymbol{\mathcal{T}}$ is the total tension to leading order defined as
\beq \label{eq:tension}
\boldsymbol{\mathcal{T}}=-\int dV_{(p)}R_0\left(\gamma^{ab}+\xi^{a}\xi^{b}\right)T_{ab}~~,
\eeq
and $\xi^{a}\partial_a=\partial_\tau$ is the worldvolume Killing vector field associated with time translations of the worldvolume. Here we have assumed that $\partial_\tau$ is hypersurface orthogonal with respect to the spatial worldvolume metric. For configurations in asymptotically flat space-time (without non-compact directions) the total tension $\boldsymbol{\mathcal{T}}$ vanishes and we recover the usual Smarr relation for asymptotically flat black holes. 

\subsubsection{Regime of validity} \label{sec:rvalidity}
As mentioned at the beginning of this section, the blackfold approach is a perturbative expansion in the fields $\gamma_{ab},u^{a},r_0$ and as such its regime of validity is also defined at each step in the perturbative expansion. Generically, one must require that at each order in the expansion, the length scales associated to each of the geometrical invariants describing the intrinsic and extrinsic geometry of the blackfold to next order must be large when compared to the local brane thickness $r_0$. To be precise, it can be shown by dimensional analysis that the transport coefficients $\upsilon_i,\lambda_i$ scale as $r_0^{n+2}$,\footnote{Since we know that the transport coefficients $\lambda_i$ scale as $r_0^{n+2}$ \cite{Armas:2013hsa} the remaining scalings can be obtained using Gauss-Codazzi equations and similar relations relating fluid data given in \cite{Armas:2013hsa}. Alternatively, one may use the thermodynamic identities found in \cite{Armas:2014rva}.} thus by \eqref{eq:pressure} as $\textbf{k}^{n+2}$, and therefore by looking at the second order free energy functional \eqref{eq:free} one must require that to leading order
\beq \label{eq:inv}
\!\!\!\!r_0\ll\left(|\frac{\nabla_{a}\nabla^a\textbf{k}}{\textbf{k}}|^{-\frac{1}{2}},|\mathcal{R}|^{-\frac{1}{2}},|u^{a}u^{b}\mathcal{R}_{ab}|^{-\frac{1}{2}},|K^{i}K_{i}|^{-\frac{1}{2}},|K^{abi}K_{abi}|^{-\frac{1}{2}},|u^{a}u^{b}{K_{a}}^{ci}K_{bci}|^{-\frac{1}{2}}\right).
\eeq
This ensures that locally on $\mathcal{W}_{p+1}$ the geometry can be seen as an asymptotically flat brane \eqref{ds:blackp}. The geometric invariants presented in \eqref{eq:inv} correspond to a particular choice \cite{Armas:2013hsa} as these are related to the background Riemann curvature via the Gauss-Codazzi equation
\beq
R_{abcd}=\mathcal{R}_{abcd}-{K_{ac}}^{i}{K_{bdi}}+{K_{bc}}^{i}{K_{adi}}~~,
\eeq
where $R_{abcd}$ is the projection of the Riemann curvature tensor of the ambient space-time along worldvolume directions. Contracting this equation with combinations of $\gamma^{ab}$ and $u^{a}$ one finds the two equations
\beq \label{eq:gc2}
\begin{split}
R_{||}&=\mathcal{R}-K^{i}K_{i}+{K^{ab}}_{i}{K_{ab}}^{i}\\
R_{//}&=u^{a}u^{c}\mathcal{R}_{ac}-u^{a}u^{c}{K_{ac}}^{i}K_{i}+u^{a}u^{c}{K_{ab}}^{i}{K^{b}}_{ci}~~,
\end{split}
\eeq
where we have defined $R_{||}=\gamma^{ac}\gamma^{bd}R_{abcd}$ and $R_{//}=u^{a}u^{c}\gamma^{bd}R_{abcd}$. Using this, we will recast the free energy functional \eqref{eq:free} in a way which will be more suitable for the study of minimal surface embeddings. We note that the second term on the r.h.s. of the second equation in \eqref{eq:gc2} can be exchanged, to second order, by a term proportional to $K^{i}K_{i}$ using the equations of motion \eqref{eq:bfeom1} as explained in \cite{Armas:2013hsa, Armas:2013goa}. Therefore we can write the free energy functional \eqref{eq:free} as
\beq\label{eq:free1}
\begin{split}
\mathcal{F}[X^{i}]=-\int \sqrt{-\gamma}d\sigma^{p}\Big(P&+\upsilon_1\textbf{k}^{-1}\nabla_{a}\nabla^a\textbf{k}+(\upsilon_2-\lambda_2) \mathcal{R}+(\upsilon_{3}-\lambda_3)u^{a}u^{b}\mathcal{R}_{ab} \\
&+(\lambda_1+\lambda_2+\frac{\lambda_3}{n}) K^{i}K_{i}+\lambda_2 R_{||}+\lambda_3 R_{//}\Big)~~.
\end{split}
\eeq
The validity conditions to leading order can then be recast as
\beq \label{eq:req}
r_0\ll\left(|\frac{\nabla_{a}\nabla^a\textbf{k}}{\textbf{k}}|^{-\frac{1}{2}},|\mathcal{R}|^{-\frac{1}{2}},|u^{a}u^{b}\mathcal{R}_{ab}|^{-\frac{1}{2}},|K^{i}K_{i}|^{-\frac{1}{2}},|R_{||}|^{-\frac{1}{2}},|R_{//}|^{-\frac{1}{2}}\right).
\eeq
In order to show the usefulness of these manipulations, we apply this to the case of (Anti)-de Sitter space-time. Using the fact that it is maximally symmetric $R_{\mu\nu\lambda\rho}=L^{-2}(g_{\mu\lambda}g_{\nu\rho}-g_{\mu\rho}g_{\nu\lambda})$ we compute the background curvature invariants
\beq
|R_{||}|^{-\frac{1}{2}}=\frac{L}{\sqrt{p(p+1)}}~~,~~|R_{//}|^{-\frac{1}{2}}=\frac{L}{\sqrt{p}}~~,
\eeq
where $L$ is the (Anti)-de Sitter radius. Therefore one obtains the requirement
\beq
r_0\ll L~~,
\eeq
which justifies the arguments used in \cite{Caldarelli:2008pz, Armas:2010hz}.

If we focus on minimal surfaces, which by definition satisfy condition
\eqref{eq:min} then of the six invariants involved in \eqref{eq:req}
only five are non-trivial. In this case the perturbative expansion
\eqref{eq:free1}, to second order in derivatives, can be seen as a purely
hydrodynamic expansion in a curved background.

\subsubsection*{Blackfolds with boundaries}
Most configurations analysed in this paper have boundaries, which as mentioned above, are described by the condition $\textbf{k}=0$. The effective free energy \eqref{eq:free} is given by a derivative expansion, and
is a priori unrelated to effects due to the presence of boundaries. In particular, as a long-wavelength effective theory, the blackfold approach will not be able to probe distances below a certain scale that we denote by $\ell$. If $\rho_+$ is the location of the boundary and $\epsilon$ the distance away from it, then one must require 
\beq
\rho_+-\epsilon\gg\ell~~,
\eeq
for the approximation to be valid. In fact, the existence of this break down of the approximation can be seen directly from the requirement \eqref{eq:req} associated with the invariant $|\textbf{k}^{-1}\nabla_{a}\nabla^a\textbf{k}|^{-1/2}$. In general we have that
\beq \label{eq:T}
|\textbf{k}^{-1}\nabla_{a}\nabla^a\textbf{k}|\propto \textbf{k}^{-4}~~,
\eeq
and therefore the requirement \eqref{eq:req} reduces to
\beq
r_+\ll \textbf{k}~~,~~r_+=\frac{n}{4\pi T}~~.
\eeq
As $\textbf{k}$ approaches $0$ at the boundary, it is not possible to satisfy this condition, signalling a possible break down of the expansion.

The effective description of blackfolds is given in terms of a hydrodynamic and elastic expansion, however, when boundaries are present one should also consider a \emph{boundary expansion} in powers of $\epsilon$, in which case the description can become increasingly better with the addition of higher-order corrections. We note that what is considered leading order terms or higher-order corrections in the effective free energy \eqref{eq:free} in a derivative expansion (either hydrodynamic or elastic) is not necessarily the same from the point of view of a boundary expansion. In fact, by looking at \eqref{eq:T} we see that the correction term in \eqref{eq:free}, from the point of view of a derivative expansion, associated with $|\textbf{k}^{-1}\nabla_{a}\nabla^a\textbf{k}|$ scales as $\textbf{k}^{n-2}$ and hence, as one approaches the boundary $\textbf{k}=0$, this term is not sub-leading when compared with $P\propto\textbf{k}^{n}$. 

Presumably, though not necessarily, for a given black hole, the blackfold description may be the correct one in a patch of the geometry while another patch may not be locally described by a metric of the form \eqref{ds:blackp}. Examples of situations where this behaviour may be the case are found in the context of BIon solutions \cite{Grignani:2010xm, Grignani:2011mr} and M2-M5 intersections \cite{Niarchos:2012pn, Niarchos:2012cy, Niarchos:2013ia}. If the geometry has boundaries, this could potentially signify that the geometry near the boundary would have to be replaced by something else than \eqref{ds:blackp} but which would smoothly connect to \eqref{ds:blackp}. Alternatively, one can demand the existence of a smooth limit of the blackfold description near the boundary under the assumption that, even though the approximation is expected to break down, the existence of a smooth limit when $r_0\to0$ yields the correct gravitational description. This, as we will review in Sec.~\ref{sec:blackdisc}, is exactly what happens for ultraspinning Myers-Perry black holes and can be seen by analysing the exact analytic metric as in \cite{Armas:2011uf}. This illustrates that in certain circumstances the blackfold approach appears to work better than one a priori has the right to expect.

While a deeper understanding of these issues is of interest, this is beyond the scope of this paper. Instead, and in the absence of exact analytic solutions, we will construct several blackfold geometries with boundaries assuming that a well defined boundary expansion exists and show, in Sec.~\ref{sec:blackdisc}, that their thermodynamic properties can be obtained exactly, to leading order in $\epsilon$, regardless of what the correct boundary description might be.

\subsubsection*{Multiple blackfolds and self-intersections}
It is important to mention that the second order corrected free energy \eqref{eq:free} has not taken into account corrections due to gravitational backreaction or gravitational self-force. In particular, when one is considering a configuration of multiple
worldvolumes, then the blackfold approximation is expected to break down when the distance $d$ between two worldvolumes becomes of the order of $r_0$. One therefore also needs to require \cite{Emparan:2009vd} 
\begin{equation} \label{r:d}
r_0 \ll d\;\;.
\end{equation}
A fortiori this means that intersecting (or self-intersecting - see Sec.~\ref{sec:catenoid} for an example)
worldvolume configurations lie outside the regime of validity of the blackfold approximation, but one might
expect gravitational backreaction to regularise or smooth out such intersections, much as in the case of
backreacted intersecting brane geometries in string theory, and it would certainly be of interest to investigate this further.


\subsection{Plane waves}  \label{sec:pw}

Among the background space-times that we will consider in this work are plane waves, 
and here we briefly summarise the properties of plane waves that we will make use of
later on.

Plane wave space-times have metrics of the form
\beq
ds^2 = 2 dudv -2 A(u,x^q) du^2 + d\mathbb{E}^{2}_{(D-2)}(x^{q})~~,
\eeq
where,
\beq
d\mathbb{E}^{2}_{(D-2)}(x^{q}) = \sum_{q=1}^{D-2}(dx^q)^2 
\eeq
is the $(D-2)$-dimensional Euclidean metric describing the planar wave front of the gravitational 
wave, and the function $A(u,x^{q})$ describing the wave profile is a quadratic function 
\beq
A(u,x^{q})=A_{qr}(u)x^{q}x^{r}
\eeq
of the transverse coordinates, $q,r=1,...,D-2$. This quadratic function encodes all 
the non-vanishing components of the Riemann tensor, namely 
\beq 
R_{uqur} = 2 A_{qr}(u)
\eeq 
(the somewhat unconventional prefactor of 2 here and in 
the metric serves the purpose of avoiding a proliferation of factors of 2 later on, when 
using  standard time and space coordinates $(t,y)$ instead of the null(ish) coordinates
$(u,v)$). This implies that the only non-vanishing component of the Ricci tensor is
\beq
R_{uu} = 2 \text{Tr}(A_{qr})~~,
\eeq
and that the Ricci scalar is zero, 
\beq
R=0\;\;.
\eeq
In particular, therefore, solutions of the non-linear vacuum Einstein equations correspond to 
transverse traceless matrices $A_{qr}(u)$ (``gravitons''). 

For the blackfold approach we are 
interested in stationary background space-times, and therefore we will focus on time-independent
plane waves, with a $u$-independent profile 
\beq
A(x^{q})=A_{qr}x^{q}x^{r}\;\;.
\eeq
Even though this is not manifest in these coordinates, these space-times are homogeneous
(even symmetric) and, in particular, the origin $x^q=0$ of the transverese coordinates is 
not in any way a special locus in space-time (only in these coordinates).
By introducing the coordinates
\beq \label{eq:uv}
u = (y+t)/\sqrt{2}\quad,\quad v = (y-t)/\sqrt{2}\;\;,
\eeq
these metrics then take the standard stationary (but not static) form
\beq
\label{ds:pw0}
ds^2=-(1+A(x^{q}))dt^2+(1-A(x^{q}))dy^2-2A(x^{q})dtdy+d\mathbb{E}^{2}_{(D-2)}(x^{q})~~,
\eeq
with Killing vector $\partial_t$. In these coordinates, the components of 
the Riemann and Ricci tensors are
\beq
R_{\mu q \nu r} = A_{qr} \quad,\quad R_{\mu\nu} = \text{Tr}A_{qr} ~~,
\eeq
for $\mu,\nu \in \{t,y\}$. 

In this time-independent case, constant $SO(D-2)$ transformations of the transverse coordinates can be
used to diagonalise the constant symmetric matrix $A_{qr}$, 
\beq
A_{qr} = A_q \delta_{qr} \;\;.
\eeq
Moreover, by a boost in the $(t,y)$ (or $(u,v)$) plane, the eigenvalues can be rescaled by an 
overall positive factor, 
\beq
(u,v) \rightarrow (\lambda u,\lambda^{-1}v) \quad\Rightarrow\quad A_q \rightarrow \lambda^2 A_q\;\;.
\eeq
Thus a priori only ratios of eigenvalues of $A_{qr}$ have an invariant physical meaning. However, 
via the embedding of $p$-branes into the plane wave background, in particular via the identification
$t=\tau$ of the worldvolume and background time coordinates, this boost invariance is broken and 
the magnitudes of the individual eigenvalues have physical significance. Moreover, such an embedding
will reduce the transverse $SO(D-2)$-invariance, and thus in principle off-diagonal matrix elements
could be present. However, in none of the numerous examples that we have investigated did such 
non-diagonal elements turn out to be particularly useful (let alone necessary). For that reason, and
in order not to unduly burden the notation, we will concentrate on diagonal wave profiles in the following
(and only add a comment here and there on off-diagonal contributions).

\subsection{Classes of embedding space-times and classes of embedded geometries} \label{sec:cembed}
In this paper we consider three different classes of $D$-dimensional Lorentzian 
embedding space-times $\mathbb{L}^{(D)}$ 
into which we will embed different classes of geometries. Some of these space-times have inherent limiting surfaces and hence provide an interesting playground for constructing compact minimal surfaces. These are:
\begin{itemize}
\item \textbf{Flat space-time:} we write down the metric of flat space-time in the form
\beq \label{ds:flat}
ds^2=-dt^2+d\mathbb{E}^{2}_{(D-1)}(x^{q})~~,
\eeq
where $d\mathbb{E}^{2}_{(D-1)}(x^{q})$ is the metric on the $(D-1)$-dimensional Euclidean space $\mathbb{E}^{(D-1)}$ parametrised in terms of the coordinates $x^{q}$ where the index $q$ runs over $q=1,...,D-1$. We use the indices $q,r,t,s$ to label space-time directions in $\mathbb{E}^{(D-1)}$. For this class of ambient space-times the background curvature invariants $R_{||},R_{//}$ in \eqref{eq:req} vanish.

\item \textbf{Plane wave space-times:} as discussed above, 
we consider time-independent plane wave space-times equipped with the metric
\eqref{ds:pw0}
\beq \label{ds:pw}
ds^2=-(1+A(x^{q}))dt^2+(1-A(x^{q}))dy^2-2A(x^{q})dtdy+d\mathbb{E}^{2}_{(D-2)}(x^{q})~~,
\eeq
and with 
\beq
A_{qr}=A_q\delta_{qr}\;\;.
\eeq
If at least one of the eigenvalues is negative, then the plane wave space-time will have a limiting surface 
where the time-like Killing vector field $\partial_t$ becomes null, i.e., where $(1+A(x^{q}))=0$.

We focus on the class of plane waves which are solutions of the
vacuum Einstein equations, therefore we impose $\text{Tr}A_{qr} =0$.
For these space-times the background curvature invariants $R_{||},R_{//}$
in \eqref{eq:req} depend on the precise form of the embedding. For the
two types of embeddings that we consider below these invariants either
vanish or are given in terms of linear combinations of the eigenvalues $A_q$.

\item \textbf{de Sitter space-times:} we consider de Sitter space-times in the presence of a black hole, where the metric is written as
\beq \label{ds:ds1}
ds^2=-f(r)dt^2+f(r)^{-1}dr^2+r^2d\Omega_{(D-2)}^2~~,~~f(r)=1-\frac{r_m^{D-3}}{r^{D-3}}-\frac{r^2}{L^2}~~.
\eeq
When $r_m=0$ the black hole horizon is no longer there and we recover pure de Sitter with radius $L$. The range of the coordinate $r$ lies in between the two real and positive roots of $f(r)=0$. Therefore, this class of space-times inherits two limiting surfaces located at the black hole horizon and at the cosmological horizon where the time-like Killing vector field $\partial_t$ becomes null. When $r_m=0$ the only limiting surface is located at the cosmological horizon where $r=L$.

It will sometimes be useful to introduce spatially conformally flat coordinates by defining a new coordinate $\tilde r$ such that $r^2=h(\tilde r)^{-1}\tilde r^2$ and $f(r)^{-1}dr^2=h(\tilde r)^{-1}d\tilde r^2$. The metric \eqref{ds:ds1} then takes the form
\beq \label{ds:ds2}
ds^2=-f(\tilde r)dt^2+h(\tilde r)^{-1}d\mathbb{E}^{2}_{(D-1)}(x^{q})~~,~~\tilde r^2=\sum_{q=1}^{D-1}x_{q}^{2}~~,
\eeq
where the index $q$ runs over $q=1,...,D-1$. We will also write the spatial part of the metric \eqref{ds:ds2}, as the metric $d\tilde{\mathbb{E}}_{(D-1)}^2(x^q)=h(\tilde r)^{-1}d\mathbb{E}^{2}_{(D-1)}(x^{q})$, on the conformally Euclidean space ${\tilde {\mathbb{E}}}^{(D-1)}$.

\end{itemize}

In these space-times we embed three classes of worldvolume geometries which are either minimal or are constructed using minimal surfaces in $\mathbb{E}^{(D-1)}$ (in flat and de Sitter space-times) or in $\mathbb{E}^{(D-2)}$ (in plane wave space-times) as the starting point. In App.~\ref{sec:blackodd} we focus on a class of worldvolume geometries with constant mean curvature related to the example presented in \eqref{eq:dspodd}. These classes of embeddings, which may be static or stationary with Killing vector field \eqref{eq:kvf}, have boundaries when $\textbf{k}=0$ and are of the following form:
\begin{itemize}
\item \textbf{Type I:} this class of $(p+1)$-dimensional worldvolume geometries have induced metric
\beq \label{ds:type1}
\textbf{ds}^2=-R_0^2(X^q_M)d\tau^2+d\tilde{\mathbb{E}}_{(p)}^2(X^q_M)~~,
\eeq
where $\textbf{ds}^2=\gamma_{ab}d\sigma^{a}d\sigma^{b}$ is the induced volume element while $d\tilde{\mathbb{E}}_{(p)}^2(X^q_M)$ is the induced $p$-dimensional spatial metric obtained by restricting the metric $d\tilde{\mathbb{E}}^{2}_{(D-1)}$ (in the case of flat or de Sitter space-times), or $d{\mathbb{E}}^{2}_{(D-2)}$ (in the case of plane wave space-times),  to the minimal embedding $x^{q}=X^q_M$ with respect to ${{\mathbb{E}}}^{(D-1)}$ or $\mathbb{E}^{(D-2)}$. These geometries can be obtained by choosing the embedding coordinates $(t,x^{q})=(\tau,X^{q}_M)$ in the space-times \eqref{ds:flat} and \eqref{ds:ds2} or by choosing $(t,y,x^{q})=(\tau,0,X^{q}_M)$ in the space-times \eqref{ds:pw}. Furthermore, here and in the next two types of embeddings the mapping functions $X^{q}_M$ do not depend on $\tau$. If embedded into plane wave space-times then explicit evaluation of the invariants $R_{||}$ and $R_{//}$ yields
\beq \label{r:pw}
R_{||}=\frac{R_0}{\sqrt{2}}|\gamma^{qr}A_{qr}|^{-\frac{1}{2}}~~,~~R_{//}=R_0|u^{q}u^{r}A_{qr}|^{-\frac{1}{2}}~~.
\eeq
For the purpose of analysing the regime of validity of these geometries it is useful to compute the induced Ricci tensor of the class of metrics \eqref{ds:type1}. This is given by
\beq \label{r:type1}
\mathcal{R}=\mathcal{R}_{\tilde{\mathbb{E}}}-2\frac{\Delta_{\tilde{\mathbb{E}}} R_0}{R_0}~~,
\eeq
where $\mathcal{R}_{\tilde{\mathbb{E}}}$ is the Ricci scalar of the spatial $p$-dimensional metric and $\Delta_{\tilde{\mathbb{E}}}$ is the Laplace operator on that $p$-dimensional space.

\item \textbf{Type II:} this class of $(p+1)$-dimensional worldvolume geometries have induced metric
\beq \label{ds:type2}
\textbf{ds}^2=-R_0^2(X^q_M)d\tau^2\!+\!2(1-R_0^2(X^q_M))d\tau dz\!+\!(2-R_0^2(X^q_M))dz^2\!+\!d\mathbb{E}_{(p-1)}^2(X^q_M)
\eeq
and describe a wave with non-planar wave front whose geometry is described by the induced 
$(p-1)$-dimensional metric $d\mathbb{E}_{(p-1)}^2(X^q_M)$. 
This class of embeddings is obtained only in plane wave space-times \eqref{ds:pw} by choosing the embedding coordinates $(t,y,x^{q})=(\tau,z,X^{q}_M)$ and are non-compact along the $z$-direction. In this case, explicit computation of the invariants $R_{||},R_{//}$ leads to $R_{||}=R_{//}=0$. The induced Ricci scalar for the metrics \eqref{ds:type2} is simply
\beq \label{r:type2}
\mathcal{R}=\mathcal{R}_{\mathbb{E}}~~,
\eeq
where $\mathcal{R}_{\mathbb{E}}$ is the Ricci scalar of the $(p-1)$-dimensional spatial metric $d\mathbb{E}_{(p-1)}^2(X^q_M)$.

\item \textbf{Type III:} this last class of $(p+1)$-dimensional worldvolume geometries have induced metric
\beq \label{ds:type3}
\textbf{ds}^2=-R_0^2(\rho,X^q_M)d\tau^2+H_0^{-2}(\rho,X^q_M)d\rho^2+\rho^2d\Omega_{(p-1)}^2(X^q_M)~~,
\eeq
where the minimal embedding $X^q_M$ is defined on the unit $(p-1)$-sphere $\mathbb{S}^{(p-1)}$. This class of embeddings can be obtained by choosing the mapping functions $(t,r,x^{q})=(\tau,\rho,X^{q}_M)$ in de Sitter space-times \eqref{ds:ds1} where $H_0=R_0$ but it can also be obtained in flat space-time \eqref{ds:flat}, by writing the metric on $\mathbb{E}^{(D-1)}$ as $d\mathbb{E}^{2}_{(D-1)}(x^{q})=dr^2+r^2d\Omega^{2}_{(D-2)}(x^{q})$, where $R_0=H_0=1$, and choosing $(t,r,x^{q})=(\tau,\rho,X^{q}_M)$ or similarly, in plane wave space-times \eqref{ds:pw}, where $H_0=1$.

\end{itemize}

\subsubsection*{Limiting surfaces and validity of embedded geometries}
The classes of embedded geometries presented above may be static or stationary and characterised by a worldvolume Killing vector field of the form \eqref{eq:kvf}. Evaluating explicitly the modulus $\textbf{k}$ we find
\beq
\textbf{k}^2=R_0^2-(\Omega^{a})^2R_{a}^2~~.
\eeq
Therefore, generically these space-times, where the surfaces are embedded, are characterised by limiting surfaces described by the equation $R_0^2-(\Omega^{a})^2R_{a}^2=0$. If the worldvolume geometry is static $\Omega^{a}=0$ (hence $\textbf{k}^2=R_0^2$) then limiting surfaces can also be present as long as $R_0^2=0$ at some point on $\mathcal{W}_{p+1}$.

It is important to study the regime of validity of the blackfold approach \eqref{eq:req} for the geometries we consider and, in particular,
the behaviour near the boundary $\textbf{k}=0$. In general, the six invariants
presented in \eqref{eq:req} must be evaluated explicitly for each
configuration. However, certain universal features exist. First of all, 
since all the
geometries we consider turn out to be minimal (not only in $\mathbb{E}^{(D-1)}$ or
$\mathbb{E}^{(D-2)}$ but also in Lorentzian space-time $\mathbb{L}^{(D)}$),
and hence satisfy \eqref{eq:min},
of the six invariants in \eqref{eq:req} only five need to be evaluated.  
Secondly, note that from
\eqref{eq:pressure} since $r_0\propto \textbf{k}$, then according to
\eqref{eq:req} none of the five relevant scalars divided by $\textbf{k}$
should vanish over the geometry, or in other other words, the intrinsic or
extrinsic curvature scales should not diverge faster than $\textbf{k}^{-1}$
over $\mathcal{W}_{p+1}$. This leads us to the following conclusions:

\begin{itemize}

\item Since all the embeddings presented above have boundaries then the analysis around \eqref{eq:T} holds. The fact that the invariant $|\textbf{k}^{-1}\nabla_a\nabla^{a}\textbf{k}|$ diverges too quickly as $\textbf{k}\to0$ signals a break down of the approximation near the boundary. For that reason we consider these blackfold configurations valid up to a distance $\epsilon$ from the boundary.

\item For static embeddings of \textbf{Type I}, which have
the plane wave space-time \eqref{ds:pw} as the ambient
space-time, using \eqref{r:type1}, one has that
$|\mathcal{R}|^{-\frac{1}{2}}\textbf{k}^{-1}\propto\textbf{k}^2$.
Therefore, since this invariant vanishes too quickly near the boundary then
the requirements \eqref{eq:req} cannot be satisfied. If the ambient
space-time was flat or de Sitter space-time this would not
constitute a problem. In particular, in the latter case, we find that
$|\mathcal{R}|^{-\frac{1}{2}}\propto\textbf{k}$. By contrast, embeddings
of \textbf{Type II}, according to \eqref{r:type2}, do not suffer from
a divergence at the boundary and it is only required that $\textbf{k}^2
\mathcal{R}$ does not diverge anywhere over $\mathcal{W}_{p+1}$.

\item If $X^{q}_M$ parametrises a $(p-1)$-dimensional sphere
in embeddings of \textbf{Type III}, then such embeddings lie
within the regime of validity. However, if $X^{q}_M$ does not
parametrize a $(p-1)$-dimensional sphere then the spatial metric
$H_0^{-2}(\rho,X^q_M)d\rho^2+\rho^2d\Omega_{(p-1)}^2(X^q_M)$ will
suffer from a conical singularity at $\rho=0$ in the case of flat and
plane wave space-times ($H_0=1$) and hence $\mathcal{R}\to\infty$
as $\rho\to0$. However, potential singularities at $r=0$ can be
shielded behind a black hole horizon. Therefore we need $r_m\ne0$
in \eqref{ds:ds2}. This screening effect is also present if we send
$L\to\infty$, i.e.\  for the asymptotically flat Schwarzschild-Tangherlini black hole.

\end{itemize}

\subsection{Theorems for minimal surfaces} \label{sec:theorems}
In this section we prove several results for minimal surfaces embedded into ambient space-times in the manner
described in the previous section. We further analyse what conditions these geometries need to satisfy in order
to solve the blackfold equations \eqref{eq:bfeom}, \eqref{eq:bfeom1}. 
To set the stage, note that in
general, in order to satisfy the blackfold equation
$K^{i}=nu^{a}u^{b}{K_{ab}}^{i}$, it is neither necessary nor sufficient
for the embedding to define a minimal surface ($K^i=0$) in the Lorentzian embedding
space $\mathbb{L}^{(D)}$, and our discussion below will reflect this dichotomy.
Nevertheless, all the explicit geometries that we will construct later on will
involve minimal Lorentzian surfaces.


\subsubsection*{Embeddings of \textbf{Type I}}
We begin by reviewing a result of \cite{Emparan:2009vd}, namely,
\begin{theorem} \label{theo:bfstatic}
(from \cite{Emparan:2009vd}) If the embedding is static and of \textbf{Type I}, embedded into flat space-time, then any minimal surface in $\mathbb{E}^{(D-1)}$ is a minimal surface in $\mathbb{L}^{(D)}$ and, furthermore, it solves the blackfold equations \eqref{eq:bfeom}.
\end{theorem}
\begin{proof}
We label the spatial indices of the worldvolume by $\hat{a},\hat{b},..=1,...,p$. If the surface $X^{q}_M$ is
minimal in $\mathbb{E}^{(D-1)}$ then we have that
$\tilde{K}^{i}=\gamma^{\hat{a}\hat{b}}{K_{\hat{a}\hat{b}}}^{i}=0$. Since the embedding is of \textbf{Type I}, the mapping functions $X^{q}_M$ do not depend on $\tau$ and since the ambient space-time is flat we can always choose coordinates such that ${\Gamma^{\rho}}_{\mu\nu}=0$. Therefore, from \eqref{eq:ext} it follows that ${K_{\tau a}}^{i}=0$ and we obtain $K^{i}=\gamma^{ab}{K_{ab}}^{i}=0$. Hence, the surface is minimal in $\mathbb{L}^{(D)}$. Since ${K_{\tau a}}^{i}=0$ and the embedding is static we have that $u^{a}{K_{ab}}^{i}=0$. Therefore both sides of equation \eqref{eq:bfeom} are satisfied. 
\end{proof}
From this it follows that minimal surfaces in $\mathbb{E}^{(D-1)}$, which satisfy the validity requirements
\eqref{eq:req} and \eqref{r:d}, provide geometries for non-compact black hole horizons, because
flat space-time with embedded static geometries has no limiting surfaces. For more general stationary
embeddings it follows 
from theorem~\ref{theo:bfstatic} that
\begin{corollary} \label{cor:type1}
If the embedding is stationary and of \textbf{Type I}, embedded into flat space-time, then any minimal surface in $\mathbb{E}^{(D-1)}$ will satisfy the blackfold equations as long as $u^{\hat a}u^{\hat b}{K_{\hat a\hat b}}^{i}=0$.
\end{corollary}
\begin{proof}
Stationary minimal surfaces are characterised by a Killing vector field of the form \eqref{eq:kvf} which maps
onto a background Killing vector field in the ambient space-time (we take this as a definition of a stationary
surface in the present context).
The introduction of rotation does not alter the extrinsic curvature tensor of the geometry
\eqref{eq:ext} and therefore such configurations still satisfy $K^{i}=0$ and ${K_{\tau\hat a}}^{i}=0$. Thus, if
$u^{\hat a}u^{\hat b}{K_{\hat a \hat b}}^{i}=0$ both sides of \eqref{eq:bfeom} are separately zero.
\end{proof}
In this case, flat space-time, with embedded stationary geometries, will inherit a limiting surface and the
minimal embedding must be compact, at least in some directions. As we will see in the next section,
amongst all minimal surfaces embedded in $\mathbb{R}^{3}$ we can only achieve this for the plane $\mathbb{R}^{2}$
and the helicoid. 

We now wish to establish corresponding statements that also hold in non-trivial ambient
space-times such as plane wave and de Sitter space-times. We begin with the rather trivial
obervation that any embedding which solves
the blackfold equations in flat space-time also solves the blackfold equations in plane wave space-times
which are flat along the transverse embedding directions:
\begin{theorem} \label{theo:full}
If $X^{q}$ parametrises an embedding surface that solves the blackfold equations in flat space-time then it also solves the blackfold equations in plane wave space-times if $X^{y}=0$ and $A_{rs}=0$ for $r,s$ satisfying $X^{r}\ne0, X^{s}\ne0$.
\end{theorem}
\begin{proof}
The proof follows easily from the fact that if $X^{y}=0$ and $A_{rs}=0$ for $r,s$ satisfying $X^{r}\ne0, X^{s}\ne0$ then ${K_{\tau a}}^{i}={n^{i}}_{\mu}{\Gamma^{\mu}}_{\tau\hat{a}}=0$ and the blackfold equations reduce to those in flat space-time.
\end{proof}
In fact this theorem shows that all configurations constructed in
\cite{Emparan:2009vd} are also valid constructions in plane wave
space-times. However, we note that for the space-time \eqref{ds:pw}
to still be a vacuum solution with a non-trivial plane wave profile, we
need to require the existence of (at least) two additional directions
$i,j$ where the brane is point-like (sitting at $x^{i}=x^{j}=0$) with
$A_{ii}+A_{jj}=0$. This means that all
configurations in \cite{Emparan:2009vd} can be embedded in plane wave
space-times with $D\ge7$.\footnote{The five-dimensional black rings,
helical rings and helical strings found in \cite{Emparan:2009vd} can
be embedded into plane wave space-times in $D\ge6$.} Also, since for
these configurations the induced geometry is exactly the same as in flat
space-time, the free energy functional \eqref{eq:free} to leading order
is also the same and hence, according to \eqref{eq:thermo}, also their
thermodynamic properties.

The type of solutions expressed in theorem \ref{theo:full} are of interest but they do not give rise to black
hole geometries which exhibit the full non-trivial structure of plane wave space-times. We now wish to consider
more non-trivial embeddings into plane wave space-times (along directions where the components $A_{rs}$ are not
necessarily zero) and also into de Sitter space-times. To that end we now first prove 
a more general statement regarding minimal surfaces which establishes the intuitively obvious fact
that a spatial minimal surface extended geodesically (i.e.\ by an extremal curve) in the time direction
is a Lorentzian minimal surface:
\begin{theorem} \label{theo:type1}
If $X^{q}_M$ parametrises a minimal surface in $\tilde{\mathbb{E}}^{(D-1)}$ or $\mathbb{E}^{(D-2)}$ in an embedding of \textbf{Type I} then $X^{q}_M$ parametrises a minimal surface in $\mathbb{L}^{(D)}$ if and only if the embedding is geodesically extended along the time direction, i.e., ${n^{i}}_{\rho}\nabla_{\dot X}\dot X^{\rho}=0~,~\dot X^{\rho}=\partial_\tau X^{\rho}$.
\end{theorem}
\begin{proof}
First note that if $X^{q}_M$ parametrises a minimal surface in $\tilde{\mathbb{E}}^{(D-1)}$ or $\mathbb{E}^{(D-2)}$ in an embedding of \textbf{Type I}  then we have that $\tilde K^{i}=0$. We therefore only need to show that $\gamma^{\tau\tau}{K_{\tau \tau}}^{i}+\gamma^{\tau\hat{a}}{K_{\tau \hat{a}}}^{i}=0$. However, for embeddings of \textbf{Type I} one has that $\gamma_{\tau \hat a}=\gamma^{\tau \hat a}=0$. Therefore, since that for any of the embeddings presented in the previous section one has that $\partial_\tau X^{t}=1$ and that $\partial_\tau X^{\mu}=0~\text{if}~\mu\ne t$, then we must have
\beq \label{eq:c}
{K_{\tau \tau}}^{i}={n^{i}}_{\rho}\left(\ddot X^{\rho}+{\Gamma^{\rho}}_{\mu\nu}\dot X^{\mu}\dot X^{\nu}\right)={n^{i}}_{\rho}\nabla_{\dot X}\dot X^{\rho}=0~~.
\eeq
\end{proof}
Condition \eqref{eq:c} imposes no restrictions if the ambient space-time is flat since the connection vanishes and $\ddot X^{\rho}=0$ for all embeddings presented in the previous section. However, in the case of plane waves or de Sitter space-times condition \eqref{eq:c} reduces to ${n^{i}}_{\rho}{\Gamma^{\rho}}_{tt}=0$ and we obtain specific constraints:
\begin{itemize}

\item \textbf{Plane-wave space-times:} in this case direct to computation leads to 
\beq \label{eq:cpw}
{n^{i}}_{\rho}{\Gamma^{\rho}}_{tt}=\sum_{\rho=1}^{D-1}{n^{i}}_{\rho}x_{\rho}A_{\rho}=0~~.
\eeq
This constraints greatly the number of possible minimal surfaces in \eqref{ds:pw} for embeddings of \textbf{Type I} and, as we will see, also for embeddings of \textbf{Type II}. As we will show in the next section, amongst the minimal surfaces in $\mathbb{R}^3$, only the plane $\mathbb{R}^2$ and the helicoid solve this equation for specific choices of $A_{qr}$. It is important to note that in transverse directions to the worldvolume $i$ where the worldvolume is point-like and located at $x^{i}=0$ then ${n^{i}}_{\rho}{\Gamma^{\rho}}_{tt}=0$.\footnote{If we allow for non-vanishing off-diagonal components of $A_{qr}$ then formula \eqref{eq:cpw} is modified. If this is the case, then one can show that the off-diagonal components $A_{ai}$, where $a$ labels a longitudinal direction along the surface and $i$ labels a direction where the brane is point-like and located at $x_i=0$, must vanish in order for \eqref{eq:cpw} to have a solution.}

\item \textbf{de Sitter space-times:} in space-times of the form \eqref{ds:ds2} we find the constraint 
\beq \label{eq:cds}
{n^{i}}_{\rho}{\Gamma^{\rho}}_{tt}=\frac{1}{2}\sum_{\rho=1}^{D-1}{n^{i}}_{\rho}h(\tilde r)\partial_\rho f(\tilde r)=0~~.
\eeq
The number of minimal embeddings which satisfy this constraint is even more constrained than in plane wave space-times as there are no parameters to tune, by contrast with the components $A_{qr}$. In this case we find that among the various minimal surfaces in $\mathbb{R}^3$ only the plane is a solution. In transverse directions $i$ to the worldvolume where the worldvolume is point-like and located at $x^{i}=0$ we find $\partial_i f(\tilde r)=(x_i/\tilde r)\partial_{\tilde r} f(\tilde r)=0$ and hence ${n^{i}}_{\rho}{\Gamma^{\rho}}_{tt}=0$ along those directions. Theorem \ref{theo:type1} started with the assumption that $X^{q}_M$ parametrises a minimal surface in $\tilde{\mathbb{E}}^{(D-1)}$. However, since we are interested in using known minimal embeddings in $\mathbb{E}^{(D-1)}$ for black hole horizons, it is important to know which of those will also be minimal in conformally Euclidean spaces $\tilde{\mathbb{E}}^{(D-1)}$ in case we want to find black hole horizons in de Sitter space-times \eqref{ds:ds2}. A simple computation, similar to \eqref{eq:cds}, leads to the requirement
\beq \label{eq:cds2}
\gamma^{\hat a \hat b}{\Gamma^{i}}_{\hat{a}\hat{b}}=0~~,
\eeq
which is solved, for example, for the plane $\mathbb{R}^2$ embedded into $\mathbb{R}^3$. As in the case of ${\Gamma^{i}}_{tt}$ one can also easily show that ${\Gamma^{i}}_{\hat{a}\hat{b}}=0$ for transverse directions where the brane is point-like and located at $x^{i}=0$.

\end{itemize}

Theorem \ref{theo:type1} gives the necessary condition for surfaces to be minimal in $\mathbb{L}^{(D)}$.
However, we would like to know what conditions are required for such surfaces to solve the blackfold equations
\eqref{eq:bfeom1}. Similarly to corollary \ref{cor:type1} it follows that
\begin{corollary} \label{cor:type11}
If $X^{q}_M$ is an embedding of \textbf{Type I} satisfying \eqref{eq:c} then it will also satisfy the blackfold equations as long as $u^{\hat a}u^{\hat b}{K_{\hat a \hat b}}^{i}=0$.
\end{corollary}
The proof of corollary \ref{cor:type11} is essentially the same as that given for corollary \ref{cor:type1}. However, if $X^{q}_M$ does not satisfy condition \eqref{eq:c}, and hence is not minimal in $\mathbb{L}^{(D)}$, but is minimal in $\mathbb{E}^{(D-1)}$ or $\mathbb{E}^{(D-2)}$ then we have
\begin{corollary} \label{cor:type111}
If $X^{q}_M$ is an embedding of \textbf{Type I} and is minimal in $\mathbb{E}^{(D-1)}$ or $\mathbb{E}^{(D-2)}$ then it will satisfy the blackfold equations if 
\beq \label{eq:bftype1}
\frac{(\textbf{k}^2+nR_0^2)}{R_0^2}{K_{\tau\tau}}^{i}+n\Omega^a\Omega^b{K_{\phi_a\phi_b}}^{i}-\textbf{k}^2\gamma^{\hat a \hat b}{\Gamma^{i}}_{\hat a \hat b}=0~~.
\eeq
\end{corollary}

Eq.~\eqref{eq:bftype1} follows simply from \eqref{eq:bfeom}, the induced metric \eqref{ds:type1} and \eqref{eq:cds2}. In particular, the last term in \eqref{eq:bftype1} vanishes in flat and plane wave space-times and for minimal surfaces in $\tilde{\mathbb{E}}^{(D-1)}$. This exhausts our study of embeddings of \textbf{Type I}.


\subsubsection*{Embeddings of \textbf{Type II}}

We now turn our attention to embeddings of \textbf{Type II}. The geometric properties of the embeddings \eqref{ds:type2} lead to the simple conclusion
\begin{theorem}
If the embedding is of \textbf{Type II} and $X^{q}_M$ parametrises a minimal surface in $\mathbb{E}^{(D-2)}$ then $X^{q}_M$ is also a minimal surface in $\mathbb{L}^{(D)}$.
\end{theorem}
\begin{proof}
If $X^{q}_M$ parametrises a minimal surface in $\mathbb{E}^{(D-2)}$ then $\tilde K^{i}=0$ and we only need to show that $\gamma^{\tau\tau}{K_{\tau\tau}}^{i}+2\gamma^{\tau z}{K_{\tau z}}^{i}+\gamma^{zz}{K_{zz}}^{i}=0$. Direct computation shows that ${K_{\tau\tau}}^{i}={K_{\tau z}}^{i}={K_{zz}}^{i}$ where ${K_{\tau\tau}}^{i}$ is given by \eqref{eq:cpw}. Therefore we must show that $(\gamma^{\tau\tau}+2\gamma^{\tau z}+\gamma^{zz}){K_{\tau\tau}}^{i}=0$. However, for embeddings of \textbf{Type II} we have that $\gamma^{\tau\tau}+2\gamma^{\tau z}+\gamma^{zz}=0$.\footnote{This is easily seen when changing to $(u,v)$ coordinates by performing the inverse transformation of \eqref{eq:uv}. Then one finds that $\gamma^{\tau\tau}+2\gamma^{\tau z}+\gamma^{zz}\propto \gamma^{uu}=0$.} Therefore we obtain $K^{i}=0$.
\end{proof}
This shows that embeddings of \textbf{Type II} are always minimal embeddings. Therefore we obtain a variant of corollary \ref{cor:type11}, namely, 
\begin{corollary}  \label{cor:type2c}
If the embedding is of \textbf{Type II} then it satisfies the blackfold equations if $u^{a}u^{b}{K_{ab}}^{i}=0$, i.e.,
\beq \label{eq:c3}
{K_{\tau\tau}}^{i}+\Omega^a\Omega^b{K_{\phi_a\phi_b}}^{i}=0~~.
\eeq
\end{corollary}

This concludes the analysis of embeddings of \textbf{Type II}.


\subsubsection*{Embeddings of \textbf{Type III}}

Finally we turn our attention to embeddings of \textbf{Type III}. We will now show that
\begin{theorem} \label{theo:type3}
If $X^{q}_M$ parametrises a static minimal surface in $\mathbb{S}^{(p-1)}$ then embeddings of \textbf{Type III} are minimal surfaces in $\mathbb{L}^{(D)}$ and, furthermore, solve the blackfold equations.
\end{theorem}
\begin{proof}
We begin by showing that if $X^{q}_M$ parametrises a $(p-1)$-dimensional sphere then embeddings of \textbf{Type
III} are miminal. We note that these embeddings can be obtained by first introducing conformally spatially flat
coordinates as in \eqref{ds:ds2}, setting $X^{q}_M=0$ for $q=p+1,...,D-1$, switching back to the original
coordinates \eqref{ds:ds1} and choosing the remaining functions $X^{q}_M$ for $q=1,..,p$ to parametrise the
$(p-1)$-dimensional sphere. Therefore, in all transverse directions $i=p+2,...,D-1$ these embeddings are
point-like and located at $X^{i}_M=0$. Therefore, by the arguments given below \eqref{eq:cds} we have that
${K_{ab}}^{i}=0$. Embeddings of \textbf{Type III} where $X^{q}_M$ parametrises a $(p-1)$-dimensional sphere are
in fact just embeddings of $\mathbb{R}^{(p)}$ into $\mathbb{R}^{(D-1)}$ and hence are of course minimal, as in
the case of the plane $\mathbb{R}^{2}$ into $\mathbb{R}^{3}$. The blackfold equations \eqref{eq:bfeom} are
still satisfied if we introduce rotation since for this embedding the extrinsic curvature is identically 
zero, ${K_{ab}}^{i}=0$. 

Given this, it is now easy to show that if $X^{q}_M$ parametrises a static minimal surface on the unit $(p-1)$-dimensional sphere then the embedding is still minimal and it solves the blackfold equations. If $X^{q}_M$ parametrises a minimal surface then some components of the extrinsic curvature tensor will be non-vanishing along the directions where the geometry is not point-like. Labelling the coordinates on the $(p-1)$-dimensional sphere as $\hat{a},\hat{b}...=1,...,p-1$ then since $X^{q}_M$ is minimal one has that $\gamma^{\hat a \hat b}{K_{\hat a \hat b}}^{i}=0$. Therefore we only need to check what happens to the components ${K_{\tau\tau}}^{i},{K_{\tau \hat a}}^{i},{K_{\tau \rho}}^{i},{K_{\rho\rho}}^{i},{K_{\rho\hat{a}}}^{i}$. By looking at the Christoffel symbols one sees that they all vanish except for ${K_{\rho\hat{a}}}^{i}={n^{i}}_{\lambda}{\Gamma^{\lambda}}_{\rho \hat a}$. However, due to the form of the embedding \eqref{ds:type3} one has that $\gamma^{\rho\hat{a}}=0$, therefore we obtain $K^{i}=0$. Since the geometry is static then $u^{a}u^{b}{K_{ab}}^{i}=0$ and hence the blackfold equations \eqref{eq:bfeom} are satisfied.
\end{proof}
We note that this theorem holds for all space-times of the form \eqref{ds:ds1} including the limits $L\to
\infty$ and $r_m\to 0$. However, as explained at the end of Sec.~\ref{sec:cembed}, if $r_m=0$ then these
solutions suffer from a conical singularity and do not fulfill the validity requirements \eqref{eq:req}. If we consider stationary embeddings instead, then corollary \ref{cor:type11} holds for embeddings of \textbf{Type III}.


\subsection{Classes of solutions}\label{sec:sclasses}
In this section we find different classes of stationary minimal surface solutions in the ambient space-times
described in Sec.~\ref{sec:cembed}. In order to find stationary minimal surface solutions it is necessary to
know which minimal surfaces preserve at least one $U(1)$ family of isometries of the ambient space-time. This
is a difficult problem in general but for minimal surfaces embedded into $\mathbb{R}^{3}$ we will show
that\footnote{Theorem \ref{theo:u1} has actually been proven e.g.\ in \cite{Nitsche:2011} in a different way. However we have decided to present the reader with its proof, which will be useful for the next sections in this paper.}
\begin{theorem} \label{theo:u1}
If $X^{q}_M$ parametrises a minimal surface in $\mathbb{R}^{3}$ which preserves one $U(1)$ family of isometries
of the ambient space-time then it is either the plane, the helicoid, the catenoid or a member of a
one-parameter family of surfaces interpolating between the helicoid and the catenoid (``Scherk's second
surface'').
\end{theorem}
\begin{proof}
A two-dimensional minimal surface can at most preserve one $U(1)$ family of isometries of $\mathbb{R}^{3}$. If $(\rho,\phi)$ are a set of coordinates on the surface and $\phi$ labels the coordinate associated with the isometry of the worldvolume geometry then the induced metric can be written as
\beq \label{ds:iso}
\textbf{ds}^2=f(\rho)d\rho^2+2g(\rho)d\rho d\phi+h(\rho)d\phi^2~~,
\eeq
for some functions $f(\rho),g(\rho),h(\rho)$. Note that since the worldvolume preserves one family of
isometries of the background, the metric coefficients cannot depend on $\phi$. The metric \eqref{ds:iso} has a Killing vector field $\Omega \chi^{a}\partial_a=\Omega\partial_\phi$, where $\Omega$ is the boost velocity of the embedding. If the embedding \eqref{ds:iso} preserves at least a $U(1)$ symmetry of the background then this Killing vector field must map to a Killing vector field of the ambient space-time, i.e., $k^{\mu}=\Omega\partial_a X^{\mu}\chi^{a}$.  Coordinates on the surface can always be chosen such that the preserved $U(1)$ symmetry lies in the $(x_1,x_2)$ plane. Therefore we must have
\beq \label{eq:killing}
k^{\mu}\partial_\mu=\alpha\left(x_1\partial_{x_2}-x_2\partial_{x_1}\right)+\beta \partial_{x_3}~~,
\eeq
for some constants $\alpha$ and $\beta$.\footnote{Note that for this to be a Killing vector field in plane wave
space-times \eqref{ds:pw} we must have $A_1=A_2$ and $A_3=0$ as well as $A_{x_q x_r}=0$ for $q\ne r$ and
$q,r=1,2,3$. Note also that 
it is possible to consider additonally translations in the $x_1$ and $x_2$ directions. 
However, these do not affect the 
results in any significant way since they can always be absorbed by shifting $X^{1}$ and $X^{2}$ by a constant.}. Using that $k^{\mu}=\Omega\partial_a X^{\mu}\chi^{a}$ we find that
\beq \label{eq:kill}
\Omega \partial_\phi X^{1}=-\alpha X_2~~,~~\Omega \partial_\phi X^2=\alpha X_1~~,~~\Omega \partial_\phi X^{3}=\beta~~.
\eeq
From \eqref{ds:iso} we also have that
\beq
\sum_{\mu=1}^{3}\left(\partial_\phi X^{\mu}\right)^{2}=h(\rho)~~.
\eeq
Introducing \eqref{eq:kill} into the above equation we find
\beq
(X^1)^2+(X^2)^2=\frac{\Omega^2 h(\rho)-\beta^2}{\alpha^2}~~.
\eeq
We see that this is the equation for a circle in $X^{1},X^{2}$, therefore we are free to introduce coordinates such that
\beq \label{e:iso}
\begin{split}
&X^{1}(\rho,\phi)=\tilde \lambda y(\rho)\sin(a\phi)+\sqrt{1-\tilde\lambda^2}z(\rho)\cos(a\phi)~~,\\
&X^{2}(\rho,\phi)=-\tilde \lambda y(\rho)\cos(a\phi)+\sqrt{1-\tilde\lambda^2} z(\rho)\sin(a\phi)~~,\\
&X^{3}(\rho,\phi)=\tilde \lambda a \phi + \sqrt{1-\tilde\lambda^2}k(\rho)~~,
\end{split}
\eeq
for some constants $a,\tilde \lambda$ and some functions $y(\rho),z(\rho),k(\rho)$. From here we see that $\alpha=a\Omega$ and $\beta=\tilde \lambda a\Omega$ and that 
\beq
\begin{split}
&f(\rho)=\tilde \lambda^2y'(\rho)^2+(1-\tilde \lambda^2)\left(z'(\rho)^2+k'(\rho)^2\right)~~,\\
&g(\rho)=-a\tilde \lambda\sqrt{1-\tilde \lambda^2} \left(z(\rho)y'(\rho)- k'(\rho)-y(\rho)z'(\rho)\right)~~,\\
&h(\rho)=a^2\left(\tilde \lambda^2(1+y^2(\rho))+(1-\tilde \lambda^2)z^2(\rho)\right)~~,
\end{split}
\eeq
where the prime denotes a derivative with respect to $\rho$. It is always possible to introduce isothermal
coordinates $(\tilde \rho, \tilde \phi)$ such that the induced metric is conformally flat, i.e.\ 
$\gamma_{\tilde{\rho}\tilde{\rho}} = \gamma_{\tilde{\phi}\tilde{\phi}}, \gamma_{\tilde{\rho}\tilde{\phi}}=0$
(see e.g \cite{OperaBook, MeeksReview}). In the case at hand, this can be done in a way compatible with
the manifest isometry associated with the original $\phi$ direction, i.e.\ in such a way that 
$f(\tilde \rho)=h(\tilde \rho)$ and $g(\tilde \rho)=0$. Indeed, this can be accomplished by 
performing the transformation
$\rho=w_1(\tilde \rho)$ and $\phi\to\tilde \phi+w_2
(\tilde \rho)$ for some functions $w_1(\tilde \rho)$ and $w_2(\tilde \rho)$. Then the embedding \eqref{e:iso} takes the same form but
with modified functions $\tilde y(\tilde \rho),\tilde z(\tilde \rho),
\tilde k(\tilde \rho)$. Therefore, dropping the tildes, we can always
choose functions $y(\rho),z(\rho),k(\rho)$ such that
\beq \label{helcat:c}
\begin{split}
&\tilde \lambda^2\left(y'(\rho)^2-y(\rho)^2-1\right)+(1-\tilde\lambda^2)\left(z'(\rho)^2-z(\rho)^2+k'(\rho)^2\right)=0~~,\\
&z(\rho)y'(\rho)-y(\rho)z'(\rho)-k'(\rho)=0~~,
\end{split}
\eeq 
where we have rescaled $\phi\to\phi/a$ for simplicity. We do not need to solve this explicitly, instead we note that if a surface, embedded in $\mathbb{R}^{3}$, is written in isothermal coordinates then it is minimal if $X^i(\rho,\phi)$ for $i=1,2,3$ is harmonic \cite{OperaBook, MeeksReview}, i.e., if $\partial^2_\rho X^{i}(\rho,\phi)+\partial^2_\phi X^{i}(\rho,\phi)=0$. This means that, assuming $y(\rho),z(\rho),k(\rho)$ to satisfy \eqref{helcat:c}, then from \eqref{e:iso} we must have
\beq \label{sol:helcat}
\begin{split}
&\tilde \lambda \sin\phi \left(y''(\rho)-y'(\rho)\right)+\sqrt{1-\tilde\lambda^2}\cos\phi\left(z''(\rho)-z'(\rho)\right)=0~~,\\
&k''(\rho)=0~~.
\end{split}
\eeq
The first condition was obtained from $X^{1}(\rho,\phi)$ and is equivalent to the one obtained form $X^{2}(\rho,\phi)$. The second condition was obtained from $X^{3}(\rho,\phi)$ and is solved if $k(\rho)=a_k\rho+b_k$ for some constants $a_k,b_k$. Without loss of generality we can set $a_k=1$ and $b_k=0$. Since the first condition in \eqref{sol:helcat} must be solved for all $\phi$ then we must have that $y''(\rho)-y'(\rho)=0$ and $z''(\rho)-z'(\rho)=0$. This leads to the requirement that $y(\rho),z(\rho)$ must be of the form
\beq \label{e:gensol}
y(\rho)=a_ye^{\rho}+b_ye^{-\rho}~~,~~z(\rho)=a_ze^{\rho}+b_ze^{-\rho}~~,
\eeq
for some constants $a_y,b_y,a_z,b_z$. By using the freedom to translate $\rho$ by a constant $d$ such that $\rho\to\rho+d$ we can set $a_z=b_z$. Introducing this into \eqref{helcat:c} allows to find expressions for $a_y$ and $b_y$ in terms of $a_z,\lambda$. There is only one solution which is valid for all $\lambda$, namely, $a_y=a_z=b_z=-b_y=1/2$. By rescaling $\rho\to\rho/c$ and $X^{3}(\rho,\phi)\to cX^{3}(\rho,\phi)$ we bring the solution \eqref{e:gensol} to a more familiar form
\beq \label{sol:iso}
y(\rho)=c \sinh\left(\frac{\rho}{c}\right)~~,~~z(\rho)=c\cosh\left(\frac{\rho}{c}\right)~~,
\eeq
which is unique up to reparametrizations of the coordinate $\rho$. In fact, this configuration is a family of
minimal surfaces that interpolates between the helicoid ($\tilde \lambda=1$) and the catenoid ($\tilde
\lambda=0$) and is known as Scherk's second surface \cite{Sauvigny2010}. We shall refer to these interpolating
surfaces simply as Scherk surfaces in this paper.\footnote{This family of surfaces was also called 
\textit{Helicatenoids} in \cite{Ogawa1992}.} We will analyse in detail this general solution in Sec.~\ref{sec:catenoid}. The metric \eqref{ds:iso}, using \eqref{sol:iso} and after rescaling back $\phi\to a \phi$, is diagonal and takes the form
\beq
f(\rho)=\frac{h(\rho)}{a^2 c^2}=\cosh^2\left(\frac{\rho}{c}\right)~~,~~g(\rho)=0~~,
\eeq
with extrinsic curvature components\footnote{Note that we have omitted the transverse index $i$ from ${K_{ab}}^{i}$ since the surface is of codimension one.}
\beq \label{ext:full}
{K_{\rho\rho}}=-\frac{\sqrt{1-\tilde\lambda^2}}{c}~~,~~{K_{\rho\phi}}=-a\tilde\lambda~~,~~{K_{\phi\phi}}=a^2c\sqrt{1-\tilde\lambda^2}~~.
\eeq
The solution \eqref{sol:iso} can be seen as a combination of two different cases:
\begin{itemize}
\item \textbf{The catenoid $\tilde \lambda=0$}: in this case, the minimal surface equation \eqref{eq:min} yields 
\beq
\frac{z''(\rho)}{1+z'(\rho)^2}-\frac{1}{z(\rho)}=0~~,
\eeq
which has a unique solution, namely \eqref{sol:iso}. The extrinsic curvature components take the form
\beq \label{ext:catenoid}
{K_{\rho\rho}}=-\frac{1}{c}~~,~~{K_{\rho\phi}}=0~~,~~{K_{\phi\phi}}=a^2 c~~.
\eeq
\item \textbf{The helicoid $\tilde \lambda=1$}: in this case, by redefining $ac=\lambda$ and hence making the embedding \eqref{e:iso} into the same form as that of \eqref{e:helicoid}, the only non-vanishing extrinsic curvature component is
\beq \label{ext:helicoid}
{K_{\rho\phi}}=-\frac{a\lambda y'(\rho)}{\sqrt{(\lambda^2+a^2y^2(\rho))}}~~.
\eeq
Since the metric \eqref{ds:iso} has no component $\gamma_{\rho \phi}$ then the minimal surface equation \eqref{eq:min} is automatically satisfied for these embeddings independently of the form of $y(\rho)$. Therefore we are free to choose $y(\rho)=\rho$ as in \eqref{e:helicoid} or as that given in \eqref{sol:iso}. Hence we recover the helicoid \eqref{e:helicoid}, which if $\lambda=0$ reduces to the plane.

\end{itemize}
This completes the proof.
\end{proof}


\subsubsection*{Solutions for flat space-time}
According to corollary \ref{cor:type1}, stationary minimal surfaces of \textbf{Type I} must satisfy $u^{\hat a}u^{\hat b}{K_{\hat a \hat b}}^{i}=0$. In $\mathbb{R}^{3}$, we have seen that there are only four possibilities of stationary minimal surfaces in which case, according to \eqref{eq:stress}, one has that $u^{\tau}=\textbf{k}^{-1}$ and $u^{\phi}=\Omega \textbf{k}^{-1}$. Hence we must satisfy $u^{\phi}u^{\phi}{K_{\phi\phi}}^{i}=0$ which implies that we must have ${K_{\phi\phi}}^{i}=0$. From \eqref{ext:full} we arrive at the following conclusion
\begin{corollary} \label{cor:type1sol}
The only two stationary minimal surfaces of \textbf{Type I}, where $X^{q}_M$ parametrises a minimal surface embedded in $\mathbb{R}^{3}$, that solve the blackfold equations in flat space-time are the plane and the helicoid.
\end{corollary}
It is a difficult problem to make equivalent statements in $\mathbb{R}^{(D-1)}$, however, we will show in Sec.~\ref{sec:examples} and in App.~\ref{sec:higherhelicoid} that higher-dimensional generalisations of the plane and the helicoid, respectively, also solve the blackfold equations. 


\subsubsection*{Solutions for plane wave space-times}
In asymptotic plane wave space-times \eqref{ds:pw} solutions of the blackfold equations can be of \textbf{Type I} or \textbf{Type II}. If they are of \textbf{Type I} and $X^{q}_M$ parametrises a minimal surface in $\mathbb{L}^{(D)}$ then according to corollary \ref{cor:type11}, it must satisfy $u^{\hat a}u^{\hat b}{K_{\hat a \hat b}}^{i}=0$, otherwise it must satisfy \eqref{eq:bftype1}. In this case we can show the following
\begin{theorem} \label{theo:solpw}
The only two stationary minimal surfaces of \textbf{Type I}, where $X^{q}_M$ parametrises a minimal surface embedded in $\mathbb{R}^{3}$, that solve the blackfold equations in plane wave space-times with diagonal $A_{qr}$ are the plane and the helicoid.
\end{theorem}
\begin{proof}
First we consider the case where $X^{q}_M$ also parametrises a minimal surface in $\mathbb{L}^{(D)}$ and hence must satisfy \eqref{eq:cpw}. Using the Monge parametrisation of Sec.~\ref{sec:minimal} we can write \eqref{eq:c} as\footnote{Note that we are assuming the embedding to be point-like in all other transverse directions.}
\beq \label{eq:minpw}
{K_{\tau\tau}}=\frac{1}{\sqrt{1+f_u^2+f_v^2}}\left(-A_{1}uf_u-A_{2}vf_v+A_3f(u,v)\right)=0~~.
\eeq
We split the solutions of this equation into two sub cases:
\begin{itemize}
\item \textbf{The plane:} the simplest solution of \eqref{eq:minpw} consists of choosing $f(u,v)=0$, which also trivially solves the minimal surface equation \eqref{eq:minm}. This describes the $\mathbb{R}^2$ plane sitting at $x_3=0$, which can be seen by introducing polar coordinates $(u,v)=(r\cos\theta,r\sin\theta)$.\footnote{If one considers off-diagonal components of $A_{rs}$ it is possible to obtain an arbitrary $\mathbb{R}^2$ plane embedded into $\mathbb{R}^3$, and not necessarily sitting at $x_3=0$. This is described by an equation of the form $a f_u + b f_v + c=0$.}

\item \textbf{The helicoid:} the general solution to equation \eqref{eq:minpw} requires $f(u,v)$ to be of the form
\beq
f(u,v)=u^{\frac{A_3}{A_1}}f(u^{-\frac{A_2}{A_1}}v)~~.
\eeq 
Introducing this into the minimal surface equation \eqref{eq:minm} requires to set $A_1=A_2~,~A_3=0$ and solving the equation
\beq
2uvf'\left(\frac{v}{u}\right)+(u^2+v^2)f''\left(\frac{v}{u}\right)=0~~,
\eeq
where the prime represents a derivative with respect to $v/u$. This has a unique solution
\beq
f(u,v)=\alpha \arcsin\left(\frac{v}{u\sqrt{1+\frac{v^2}{u^2}}}\right)~~,
\eeq
for some constant $\alpha$, up to reparametrizations of $f(u,v)$. By introducing polar coordinates $(u,v)=(r\cos(a\theta),r\sin(a\theta))$ and defining $\lambda=\alpha a$ this gives the parametrisation of the helicoid \eqref{e:helicoid}.
\end{itemize}
Since the helicoid has extrinsic curvature \eqref{ext:helicoid} (which includes the case of the plane when $\lambda=0$) then they satisfy corollary \ref{cor:type11} and also \eqref{eq:bftype1}. Hence, they are solutions of the blackfold equations in these space-times.\footnote{We have considered off-diagonal terms in $A_{qr}$, in which case, more solutions to $K_{\tau\tau}^{i}=0$ can be found analytically but they do not satisfy \eqref{eq:minm}. We have also tried to solve it for classical minimal surfaces such as Enneper surface, Scherk first surface, Henneberg surface and Bour's surface but these do not solve $K_{\tau\tau}^{i}=0$. }

If $X^{q}_M$ does not parametrize a minimal surface in $\mathbb{L}^{(D)}$ then it must satisfy \eqref{eq:bftype1}. The plane and the helicoid trivially satisfy this equation since adding rotation does not affect the result due to the form of the extrinsic curvature \eqref{ext:helicoid}. Therefore we are only left with the catenoid and Scherk surfaces as the last possibilities. Focusing first on the catenoid, since we are in plane wave space-times we have that ${\Gamma^{i}}_{\hat a \hat b}=0$. Using the parametrisation \eqref{e:iso} with $\tilde \lambda=0$, $a=1$ and $z(\rho)=c\cosh(\rho/c)$, which highlights the $U(1)$ symmetry, the equation of motion \eqref{eq:bftype1} reduces to
\beq \label{eq:nocat}
c^2A_1 \cos^2\left(\frac{\rho}{c}\right)\left((n-1)A_1+\Omega^2(n+1)\right)+\left(n\Omega^2-(n+1)A_1\right)=0~~,
\eeq
where we were forced to set $A_3=0$ otherwise a term proportional to $\rho\tan(\rho/c)$ would appear and also $A_2=A_1$ otherwise the Killing vector field \eqref{eq:killing} would not be a Killing vector field of \eqref{ds:pw}. We have also used that $\textbf{k}^2=R_0^2-c^2\Omega^2\cos^2(\rho/c)$ and that $R_0^2=1+A_1c^2\cosh^2(\rho/c)$. From \eqref{eq:nocat} we see that the first set of terms requires $A_1<0$ and the second set requires $A_1>0$. Therefore the catenoid does not solve \eqref{eq:bftype1}. For Scherk surfaces, this result also holds since according to \eqref{ext:full} the component ${K_{\phi\phi}}$ of the extrinsic curvature of the embedding only changes by a multiplicative factor of $\sqrt{1-\tilde\lambda^2}$ and the same happens to the component ${K_{\tau\tau}}$.

\end{proof}
In higher dimensions one can show that embeddings of $\mathbb{R}^{(p)}$ into $\mathbb{R}^{(D-1)}$ or $\mathbb{R}^{(D-2)}$ and higher-dimensional helicoids also solve \eqref{eq:cpw}. For \textbf{Type II} embeddings into plane wave space-times we can show the following
\begin{theorem} \label{theo:solpw2}
The only stationary minimal surfaces of \textbf{Type II}, where $X^{q}_M$ parametrises a minimal surface embedded in $\mathbb{R}^{3}$, that solve the blackfold equations in plane wave space-times with diagonal $A_{qr}$ are the plane, the helicoid, the catenoid and Scherk's second surface.
\end{theorem}
\begin{proof}
From the last theorem it follows that the plane and the helicoid trivially satisfy Eq.~\eqref{eq:c3}. For the catenoid, using \eqref{e:iso} with $\lambda=0$ and $z(\rho)=c\cosh(\rho/c)$, as well as \eqref{ext:catenoid}, Eq.~\eqref{eq:c3} reduces to
\beq \label{sol:cat2}
-cA_1\cos^2\theta-cA_2\sin^2\theta+A_3\rho\tanh\left(\frac{\rho}{c}\right)+ca^2\Omega^2=0~~,
\eeq
Since the Killing vector field \eqref{eq:killing} must be a Killing vector field of \eqref{ds:pw} then we must set $A_2=A_1$ and since we have a linear term proportional to $\rho\tanh(\rho/c)$ we must set $A_3=0$. Therefore we obtain a solution if $A_1=a^2\Omega^2$ and $A_1>0$. Again, ${K_{\tau\tau}}$ and ${K_{\phi\phi}}$ only change by a multiplicative factor of $\sqrt{1-\tilde\lambda^2}$, so this result is also valid for Scherk surfaces.
\end{proof}
In Sec.~\ref{sec:examples} and in Apps.~\ref{sec:higherhelicoid}-\ref{sec:highercatenoid} we will show that these results also hold for higher-dimensional planes, helicoids and catenoids. 


\subsubsection*{Solutions for de Sitter space-time}
 In asymptotically de Sitter space-times the number of solutions of \textbf{Type I} is more constrained than in plane wave space-times. In this case one can show the following:
\begin{theorem} \label{theo:solpw}
The only stationary minimal surface of \textbf{Type I}, where $X^{q}_M$ parametrises a minimal surface embedded in $\mathbb{R}^{3}$, that solve the blackfold equations in de Sitter space-times is the plane.
\end{theorem}
\begin{proof}
If $X^{q}_M$ also parametrises a minimal surface in $\mathbb{L}^{(D)}$ then it must satisfy \eqref{eq:cds}. Using the Monge parametrisation we can write \eqref{eq:cds} explicitly as
\beq
{K_{\tau\tau}}=\frac{1}{2}\frac{h(\tilde r)\partial_{\tilde r}f(\tilde r)}{\sqrt{1+f_u^2+f_v^2}} \frac{1}{\tilde r}\left(-f_u x_1-f_v x_2+f(u,v) x_3\right)~~.
\eeq
This case is very similar to \eqref{eq:minpw}, the difference being that the length scale associated with each of the coordinates $x_{q}$ is the same (in the case of pure de Sitter $r_m=0$ these are just equal to $L$). Therefore following the same analysis as for plane wave space-times, only the plane $\mathbb{R}^{2}$, described by $f(u,v)=0$, is a solution. Since $f(u,v)=0$ then ${\Gamma^{i}}_{\hat a \hat b}=0$ and hence \eqref{eq:bftype1} is satisfied. Explicit evaluation of \eqref{eq:bftype1} for the helicoid, catenoid and Scherk surfaces shows that \eqref{eq:bftype1} cannot be satisfied for these configurations.
\end{proof}
It is trivial to show that configurations consisting of $\mathbb{R}^{(p)}$ embedded into $\mathbb{R}^{(D-1)}$ satisfy \eqref{eq:bftype1}, as shown in theorem \ref{theo:type3}. For embeddings of \textbf{Type III}, as shown in theorem \ref{theo:type3}, all minimal surfaces on the unit sphere provide solutions to the blackfold equations.

\section{Minimal surfaces and black hole horizons} \label{sec:examples}
In this section we explicitly construct the blackfold solutions within the classes presented in Sec.~\ref{sec:sclasses}. We study their limiting surfaces, thermodynamic properties and their validity within the blackfold approximation. These configurations consist of planes, helicoids, catenoids and Scherk surfaces as well as of minimal surfaces on the unit sphere such as the Clifford torus. We deal with higher dimensional versions of helicoids and catenoids in App.~\ref{sec:higherhelicoid} and App.~\ref{sec:highercatenoid}.

 \subsection{Black discs and helicoids in flat space-time}\label{sec:helicoid}
As mentioned in Sec.~\ref{sec:minimal}, the embedding of the helicoid \eqref{e:helicoid} includes the $\mathbb{R}^{2}$ plane as a special case when $\lambda=0$. In order to understand better the case of the helicoid we first review the case of the $\mathbb{R}^{2}$ plane first studied in \cite{Emparan:2009vd} and we also analyse its regime of validity according to the prescription of Sec.~\ref{sec:validity}, which was not done in \cite{Emparan:2009vd} . 

\subsubsection{Black discs}\label{sec:blackdisc}
From corollary \ref{cor:type1sol}, the plane is a \textbf{Type I} embedding that solves the blackfold equations in flat space-time. The mapping functions are chosen such that
\beq \label{map:plane}
t=\tau~~,~~X^{1}(\rho,\phi)=\rho\cos\phi~~,~~X^{2}(\rho,\phi)=\rho\sin\phi~~,~~X^{i}=0~,~i=3,...,D-1~~,
\eeq 
where $\rho\ge0$ and $0\le \phi\le2\pi$ and hence the induced metric \eqref{ds:type1} takes the form of \eqref{ds:plane}
\beq\label{ds:plane1}
\textbf{ds}^2=-d\tau^2+d\rho^2+\rho^2d\phi^2~~.
\eeq 
As explained in Sec.~\ref{sec:cembed}, introducing rotation creates a limiting surface in the space-time which can make the geometry compact. Therefore we add rotation to the plane such that the geometry is characterised by a Killing vector field of the form \eqref{eq:kill},
\beq
\textbf{k}^{a}\partial_a=\partial_\tau+\Omega\partial_\phi~~,~~\textbf{k}^2=1-\Omega^2\rho^2~~.
\eeq 
From the form of $\textbf{k}$ we see that a limiting surface appears at $\textbf{k}=0$. According to \eqref{eq:bfeom1} this satisfies the boundary condition \eqref{eq:bfeom1} and hence the geometry has a circular boundary at $\rho_+=\Omega^{-1}$, rendering the plane $\mathbb{R}^{2}$ compact. The geometry is thus that of a disc $\mathbb{D}$ of radius $\Omega^{-1}$ with an $(n+1)$-sphere of radius $r_0(\rho)$ fibered over it due to the transverse sphere in the metric \eqref{ds:blackp}. Therefore these black hole horizons (spatial sections) have topology $\mathbb{S}^{(D-2)}$. The size of the transverse sphere $r_0(\rho)$ is simply given by \eqref{eq:pressure}, i.e.,
\beq
r_0(\rho)=\frac{n}{4\pi T}\sqrt{1-\Omega^2\rho^2}~~.
\eeq
It is clear from this expression that $r_0(\rho)$ is maximal at the centre of the disc when $\rho=0$ and shrinks to zero at the boundary. The geometry is depicted in Fig.~\ref{fig:disc}.
\begin{figure}[h!] 
\centering
  \includegraphics[width=0.4\linewidth]{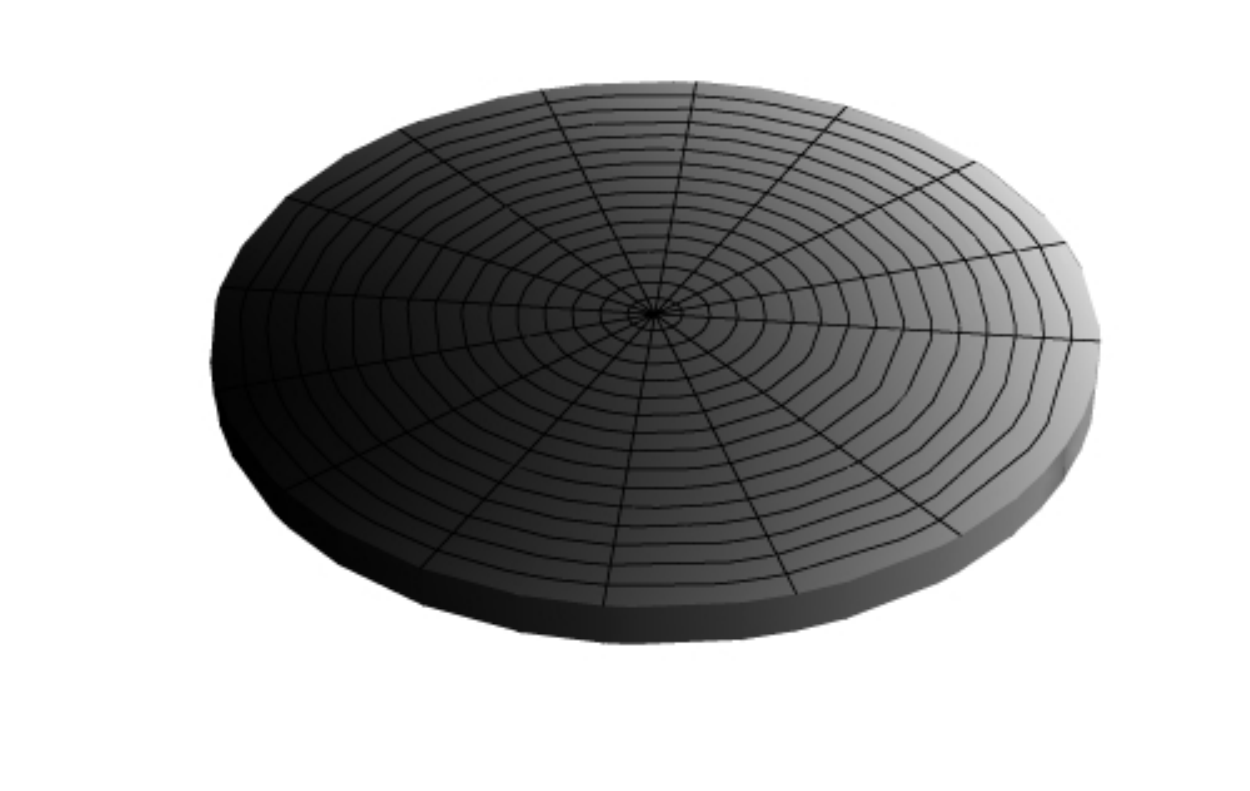}
  \begin{picture}(0,0)(0,0)
\end{picture} 
\caption{Embedding of the rotating black disc in $\mathbb{R}^{3}$ with $\Omega=1$. } \label{fig:disc}
\end{figure}
This configuration describes the ultraspinning regime of singly-spinning Myers-Perry black holes \cite{Emparan:2009vd} and exists in $D\ge6$.

\subsubsection*{Validity analysis}
As this disc geometry, embedded into flat space-time, constitutes the simplest example of a blackfold geometry we will apply the validity analysis of Sec.~\ref{sec:validity} in order to exhibit its usefulness. For this geometry all intrinsic and extrinsic curvature invariants vanish, since it is Ricci-flat and trivially embedded in flat space-time. Therefore, of all the invariants described in \eqref{eq:req}, the only non-vanishing one is the invariant associated with variations in the local temperature (or thickness) $|\textbf{k}^{-1}\nabla_a\nabla^{a}\textbf{k}|^{-\frac{1}{2}}$. Explicitly, this leads to the requirement,
\beq \label{req:mpfull}
r_0\ll \frac{1-\Omega^2\rho^2}{\Omega \sqrt{2-\Omega^2\rho^2}}~~.
\eeq
Since, from \eqref{eq:pressure}, we have that $r_0\propto\textbf{k}$ then this implies that near the axis of rotation $\rho=0$ we must have 
\beq \label{req:mp}
r_+\Omega\ll1~~,~~r_+=\frac{n}{4\pi T}~~.
\eeq
According to the identification with the thermodynamics of Myers-Perry black holes in \cite{Emparan:2009vd},
the angular velocity is given by $\Omega=b^{-1}$ where $b$ is the rotation parameter of the singly-spinning
Myers-Perry black hole. Therefore one should require $r_+\ll b$, which is the original assumption when taking
the ultraspinning limit of Myers-Perry black holes and focusing only on the axis of rotation \cite{Emparan:2003sy}. At any other point on the worldvolume the requirement \eqref{req:mpfull} reproduces the result of App. B of \cite{Armas:2011uf} where the ultraspinning limit was taken at an arbitrary point on the disc and not only at the axis of rotation. Near the boundary $\textbf{k}=0$ the requirement \eqref{req:mpfull} cannot be satisfied for any finite value of $\Omega$. Therefore we introduce $\epsilon\ll1$ and consider the approximation valid in the interval $0\le\rho\le\rho_+-\epsilon$ while assuming the existence of a well defined boundary expansion.

This analysis shows that validity requirements based on the second order corrected free energy
are well based. The requirement \eqref{req:mp} can also be recast as $r_+\ll\rho_+$, exhibiting the need for
two widely separated horizon length scales. This provides a nice illustration of the fact that 
a classification of second order invariants is required in order to assess the validity of blackfold configurations to leading order. For most of the configurations in the core of this paper we will simply state the results obtained from a detailed analysis of the invariants \eqref{eq:req}, which is presented in App.~\ref{sec:valanal}.

\subsubsection*{Free energy}
Despite the fact that, according to the analysis above, the approximation is expected to break down around $\rho=\rho_+-\epsilon$ we can determine its thermodynamic properties exactly to leading order in $\epsilon$. The leading order free energy, using \eqref{eq:free}, is given by
\beq
\begin{split}
\mathcal{F}=\frac{\Omega_{(n+1)}}{16\pi G}r_+^{n}\int_{0}^{2\pi}d\phi&\int_{0}^{\Omega^{-1}-\epsilon}d\rho~\rho\left(1-\Omega^2\rho^2\right)^{\frac{n}{2}}~~\\
&=\frac{\Omega_{(n+1)}}{8 G}r_+^{n}\frac{1+\Omega  \epsilon  (\Omega  \epsilon -2) (\Omega  \epsilon  (2-\Omega  \epsilon ))^{\frac{n}{2}}}{(n+2) \Omega
   ^2}~~.
\end{split}
\eeq	
Since that $\epsilon\Omega\ll1$, the above expression for the free energy reduces to
\beq \label{free:mpfull}
\mathcal{F}=\frac{\Omega_{(n+1)}}{8 G}r_+^{n}\left(\frac{1}{(n+2)\Omega^2}+\mathcal{O}(\epsilon^{\frac{n+2}{2}})\right)~~.
\eeq
From here we see that, according to the identification given in \cite{Emparan:2009vd}, the free energy for these configurations matches, to leading order in $\epsilon$, the free energy of ultraspinning Myers-Perry black holes. We note that, the analysis of App. B of \cite{Armas:2011uf} shows that, even though the blackfold approximation is expected to break down near the boundary for ultraspinning Myers-Perry black holes, the metric all the way to the boundary is still that of a locally flat brane \eqref{ds:blackp}. This provides an example in which the assumption of the existence of a smooth limit of the blackfold description when $r_0\to0$ gives rise to the correct description of the gravitational object. 

The fact that the free energy \eqref{free:mpfull} gives rise to the correct thermodynamic properties of the configuration, to leading order in $\epsilon$, is generic for all configurations with boundaries that we consider. The reason for this is due to the fact that the free energy \eqref{eq:free} to leading order approaches zero near the boundary and hence contributions of the integrand near the boundary are highly suppressed.\footnote{Higher-order contributions in a derivative expansion are also suppressed if $n>2$. In the cases $n=1,2$ backreaction and self-force corrections are expected to be dominant with respect to derivative corrections \cite{Armas:2011uf} and therefore the effective free energy to second order should not in general be trusted.} For this reason, in all the examples that follow, we perform integrations all the way to the boundary points but one should bear in mind that such results are only valid to leading order in a boundary expansion.

\subsubsection{Black helicoids} \label{sec:flathelicoids}
Helicoid geometries are embeddings of \textbf{Type I} in flat space-time that also solve the blackfold
equations according to corollary \ref{cor:type1sol}. Explicitly, this embedding is described by 
\beq \label{emb:helic}
t=\tau~~,~~X^{1}(\rho,\phi)=\rho\cos(a\phi)~~,~~X^{2}(\rho,\phi)=\rho\sin(a\phi)~~,~~X^{3}(\rho,\phi)=\lambda \phi~~,
\eeq
and $X^{i}=0~,~i=4,...,D-1$, where the coordinates lie within the range $-\infty<\rho,\phi<\infty$. The only physically relevant parameter in this embedding is the pitch $\lambda/a$, since, if $\lambda\ne0$, the coordinate $\phi$ can always be rescaled such that $a$ can be set to 1.  However, since we are interested in taking the limit $\lambda\to0$ we keep both parameters. Without loss of generality, we take $\lambda\ge0$ and $a>0$. The induced metric \eqref{ds:type1} takes the form
\beq \label{ds:helicoid}
\textbf{ds}^2=-d\tau^2+d\rho^2+(\lambda^2+a^2\rho^2)d\phi^2~~.
\eeq
As in the case of the plane, we boost the helicoid along the $\phi$ direction with boost velocity $\Omega$ such that
\beq \label{kill:helicoid}
\textbf{k}^{a}\partial_a=\partial_\tau+\Omega\partial_\phi~~,~~\textbf{k}^2=1-\Omega^2(\lambda^2+a^2\rho^2)~~.
\eeq
According to \eqref{eq:killing} this corresponds to a Killing vector field in the ambient space-time of the form
\beq
k^{\mu}\partial_\mu=\partial_t+a\Omega\left(x_1\partial_{x_{2}}-x_2\partial_{x_1}\right) +\lambda\Omega \partial_{x_3}~~,
\eeq
that is, the helicoid geometry is rotating in the $(x_1,x_2)$ plane with angular velocity $a\Omega$ and it is boosted along the $x_3$ direction with boost velocity $\lambda\Omega$.
From Eq.~\eqref{kill:helicoid}, we see that a limiting surface, constraining the coordinate $\rho$, appears at $\textbf{k}=0$ when
\beq \label{eq:reqh}
\rho_{\pm}=\pm\frac{\sqrt{1-\Omega^2\lambda^2}}{a\Omega}~~,
\eeq
which implies that we must have $\Omega^2\lambda^2<1$. This limiting surface makes the helicoidal geometry compact in the $\rho$ direction but leaves the $\phi$ direction unconstrained. Therefore these geometries are non-compact in the $\phi$ direction. The black hole horizons they give rise to have topology $\mathbb{R}\times \mathbb{S}^{(D-3)}$ in $D\ge6$, hence they have the topology of a black string. We therefore refer to these geometries as \emph{helicoidal black strings}, which can be thought of as the membrane generalisation of the helical strings found in \cite{Emparan:2009vd}). The fact that these geometries have string topology suggests that they can be bent into a helicoidal ring, in the same way that helical strings can be bent into helical rings \cite{Emparan:2009vd}. In a related publication \cite{Armas:2015nea}, we show that this is indeed the case.\footnote{We thank Roberto Emparan for suggesting this possibility to us.} The size of the transverse sphere $r_0(\rho)$ is given by
\beq
r_0(\rho)=\frac{n}{4\pi T}\sqrt{1-\Omega^2(\lambda^2+a^2\rho^2)}~~,
\eeq
and is again maximal at the origin $\rho=0$ and vanishes at the boundaries $\rho_\pm$. If $\Omega=0$ then the geometry is static and becomes non-compact also in the $\rho$ direction. This geometry is depicted in Fig.~\ref{fig:helicoid}
\begin{figure}[h!] 
\centering
  \includegraphics[width=0.6\linewidth]{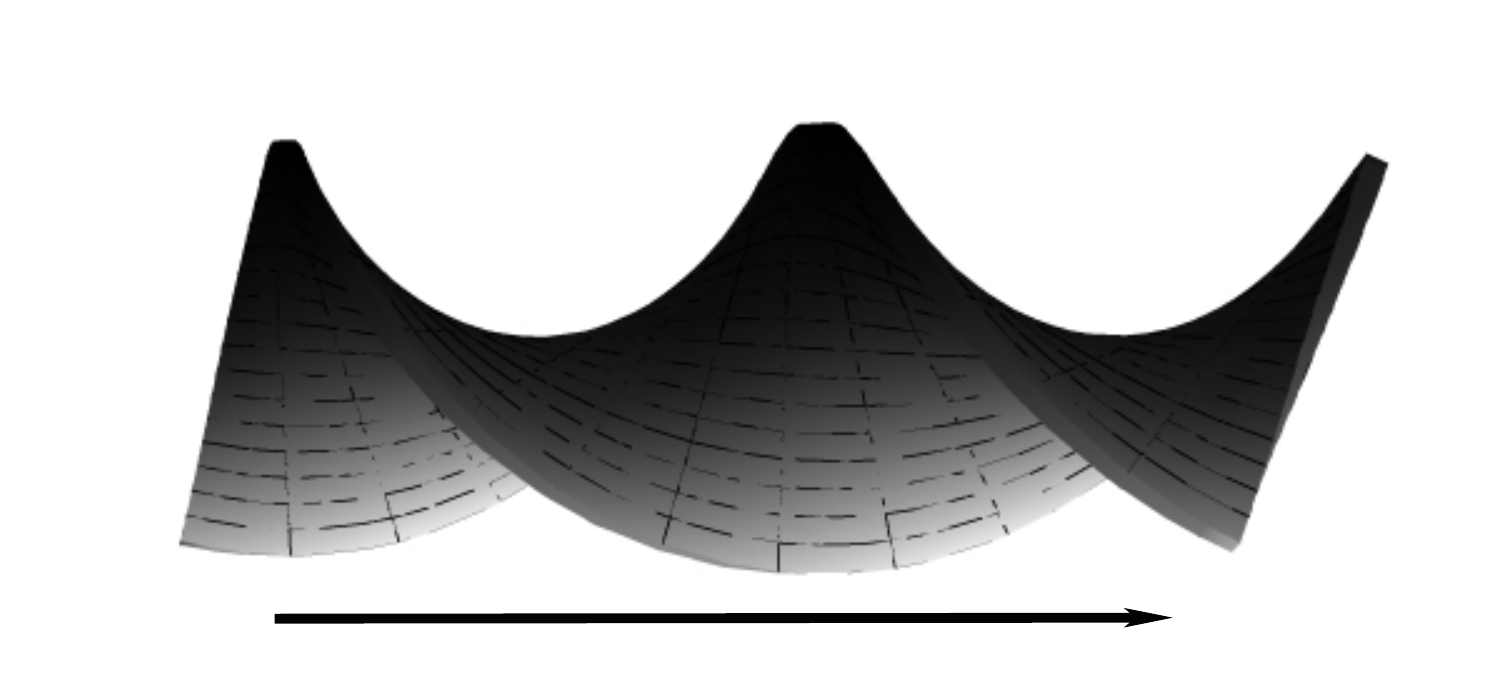}
  \begin{picture}(0,0)(0,0)
\put(-80,0){ $ \phi $}
\end{picture} 
\caption{Embedding of the rotating black helicoid in $\mathbb{R}^{3}$ with $\lambda=a=\Omega=1$, depicted in the interval $-3\le\phi\le3$. } \label{fig:helicoid}
\end{figure}
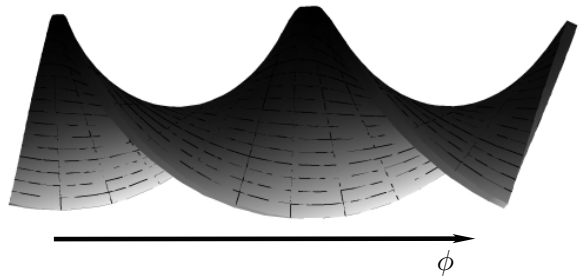

\subsubsection*{The free energy and the Myers-Perry limit}

The free energy of these configurations can be obtained by evaluating \eqref{eq:free} to leading order, yielding
 \beq \label{eq:freehelicoidflat}
 \begin{split}
 \mathcal{F}&=\frac{\Omega_{(n+1)}}{16\pi G}r_+^{n}\int d\phi\int_{\rho_-}^{\rho_+} d\rho\sqrt{\lambda^2+a^2\rho^2}\left(1-\Omega^2(\lambda^2+a^2\rho^2)\right)^{\frac{n}{2}}\\
&= \frac{\Omega_{(n+1)}}{16\sqrt{\pi } G}\frac{r_+^{n}}{a\Omega}\int d\phi \lambda  \Gamma \left(1+\frac{n}{2}\right) \left(1-\lambda ^2 \Omega
   ^2\right)^{\frac{n+1}{2}} \,
   _2\tilde{F}_1\left(-\frac{1}{2},\frac{1}{2};\frac{n+3}{2};1-\frac{1}{\lambda ^2 \Omega
   ^2}\right)~~.
 \end{split}
 \eeq
The free energy is positive for all $n$ and, since the geometry is non-compact in the $\phi$ direction, is infinite. Hence it is only physically relevant to speak about the free energy density, i.e., the free energy \eqref{eq:freehelicoidflat} modulo the integration over $\phi$. The remaining thermodynamic properties can be easily obtained from Eqs.~\eqref{eq:freer}-\eqref{eq:thermo} and we leave a more detailed analysis of these to a later publication \cite{Armas:2015nea}. We note, however, that these geometries have a non-trivial tension \eqref{eq:tension} as expected, since they are non-compact in the $\phi$-direction.

As mentioned in Sec.~\ref{sec:minimal} the embedding of the helicoid \eqref{emb:helic} reduces to that of the plane when $\lambda\to0$, however the coordinate range of $\rho$ lies in between $\rho_-<\rho<\rho_+$ instead of $0\le\rho\le\rho_+$. Therefore, in this limit, one is covering the disc twice. In order to avoid this double covering, we rescale the free energy \eqref{eq:freehelicoidflat} such that $ \mathcal{F}\to(1/2) \mathcal{F}$ when taking the limit $\lambda\to0$. More precisely, we take the limit $\lambda\to0$ while keeping $a$ fixed and make the $\phi$-coordinate periodic with period $2\pi/a$.  Integrating the free energy \eqref{eq:freehelicoidflat} in the interval $0\le\phi\le2\pi/a$ and rescaling $ \mathcal{F}\to(1/2) \mathcal{F}$ leads to the result for the disc \eqref{free:mpfull} to leading order in $\epsilon$, once finally setting $a=1$. 

The existence of this non-trivial agreement with the geometry and thermodynamics of the disc in the limit $\lambda\to0$ suggests that the family of singly-spinning Myers-Perry black holes and the family of black helicoids are connected, at least in the ultraspinning limit, in which the topology changes according to $\mathbb{R}\times \mathbb{S}^{D-3}\to\mathbb{S}^{D-2}$. These geometries, according to the analysis of App.~\ref{sec:valanal}, are valid in the regime
\beq
r_0\ll \lambda/a~~,~~r_+\ll 1/(a\Omega)~~,~~ r_+\ll\rho_+~~.
\eeq 
Near the boundary, the requirements \eqref{eq:req} are not possible to satisfy, therefore we assume that the blackfold description is valid in the interval $\rho_-+\epsilon\le\rho\le\rho_+-\epsilon$. As this configuration has a limit in which the disc, describing the Myers-Perry black hole is recovered, one expects that the blackfold description of the black helicoids also has a smooth limit when $r_0\to0$.

\subsubsection*{More families of helicoid geometries}
The configuration presented above reduces to that of a singly-spinning Myers-Perry black hole in the limit $\lambda\to0$ and hence one may wonder how to construct other helicoid geometries that capture, in a certain limit, Myers-Perry black holes with several ultraspins. There are in fact at least two ways in which this can be done with an increasing degree of generality:

\begin{itemize}

\item Myers-Perry black holes with several ultraspins can be obtained from a helicoid geometry simply by splitting $\mathbb{R}^{(D-1)}$ into a series of one $\mathbb{R}^{3}$ subspace, where the helicoid with pitch $\lambda/a$ is embedded, and several $\mathbb{R}^{2}$ planes. These geometries have topology $\mathbb{R}\times \mathbb{S}^{(D-3)}$ and reduce to Myers-Perry black holes with several ultraspins in the limit $\lambda\to0$.

\item Alternatively, Myers-Perry black holes can also be obtained by splitting $\mathbb{R}^{(D-1)}$ into a series of $l$ $\mathbb{R}^{3}$ subspaces and embedding a helicoid with pitch $\lambda_a/a_a$ in each of those subspaces. These black holes have topology $\mathbb{R}^{(2l-1)}\times\mathbb{S}^{(D-2l-1)}$ and in the limit $\lambda_a\to0~,~\forall a$ reduce to to Myers-Perry black holes with several ultraspins. In the limit in which we take $\lambda_a\to0$ but keep $\lambda_1\ne0$ this geometry reduces to the previous example.
\end{itemize}
Both of these examples trivially solve the blackfold equations in flat space-time since products of Euclidean minimal surfaces are still Euclidean minimal surfaces. We note that higher-dimensional helicoids ($p$-branes with helicoidal shape), which will be studied in App.~\ref{sec:higherhelicoid}, do not describe these geometries, as in the limit $\lambda\to0$ we recover a minimal cone geometry instead of a $p$-ball.


\subsection{Black discs and $p$-balls in plane wave space-times} \label{sec:discpw}
In this section we construct the analogue black disc configuration of the previous section in plane wave
space-times and their higher-dimensional versions. This will highlight the differences between inherent 
(various kinds of horizons) and non-inherent (induced by rotation) limiting surfaces. Black discs and $p$-balls can be of \textbf{Type I} or \textbf{Type II} and we will analyse both of them.

\subsubsection{Black discs of Type I} \label{sec:disc1}
In this case we have an embedding of the form \eqref{ds:type1} in the ambient space-time \eqref{ds:pw}, which is a solution to the blackfold equations according to theorem \eqref{theo:solpw}. Since these solutions will be rotating then for \eqref{eq:killing} to be a Killing vector field of the background we must choose $A_2=A_1$. This geometry is obtained by choosing the mapping functions
\beq
t=\tau~~,~~y=0~~,~~X^{1}(\rho,\phi)=\rho\cos\phi~~,~~X^{2}(\rho,\phi)=\rho\sin\phi~~,
\eeq
and $X^{i}=0~,~i=3,...,D-2$, leading to the induced wordvolume metric
\beq \label{ds:disc1}
\textbf{ds}^2=-R_0^2d\tau^2+d\rho^2+\rho^2d\phi^2~~,~~R_0^2=1+A_1\rho^2~~.
\eeq
We introduce worldvolume rotation by considering the Killing vector field
\beq \label{rot:p1}
\textbf{k}^{a}\partial_a=\partial_\tau+\Omega\partial_\phi~~,~~\textbf{k}^2=1+(A_1-\Omega^2)\rho^2~~.
\eeq
From this expression we see that a limiting surface is present in the space-time if $A_1-\Omega^2<0$. If $A_1=\Omega^2$ then there is no limiting surface and the disc is non-compact. If $A_1-\Omega^2<0$ then the disc is cut at $\rho_+=\sqrt{(\Omega^2-A_1)}^{-1}$. Hence we must have that $\Omega^2>A_1$. It is worth mentioning that if $A_1<0$ there exists a compact static solution with $\Omega=0$. 

The free energy for these configurations is obtained by integrating the general free energy \eqref{eq:free} to leading order such that 
\beq \label{eq:freeplane1}
 \begin{split}
 \mathcal{F}&=\frac{\Omega_{(n+1)}}{16\pi G}r_+^{n}\int_{0}^{2\pi} d\phi\int_{0}^{\rho_+} d\rho R_0\rho\left(1+(A_1-\Omega^2)\rho^2\right)^{\frac{n}{2}}\\
&= \frac{\Omega_{(n+1)}}{8 G}r_+^{n}\frac{\, _ 2F_ 1\left(-\frac{1}{2},1;\frac{n}{2}+2;\, \frac{A_1}{A_1-\Omega ^2}\right)}{(n+2) \left|\Omega ^2-A_1 \right|}~~,
 \end{split}
 \eeq
where $r_+=n/(4\pi T)$. These configurations connect to the singly spinning Myers-Perry black hole in flat space-time analysed in the previous section when sending $A_1\to0$. Clearly, when $\Omega^2=A_1$ the free energy diverges as the disc becomes non-compact. This geometry, valid in the regime $r_+\ll\sqrt{A_1}^{-1}$ and $r_+\ll\rho_+$ according to App.~\ref{sec:valanal}, has topology $\mathbb{S}^{(D-2)}$ where the size of the transverse sphere is given by
\beq \label{thick:p1}
r_0=r_+\sqrt{1+(A_1-\Omega^2)\rho^2}~~,
\eeq
and hence varies from a maximum value at $\rho=0$ and shrinks to zero at the boundary $\rho_+$. The
generalisation of Myers-Perry black holes in plane wave space-times is not known analytically. The geometries
constructed here should capture the ultraspining regime of such black holes. Note that since we want the plane wave space-time \eqref{ds:pw} to be a vacuum solution we need to require the existence of at least one extra transverse direction $i=3$ where the brane is point-like and located at $x_3=0$ such that $2A_{11}+A_{33}=0$. Therefore these solutions exist in vacuum for $D\ge6$.

\subsubsection*{Thermodynamics}
The thermodynamic properties of these black discs can be obtained from \eqref{eq:freeplane1} using \eqref{eq:thermo}. The mass and angular momentum read
\beq
\begin{split}
M=\frac{\Omega_{(n+1)}r_+^n}{8 G}\Big(&\frac{2 \Omega ^2 \, _2F_1\left(-\frac{1}{2},2;\frac{n}{2}+2;\frac{A_1}{A_1-\Omega^2}\right)}{(n+2) \left(A_1-\Omega ^2\right)^2}\\
& -\frac{(n+1) \left(A_1-\Omega ^2\right) \,_2F_1\left(-\frac{1}{2},1;\frac{n+4}{2};\frac{A_1}{A_1-\Omega ^2}\right)}{(n+2) \left(A_1-\Omega ^2\right)^2}\Big),
 \end{split}
\eeq
\beq
J=\frac{\Omega_{(n+1)}}{4 G}r_+^{n}\Omega \frac{\, _ 2F_ 1\left(-\frac{1}{2},2;\frac{n}{2}+2;\, \frac{A_1}{A_1-\Omega ^2}\right)}{(n+2) \left(\Omega ^2-A_1 \right)^2}~~.
\eeq
These quantities reduce to those obtained in \cite{Emparan:2009vd} for Myers-Perry black holes when $A_1=0$. This blackfold solution does not satisfy the Smarr-relation for flat space-time and hence it has a non-trivial tension given by
\beq
\begin{split}
\!\!\!\!\boldsymbol{\mathcal{T}}=-\frac{\Omega_{(n+1)}r_+^n}{4 G}\frac{ \Omega ^2 \, _2F_1\left(-\frac{1}{2},2;\frac{n}{2}+2;\frac{A_1}{A_1-\Omega^2}\right)+\left(A_1-\Omega ^2\right) \,_2F_1\left(-\frac{1}{2},1;\frac{n+4}{2};\frac{A_1}{A_1-\Omega ^2}\right)}{(n+2) \left(A_1-\Omega ^2\right)^2},
 \end{split}
\eeq
which vanishes in the limit $A_1\to0$.


\subsubsection{Black discs of Type II} \label{sec:disc2}
Black discs of \textbf{Type II} have embeddings of the form \eqref{ds:type2} and are solutions of the blackfold equations according to theorem \eqref{theo:solpw2}. Again, since the geometries we are interested in can be rotating, we need to choose $A_2=A_1$ such that \eqref{eq:killing} is a Killing vector field of the ambient space-time \eqref{ds:pw}. The mapping functions for this geometry are given by
\beq
t=\tau~~,~~y=z~~,~~X^{1}(\rho,\phi)=\rho\cos\phi~~,~~X^{2}(\rho,\phi)=\rho\sin\phi~~,
\eeq
and $X^{i}=0~,~i=3,...,D-2$. This leads to a worldvolume geometry which is itself a plane wave,
\beq \label{ds:disc22}
\textbf{ds}^2=-R_0^2d\tau^2+2(1-R_0^2)d\tau dz+(2-R_0^2)dz^2+d\rho^2+\rho^2d\phi^2~~,~~R_0^2=1+A_1\rho^2~~.
\eeq
We recall that, as mentioned in Sec.~\ref{sec:cembed}, all \textbf{Type II} embeddings are non-compact in the $z$-direction. Rotation is introduced exactly as in \eqref{rot:p1} and hence the discussion of limiting surfaces and boundaries is the same. 

The free energy for these configurations takes a more simple form 
\beq \label{eq:freeplane2}
 \begin{split}
 \mathcal{F}&=\frac{\Omega_{(n+1)}}{16\pi G}r_+^{n}\int dz\int_{0}^{2\pi} d\phi\int_{0}^{\rho_+} d\rho ~\rho\left(1+(A_1-\Omega^2)\rho^2\right)^{\frac{n}{2}}\\
&= \frac{\Omega_{(n+1)}}{8 G}r_+^{n}\int dz\frac{1}{(n+2)|\Omega^2-A_1|}~~,
 \end{split}
 \eeq
where $r_+=n/(4\pi T)$. These configurations connect to the singly spinning Myers-Perry black hole in flat space-time, modulo the integration over $z$, analysed in the previous section when sending $A_1\to0$. If one includes the $z$ direction then the limit $A_1\to0$ gives rise to a Myers-Perry string. We note that this free energy and also its thermodynamic properties are exactly the same, again modulo the integration over $z$, as those for black discs in (Anti)-de Sitter space-times studied in \cite{Armas:2010hz}. This becomes evident if we one identifies $A_1=L^{-2}$ where $L$ is the (Anti)-de Sitter radius.

This geometry, valid in regime $r_+\ll\sqrt{A_1}^{-1}$ and $r_+\ll\rho_+$ according to App.~\ref{sec:valanal}, has topology $\mathbb{R}\times \mathbb{S}^{(D-3)}$ where the size of the transverse sphere is given by \eqref{thick:p1} and hence behaves in the same way as for discs of \textbf{Type I}. These geometries provide evidence for the existence of yet another generalisation of Myers-Perry black holes in plane wave space-times for $D\ge7$, which is not known analytically.

\subsubsection{Black $p$-balls of Type II}
Black $p$-balls are just embeddings of $\mathbb{R}^{(p)}$ into $\mathbb{R}^{(D-1)}$ or $\mathbb{R}^{(D-2)}$ with spherical boundary conditions and hence trivially solve the blackfold equations (see theorem \ref{theo:type3}). If we consider rotation these geometries describe Myers-Perry-type of black holes with several ultraspins where $p$ must be an even number. If the embedding is of \textbf{Type II} then its thermodynamic properties are very similar to the case of rotating black $p$-balls in (Anti)-de Sitter space-times studied in \cite{Armas:2010hz}. We will however focus on static geometries which are valid for all $p$. For simplicity, we take the configuration to be of \textbf{Type II}.

In this case we embed a $p$-ball into the $(D-2)$ Euclidean metric of \eqref{ds:pw} while choosing $X^{i}=0~,~i=p+1,...,D-2$. The induced metric on the worldvolume, of the general form \eqref{ds:type2}, is given by
\beq
\textbf{ds}^2=-R_0^2d\tau^2+2(1-R_0^2)d\tau dz+(2-R_0^2)dz^2+d\rho^2+\rho^2d\Omega_{(p-1)}^2~~,
\eeq
where $R_0^2=1+A_1\rho^2$. We have chosen for simplicity $A_{x_{\hat{a}} x_{\hat{a}}}=A_{x_1x_1}=A_1~,~\hat a=1,...,p$. Since we want the geometry to be compact we assume that $A_1<0$ and hence a limiting surface appears at the boundary $\rho_+=\sqrt{A_1}^{-1}$. These geometries have topology $\mathbb{R}\times \mathbb{S}^{(D-3)}$. The transverse size of the horizon varies according to \eqref{thick:p1}.

The free energy takes the simple form
\beq 
 \begin{split}
 \mathcal{F}&=\frac{\Omega_{(n+1)}}{16\pi G}r_+^{n}\int dz\int d\Omega_{(p-1)}\int_{0}^{\rho_+} d\rho ~\rho^{p-1}\left(1+A_1\rho^2\right)^{\frac{n}{2}}\\
&= \frac{\Omega_{(n+1)}}{16 \pi G}\Omega_{(p-1)}r_+^{n}\int dz\frac{A_1^{-p/2} \Gamma \left(\frac{n}{2}+1\right) \Gamma
   \left(\frac{p}{2}\right)}{2 \Gamma \left(\frac{1}{2} (n+p+2)\right)}~~,
 \end{split}
 \eeq
where $r_+=n/(4\pi T)$ and agrees with the static limit of the disc ($p=2$) when setting $\Omega=0$ in \eqref{eq:freeplane2}. These are the analogous configurations to those of \eqref{eq:dsbde} in plane wave space-times. For even $p$ they describe static rotating black holes in plane wave space-times analogous to those in de Sitter space-time \cite{Armas:2010hz} with equal free energy, modulo the integration over $z$, provided we set $A_1=L^2$. For odd $p$ these do not arise as a limit of Myers-Perry-type black holes and hence hint to the existence of a new family of rotating black holes. These geometries also have a non-vanishing tension given by
\beq
\boldsymbol{\mathcal{T}}=-\frac{\Omega_{(n+1)}}{16 \pi G}\Omega_{(p-1)}r_+^{n}\int dz\frac{n A_1^{-p/2} \Gamma \left(\frac{n}{2}+1\right) \Gamma
   \left(\frac{p}{2}\right)}{ \Gamma \left(\frac{n+p}{2}\right)}~~.
\eeq
The validity analysis follows the same footsteps as for black discs of \textbf{Type II} as in App.~\ref{sec:valanal}. In the end we must just require that $r_+\ll \sqrt{A_1}^{-1}$ and
\beq
r_+\ll \frac{\textbf{k}}{\sqrt{A_1}\sqrt{2+A_1\rho^2}}~~,
\eeq
which leads to $r_+\ll\rho_+$ near $\rho=0$ and requires the introduction of a cut-off $\epsilon$ near the boundary.

 
 
 \subsection{Black helicoids in plane wave space-times} \label{sec:helicoidpw}
 In this section we construct the analogous configurations of the black helicoids of Sec.~\ref{sec:helicoid}. These configurations according to theorem \ref{theo:solpw} and \ref{theo:solpw2} can be of \textbf{Type I} or \textbf{Type II} and connect to the black disc geometries of the previous section in an appropriate limit. They require that $A_2=A_1$ in order to be valid solutions of the blackfold equations. We will also show that as in the case of the discs, inherent limiting surfaces allow for static helicoid configurations.

 \subsubsection{Black helicoids of Type I} \label{sec:helicoid1}
 Helicoids of \textbf{Type I} are embeddings of the general form \eqref{ds:type1} into plane wave space-times \eqref{ds:pw}. They are described by the mapping functions
\beq
t=\tau~~,~~y=0~~,~~X^{1}(\rho,\phi)=\rho\cos(a\phi)~~,~~X^{2}(\rho,\phi)=\rho\sin(a\phi)~~,~~ X^{3}(\rho,\phi)=\lambda\phi~,
\eeq
and $X^{i}=0~,~i=4,...,D-1$, where $a$ is a constant which we take to be positive without loss of generality. Similarly we take $\lambda\ge0$. The ratio $\lambda/a$ is the pitch of the helicoid.
The induced worldvolume metric takes the form
\beq
\textbf{ds}^2=-R_0^2d\tau^2+d\rho^2+(\lambda^2+a^2\rho^2)d\phi^2~~,~~R_0^2=1+A_1\rho^2~~.
\eeq
This geometry reduces to the case of the disc \eqref{ds:disc1} when $a=1$ and $\lambda=0$. The helicoid is boosted along the $\phi$ direction with boost velocity $\Omega$ such that
\beq
\textbf{k}^{a}\partial_a=\partial_\tau+\Omega\partial_\phi~~,~~\textbf{k}^2=1+A_1\rho^2-\Omega^2(\lambda^2+a^2\rho^2)~~.
\eeq
From the expression for $\textbf{k}$ we see that for the solution to be valid at $\rho=0$ we need that $\Omega^2\lambda^2<1$.\footnote{In the strict limit $\Omega\lambda=1$ all thermodynamic quantities vanish so we do not consider it.} Furthermore, a limiting surface exists whenever $A_1\rho^2-\Omega^2(\lambda^2+a^2\rho^2)<0$ in which case the helicoid is bounded in the $\rho$ direction and has boundaries at 
\beq
\rho_\pm=\pm\sqrt{\frac{\Omega^2\lambda^2-1}{A_1-a^2\Omega^2}}~~,
\eeq
where one must require that $A_1-a^2\Omega^2<0$. Note in particular that if $A_1=a^2\Omega^2$ there is no limiting surface. The limiting surface constrains the $\rho$ direction but not the $\phi$ direction, therefore these geometries are non-compact along $\phi$. As for the case of discs of \textbf{Type I} a static solution $\Omega=0$ exists provided $A_1<0$. These black hole horizons have topology $\mathbb{R}\times \mathbb{S}^{(D-3)}$ where the size of the transverse sphere $r_0(\rho)$ is given by
\beq
r_0(\rho)=\frac{n}{4\pi T}\sqrt{1+A_1\rho^2-\Omega^2(\lambda^2+a^2\rho^2)}~~,
\eeq
and attains its maximum value at the origin $\rho=0$ and vanishes at the boundaries $\rho_\pm$. These geometries, according to App.~\ref{sec:valanal}, exist in plane wave backgrounds in vacuum for $D\ge6$ in the regime 
\beq
r_0\ll \sqrt{A_1+a^2/\lambda^2}^{-1}~~,~~ r_+\ll\sqrt{A_1-a^2\Omega^2}^{-1}~~,~~r_+\ll\rho_+~~.
\eeq

\subsubsection*{Free energy}
The free energy can be obtained from \eqref{eq:free} to leading order and reads
\beq \label{free:helic1}
 \begin{split}
 \mathcal{F}&=\frac{\Omega_{(n+1)}}{16\pi G}r_+^{n}\int d\phi\int_{\rho_-}^{\rho_+} d\rho ~R_0\sqrt{\lambda^2+a^2\Omega^2}\left(1+A_1\rho^2-\Omega^2(\lambda^2+a^2\rho^2)\right)^{\frac{n}{2}}~~,
 \end{split}
 \eeq
where $r_+=n/(4\pi T)$. It is not possible to integrate this expression and obtain a closed analytical form. However, it can be easily be done numerically. In the high pitch limit $\lambda\gg1$, for example, we can obtain an approximate expression up to order $\mathcal{O}(\lambda^{-1})$,
\beq
\mathcal{F}=\frac{\Omega_{(n+1)}}{16\pi G}r_+^{n}\int d\phi\frac{\sqrt{\pi } \lambda  \Gamma \left(\frac{n}{2}+1\right) \left(1-\lambda ^2 \Omega
   ^2\right)^{\frac{n+1}{2}} \,
   _2\tilde{F}_1\left(-\frac{1}{2},\frac{1}{2};\frac{n+3}{2};\frac{A_1 \left(\lambda ^2 \Omega
   ^2-1\right)}{a^2 \Omega ^2}\right)}{a \Omega }~~.
\eeq
In the limit $\lambda\to0$ the free energy \eqref{free:helic1} reduces to the free energy of the disc by simultaneously rescaling $\mathcal{F}\to(1/2)\mathcal{F}$ due to the double covering of the coordinate $\rho$, integrating $\phi$ over the interval $0\le\phi\le2\pi/a$ and setting $a=1$. Furthermore, in the limit $A_1\to0$ it reduces to the free energy of the helicoid in flat space-time \eqref{eq:freehelicoidflat}. The remaining thermodynamic properties can be obtained using \eqref{eq:thermo}.

\subsubsection{Black helicoids of Type II} \label{sec:helicoid2}
We now turn our attention to black helicoid geometries of \textbf{Type II}. These are described by an embedding geometry of the form \eqref{ds:type2}. The embedding map is given by
\beq
t=\tau~~,~~y=z~~,~~X^{1}(\rho,\phi)=\rho\cos(a\phi)~~,~~X^{2}(\rho,\phi)=\rho\sin(a\phi)~~,~~ X^{3}(\rho,\phi)=\lambda\phi~,
\eeq
and $X^{i}=0~,~i=4,...,D-2$. The geometry is therefore non-compact in the $z$ direction. The induced metric on
the worldvolume takes the form of a non-planar-fronted wave
\beq
\textbf{ds}^2=-R_0^2d\tau^2+2(1-R_0^2)d\tau dz+(2-R_0^2)dz^2+d\rho^2+(\lambda^2+a^2\rho^2)d\phi^2~~,
\eeq
where $R_0^2=1+A_1\rho^2$. The limiting surface and its boundaries are the same as for the helicoids of \textbf{Type I} studied above. These black holes have horizon topology $\mathbb{R}^2\times \mathbb{S}^{(D-4)}$ and are valid solutions in plane wave backgrounds in vacuum for $D\ge7$ in the regime  $r_0\ll \lambda/a$,  $r_+\ll \sqrt{A_1(1-a^2\Omega^2)}^{-1}$ and $r_+\ll\rho_+$ according to App.~\ref{sec:valanal}.

The free energy for these configurations is
\beq \label{eq:freehelc2}
 \begin{split}
 \mathcal{F}&=\frac{\Omega_{(n+1)}}{16\pi G}r_+^{n}\int dz\int d\phi \int_{\rho_-}^{\rho_+} d\rho ~\sqrt{\lambda^2+a^2\rho^2}\left(1+A_1\rho^2-\Omega^2(\lambda^2+a^2\rho^2)\right)^{\frac{n}{2}}\\
&= \frac{\Omega_{(n+1)}}{16\pi G}r_+^{n}\int dz\int d\phi \frac{\sqrt{\pi } \lambda  \Gamma \left(\frac{n}{2}+1\right) \left(1-\lambda ^2
   \Omega ^2\right)^{\frac{n+1}{2}} \,
   _2\tilde{F}_1\left(-\frac{1}{2},\frac{1}{2};\frac{n+3}{2};\frac{a^2
   \left(1-\lambda ^2 \Omega ^2\right)}{\lambda ^2 \left(A_1-a^2 \Omega^2
   \right)}\right)}{\sqrt{|a^2 \Omega ^2-A_1|}}~~,
 \end{split}
 \eeq
and reduces to the free energy of the disc of \textbf{Type II} \eqref{eq:freeplane2} when $\lambda\to0$, after rescaling $\mathcal{F}\to(1/2)\mathcal{F}$, integrating $\phi$ over the interval $0\le\phi\le2\pi/a$ and setting $a=1$. In the static case $\Omega=0$ we need $A_1<0$ for the geometry to be compact. In this case the free energy reduces to
\beq
\mathcal{F}=\frac{\Omega_{(n+1)}}{16\pi G}r_+^{n}\int dz\int d\phi \frac{\sqrt{\pi } \lambda  \Gamma \left(\frac{n+2}{2}\right) \,
   _2F_1\left(-\frac{1}{2},\frac{1}{2};\frac{n+3}{2};\frac{a^2}{A_1 \lambda
   ^2}\right)}{\sqrt{|A_1|} \Gamma \left(\frac{n+3}{2}\right)}~~.
\eeq
Their thermodynamic properties can be obtained using \eqref{eq:thermo} and lead to very cumbersome expressions.
 
\subsection{Black catenoids and black Scherk surfaces in plane wave space-times} \label{sec:catenoid}
In this section we construct black catenoids and black Scherk surfaces in asymptotically plane wave space-times. These embeddings are \textbf{Type II} embeddings and they solve the blackfold equations, according to theorem \ref{theo:solpw2}, for a specific equilibrium condition between the angular velocity and the plane wave matrix $A_{qr}$ components. We begin by studying the catenoids and then move on to the slightly more complicated case of Scherk surfaces. 

\subsubsection{Black catenoids of Type II}
The catenoid is the only non-trivial minimal surface of revolution in $\mathbb{R}^{3}$ (the trivial one being the plane). Its embedding can be parametrised by
\beq
X^{1}(\rho,\phi)=c\cosh\left(\frac{\rho}{c}\right)\sin(a\phi)~~,~~X^{2}(\rho,\phi)=c\cosh\left(\frac{\rho}{c}\right)\cos(a\phi)~~,~~X^{3}(\rho,\phi)=\rho~~,
\eeq
and $t=\tau~,~y=z$ as well as $X^{i}=0~,~i=4,...,D-2$. The constants $a$ and $c$ can be chosen to be positive, without loss of generality. The coordinates range between the intervals $\rho\ge0$ and $0\le\phi\le2\pi/a$.
The induced metric, of the general form \eqref{ds:type2}, reads
\beq \label{ds:catenoid}
\textbf{ds}^2=-R_0^2d\tau^2+2(1-R_0^2)d\tau dz+(2-R_0^2)dz^2+\cosh^2\left(\frac{\rho}{c}\right)\left(d\rho^2+a^2c^2d\phi^2\right)~~,
\eeq
where $R_0^2=1+c^2A_{1}\cosh^2(\rho/c)$. Therefore these geometries are non-compact in the $z$ direction. We have also made  the choice $A_{x_1x_1}=A_{x_2x_2}=A_{1}$ and $A_{x_3x_3}=0$ since we need the catenoid to rotate and hence we require that \eqref{eq:killing} is a Killing vector field of the ambient space-time. The catenoid is rotating with angular velocity $\Omega$ such that
\beq
\textbf{k}^{a}\partial_a=\partial_\tau+\Omega\partial_\phi~~,~~\textbf{k}^2=1+c^2(A_1-a^2\Omega^2)\cosh^2\left(\frac{\rho}{c}\right)~~.
\eeq
This corresponds to a background Killing vector field of the form \eqref{eq:killing} rotating with angular
velocity $a\Omega$ in the $(x_1,x_2)$ plane and not boosted in the $x_3$ direction. In general we see that if
$A_1-a^2\Omega^2<0$ a limiting surface appears constraining the coordinate $\rho$. However the solution to the
blackfold equations \eqref{sol:cat2} requires that $A_1=a^2\Omega^2$ and hence that $\textbf{k}=1$. Therefore
we see that these catenoid geometries are rotating but are non-compact in the $\rho$ direction. These geometries give rise to black hole horizon topologies of the form $\mathbb{R}^3\times\mathbb{S}^{(D-5)}$ and exist in vacuum for $D\ge7$. These are depicted in  Fig.~\ref{fig:catenoid}.
\begin{figure}[h!] 
\centering
  \includegraphics[width=0.35\linewidth]{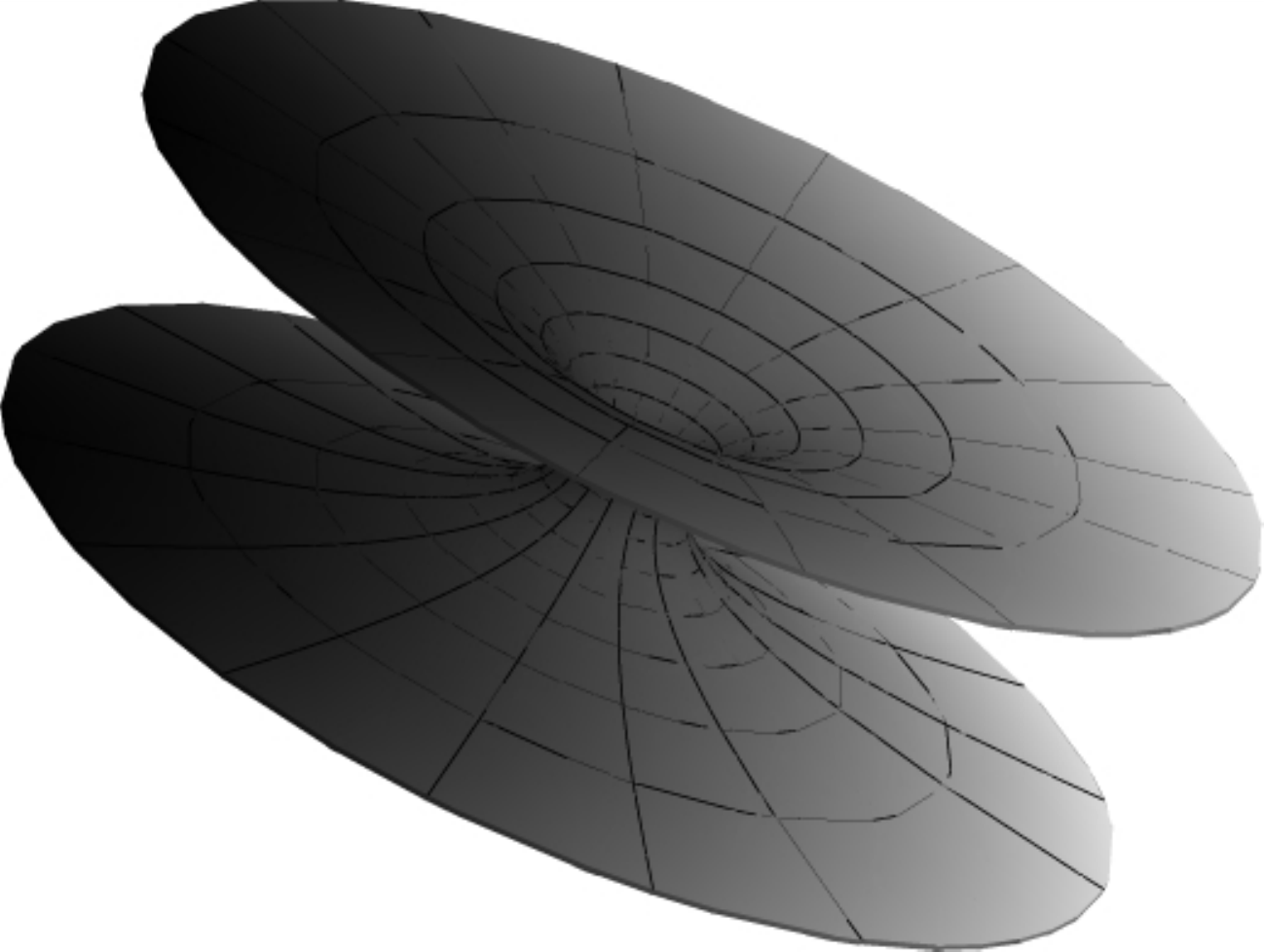}
\caption{Embedding of the rotating black catenoid in $\mathbb{R}^{3}$ with $a=c=1$, depicted in the interval $-3\le\rho\le3$. } \label{fig:catenoid}
\end{figure}

\subsubsection*{Thermodynamics and validity analysis}
The free energy of these configurations takes the following form
\beq \label{free:cat2}
\begin{split}
\mathcal{F}&=c\thinspace a\frac{\Omega_{(n+1)}}{16\pi G}\int dz \int d\rho \int_{0}^{\frac{2\pi}{a}} d\phi \cosh^{2}\left(\frac{\rho}{c}\right) r_+^{n} \left(1+c^2(A_1-a^2\Omega^2)\cosh^2\left(\frac{\rho}{c}\right)\right)^{\frac{n}{2}}\\
&=c\frac{\Omega_{(n+1)}}{8 G}\int dz \int d\rho \cosh^{2}\left(\frac{\rho}{c}\right)  r_+^{n}~~,
\end{split}
\eeq
where $r_+=\frac{n}{4\pi T}$. Since the configuration is non-compact, it is only meaningful to talk about the free energy density, i.e., the free energy modulo the integrations over $z$ and $\rho$. Since the catenoid is rotating, it has an angular momentum given by
\beq\label{free:j2}
J=c^3a^2n\Omega \frac{\Omega_{(n+1)}}{8 G}\int dz \int d\rho \cosh^{4}\left(\frac{\rho}{c}\right)  r_+^{n}~~.
\eeq

We now turn our attention to the validity of these configurations. Since this is a \textbf{Type II} embedding we need to check the invariants $\mathcal{R}$ and $u^{a}u^{b}\mathcal{R}_{ab}$. Using \eqref{r:type2} we find the requirements 
\beq 
r_0\ll\frac{c}{\sqrt{2}}\cosh^2\left(\frac{\rho}{c}\right)~~,~~r_0\ll\textbf{k}\frac{\cosh\left(\frac{\rho}{c}\right)}{\sqrt{A_1\cosh\left(\frac{2\rho}{c}\right)-a^2\Omega^2}}~~.
\eeq
The first invariants take its minimum value when $\rho=0$ while the second takes its minimum value when $\rho\to\infty$. Therefore we only have to require $r_0\ll c$ and $r_+\ll\sqrt{A_1}^{-1}$.
\subsubsection{Black Scherk surfaces of Type II}
In this section we study rotating black Scherk surfaces which unify the black catenoid, the black helicoid and the black disc of \textbf{Type II}. This family of solutions has the geometry of the associate family of the helicoid and the catenoid. The form of its embedding was already given in \eqref{e:iso}, explicitly, we have that
\beq
\begin{split}
&X^{1}(\rho,\phi)=\tilde \lambda c\sinh\left(\frac{\rho}{c}\right)\sin(a\phi)+\sqrt{1-\tilde \lambda^2}c\cosh\left(\frac{\rho}{c}\right)\cos(a\phi)~~,\\
&X^{2}(\rho,\phi)=-\tilde \lambda c\sinh\left(\frac{\rho}{c}\right)\cos(a\phi)+\sqrt{1-\tilde \lambda^2} c\cosh\left(\frac{\rho}{c}\right)\sin(a\phi)~~,\\
&X^{3}(\rho,\phi)=\tilde \lambda a c \phi + \sqrt{1-\tilde \lambda^2} \rho~~,
\end{split}
\eeq
and $t=\tau~,~y=z$ as well as $X^{i}=0~,~i=4,...,D-2$. If we set $\tilde\lambda=0$ we recover the catenoid geometry studied in the previous section, while if we set $\tilde\lambda=1$, redefine $ac=\lambda$ in $X^{3}(\rho,\phi)$ and introduce a new $\tilde\rho$ coordinate such that $c\sinh(\rho/c)=\tilde \rho$ we recover the helicoid in the form used  in Sec.~\ref{sec:helicoidpw}. The induced metric on the worldvolume takes exactly the same form as for the catenoid \eqref{ds:catenoid} since the spatial part is $\tilde\lambda$-independent. However, $R_0$ is given instead by
\beq
R_0^2=1+\frac{A_1}{2} c^2 \left(1-2\tilde\lambda^2+ \cosh \left(\frac{2 \rho}{c}\right)\right)~~.
\eeq
We have made the choice $A_{x_1x_1}=A_{x_2x_2}=A_{1}$ and $A_{x_3x_3}=0$ so that the ambient space-time has a one-family group of isometries associated with rotations in the $(x_1,x_2)$ plane. The Scherk surface is boosted along the $\phi$ direction with boost velocity $\Omega$ such that
\beq \label{k:helicat}
\textbf{k}^{a}\partial_a=\partial_\tau+\Omega\partial_\phi~~,~~\textbf{k}^2=R_0^2-\Omega^2a^2c^2\cosh^2\left(\frac{\rho}{c}\right)~~,
\eeq
with corresponds to a Killing vector field of the form \eqref{eq:killing} with angular velocity $a\Omega$ in the $(x_1,x_2)$ plane and with boost velocity $\tilde \lambda a\Omega$ along the $x_3$ direction. In general there is a limiting surface in the space-time, however, as shown in theorem \ref{theo:solpw2} the solution to the equations of motion requires $A_1=a^2\Omega^2$ and hence one finds that the limiting surface is removed since $\textbf{k}$ takes the constant value
\beq
\textbf{k}^2=1-a^2c^2\tilde\lambda^2\Omega^2~~.
\eeq
In the case of the catenoid of the previous section ($\tilde\lambda=0$) we recover the result $\textbf{k}=1$. We also see that in order to have a valid configuration we need to require $a^2c^2\tilde\lambda^2\Omega^2<1$. In the case of the helicoid, where $\tilde\lambda=1$ and $\lambda=ac$, we recover the result $\Omega^2\lambda^2<1$ obtained in Sec.~\ref{sec:helicoidpw}. These configurations give rise to black hole horizon topologies of the form $\mathbb{R}^3\times\mathbb{S}^{(D-5)}$ for all $\tilde\lambda$ in $D\ge7$. However, the geometry of the horizon varies greatly with $\tilde\lambda$ from that of a helicoid to that of the catenoid as depicted in Fig.~\ref{fig:helicat}.
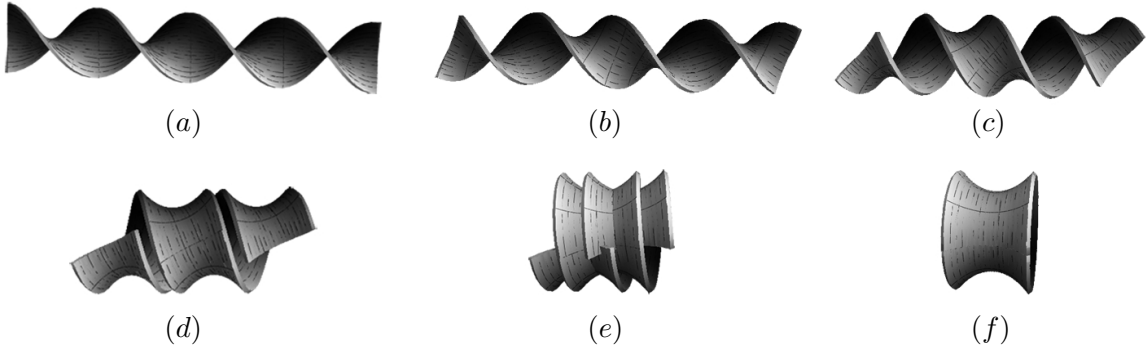
\begin{figure}[h!] 
\centering
  \includegraphics[width=1\linewidth]{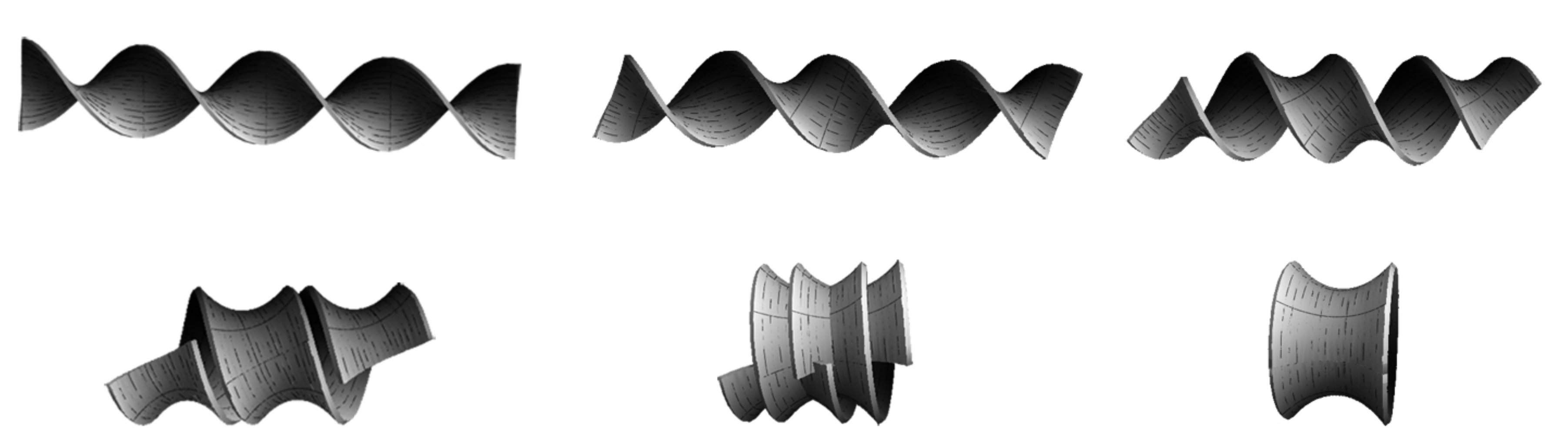}
  \begin{picture}(0,0)(0,0)
\put(-160,78){ $ (a) $}
\put(-0,78){ $ (b) $}
\put(145,78){ $ (c) $}
\put(-160,0){ $ (d) $}
\put(-0,0){ $ (e) $}
\put(145,0){ $ (f) $}
\end{picture} 
\caption{Embedding of the rotating black Scherk surface in $\mathbb{R}^{3}$ with $a=c=1$, depicted in the interval $-1\le\rho\le1$ and $-2\pi\le\phi\le2\pi$. The images show the deformation of the horizon geometry as the parameter $\tilde \lambda$ is changed. Image (a) corresponds to the helicoid $\tilde \lambda=1$, (b) to $\tilde \lambda=0.7$, (c) to $\tilde \lambda=0.5$, (d) to $\tilde \lambda=0.3$, (e) to $\tilde\lambda=0.1$ and (f) to the catenoid $\tilde\lambda=0$.} \label{fig:helicat}
\end{figure}

Because the induced metric for these geometries takes the same form as for the black catenoids \eqref{ds:catenoid}, the free energy \eqref{eq:free} to leading order also takes the same form and hence all its thermodynamic properties are the same as those presented in \eqref{free:cat2} and \eqref{free:j2}.

\subsubsection*{Validity analysis}
The analysis of the validity of these configurations is very similar to the case of the catenoids. In particular because its a \textbf{Type II} embedding then we have that the worldvolume Ricci scalar is just given by the Ricci scalar of the spatial part of the metric due to \eqref{r:type2}. Since the spatial part of the metric does not differ from \eqref{ds:catenoid} we obtain again the requirement $r_0\ll c$. From the invariant $u^{a}u^{b}\mathcal{R}_{ab}$ we simply get that $r_+\ll \sqrt{A_1}^{-1}$ while the invariant $\textbf{k}^{-1}\nabla_a\nabla^a\textbf{k}$ vanishes since $\textbf{k}$ is constant along the worldvolume.

From the point of view of the validity requirements in \eqref{eq:req} these configurations are valid blackfold solutions leading to regular horizons. However, the geometry of Scherk's second surface has self-intersections for any value of $\tilde\lambda$ which is neither the helicoid ($\tilde\lambda=1$) nor the catenoid ($\tilde\lambda=0$) (see e.g. \cite{Gray:2006}). This is visible in the image (e) of Fig.~\ref{fig:helicat}. Therefore in the strict mathematical sense, these geometries are not \emph{embedded submanifolds}. 

As explained in Sec.~\ref{sec:validity}, the set of requirements \eqref{eq:req} were obtained assuming that curvature corrections were dominant over backreaction and self-force effects. However, in the case of geometries with self-intersections, self-force effects, near the location where the horizon meets itself, are expected to dominate over curvature corrections. In particular, the existence of self-intersections means that \eqref{r:d} is not satisfied. It is still an open question of whether or not blackfold worldvolumes with self-intersections can give rise to regular black hole solutions. A more in-depth analysis of these cases, perhaps using methods similar in spirit to those of \cite{Emparan:2011ve} or by explicitly constructing the perturbative solution as in \cite{Emparan:2007wm}, would be required in order to assess its validity.


\subsection{Black $p+2$-balls and minimal surfaces in $\mathbb{S}^{(p+1)}$ in de Sitter space-times} \label{sec:desitter}
In Sec.~\ref{sec:theorems} we showed that minimal surfaces on $\mathbb{S}^{(p+1)}$ solve the blackfold equations, which led to theorem \ref{theo:type3}. The purpose of this section is to show that minimal surfaces on $\mathbb{S}^{(p+1)}$ can be useful for constructing black hole horizons in de Sitter space-time. Due to the validity issues of embeddings of \textbf{Type III}, discussed in Sec.~\ref{sec:cembed}, these geometries must be embedded in an ambient space-time with a black hole horizon in order to avoid conical singularities at the origin. The starting point of theorem \ref{theo:type3} is the existence of a $p+2$-ball solution in the space-times \eqref{ds:ds1}. These geometries are obtained by choosing $t=\tau$ and embedding a $p+2$-ball into the conformally Euclidean part of the metric \eqref{ds:ds2}, giving rise to a $(p+2)$-dimensional worldvolume geometry of the general form \eqref{ds:type3} which reads
\beq \label{ds:ball}
\textbf{ds}^2=-R_0^2dt^2+R_0^{-2}d\rho^2+\rho^2d\Omega_{(p+1)}^2~~,~~R_0^2=1-\frac{r_m^{n+p+2}}{\rho^{n+p+2}}-\frac{\rho^2}{L^2}~~.
\eeq
This metric is not a de Sitter metric any longer and its Ricci curvature has a singularity at $r=0$. However, if the black hole is present $r_m\ne0$ then this singularity is shielded behind the black hole horizon, which is located at the lowest positive real root of $R_0=0$. The space-time is defined in the coordinate range $\rho_{-}\le \rho \le \rho_+$ where $\rho_\pm$ denote the two positive real roots of $R_0=0$ describing the location of the black hole horizon and the location of the cosmological horizon. That is, the $p+2$-ball is a compact geometry. Its topology is $\mathbb{S}^{1}\times \mathbb{S}^{(D-3)}$ due to the existence of the inner black hole horizon. Since this geometry is static one has that $\textbf{k}=R_0$. Therefore the free energy is then
\beq \label{free:cliff}
\mathcal{F}=\frac{\Omega_{(n+1)}}{16\pi G}r_+^{n}\int d\Omega_{(p)}\int_{\rho_-}^{\rho_+}\rho^{p+1}~\textbf{k}^{n}~~,
\eeq
where $r_+=n/(4\pi T)$. This free energy cannot be integrated to a closed expression for $n+p+2\ge5$ (even in the limit $L\to\infty$) but it can easily be evaluated numerically. These geometries are just the analogue geometries of the $p$-balls constructed in Sec.~\ref{sec:discpw}. Theorem \ref{theo:type3} says that we can place any minimal surface in $\mathbb{S}^{(p+1)}$ and that will be a solution of the blackfold equations. We will now analyse the case of the Clifford torus and then its higher-dimensional version.

\subsubsection{Clifford torus}
We begin with the classical example of the Clifford torus which is a minimal surface in $\mathbb{S}^{3}$.\footnote{The simplest example of a minimal surface in $\mathbb{S}^{3}$ is the equator of the 3-sphere which is itself a 2-sphere. However, this configuration is already included in \eqref{ds:ball} when $p=1$.} In order to embed it we consider the metric on $S^{3}$ written in the form
\beq
d\Omega_{(3)}^2=d\theta^2+\sin^2\theta d\phi_1^2+\cos^2\theta d\phi^2~~.
\eeq
We embed the Clifford torus ($p=2$) by choosing a constant angle $\sin^2\theta=R^2$. The free energy \eqref{free:cliff} becomes
\beq
\mathcal{F}[R]=\frac{\Omega_{(n+1)}}{16\pi G}r_+^{n}(2\pi)^2\int_{\rho_-}^{\rho_+}\rho^{p}\left(1-\frac{r_m^{n+p+1}}{\rho^{n+p+1}}-\frac{\rho^2}{L^2}\right)^{\frac{n}{2}}R\sqrt{1-R^2}~~.
\eeq
Varying this free energy with respect to $R$ leads to the unique solution defining the Clifford torus $R^2=1/2$, for which each circle has equal radius. Therefore the induced metric of the $(p+2)$-dimensional geometry takes the form
\beq \label{ds:cliff}
\textbf{ds}^2=-R_0^2dt^2+R_0^{-2}d\rho^2+\frac{\rho^2}{2}\left(d\phi_1^2+d\phi_2^2\right)~~,~~R_0^2=1-\frac{r_m^{n+p+1}}{\rho^{n+p+1}}-\frac{\rho^2}{L^2}~~.
\eeq
These geometries give rise to black hole horizon topologies of the form $\mathbb{S}^{1}\times\mathbb{T}^{2}\times\mathbb{S}^{(D-5)}$, where $\mathbb{T}^{2}$ is the torus. This geometry is depicted in Fig.~\ref{fig:clifford}.
\begin{figure}[h!] 
\centering
  \includegraphics[width=0.4\linewidth]{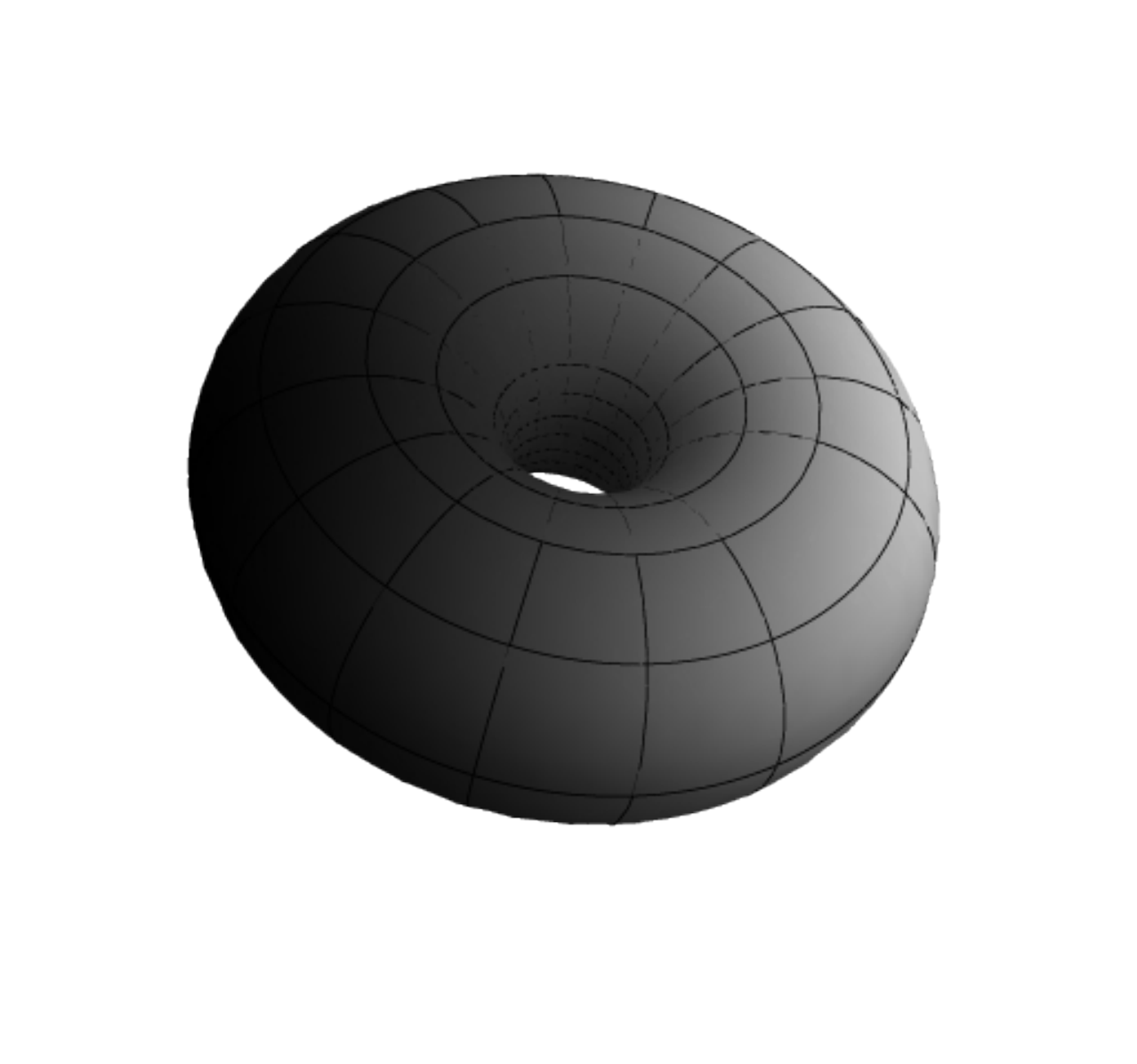}
\caption{The Clifford torus visualised in $\mathbb{R}^{3}$ by applying a stereographic projection from $\mathbb{S}^{3}$ to $\mathbb{R}^{3}$. These are constant $\tau$ and $\rho$ slices of the geometry \eqref{ds:cliff}.} \label{fig:clifford}
\end{figure}

\subsubsection{Higher-dimensional Clifford tori}
The previous configurations can be generalised to higher dimensions. We write the metric on the $(p+1)$-sphere as
\beq \label{ds:sphere}
d\Omega^{2}_{(p+1)}=d\theta^2+\cos^2\theta d\Omega_{(\kappa_1)}^2+\sin^2\theta d\Omega_{(\kappa_2)}^2~~,~~\kappa_1+\kappa_2=p~~,
\eeq
and choose the embedding $\sin^2\theta=R^2$. The free energy \eqref{free:cliff} becomes 
\beq
\mathcal{F}[R]=\frac{\Omega_{(n+1)}}{16\pi G}r_+^{n}\int d\mathbb{T}_{(p)}\int_{\rho_-}^{\rho_+}\rho^{p}\left(1-\frac{r_m^{n+p+1}}{\rho^{n+p+1}}-\frac{\rho^2}{L^2}\right)^{\frac{n}{2}}R^{\kappa_1}\left(1-R^2\right)^{\frac{\kappa_2}{2}}~~,
\eeq
where $d\mathbb{T}_{(p)}$ denotes the volume form on the $p$-dimensional torus. Varying this free energy with respect to $R$ leads to a unique solution given by
\beq \label{eq:mintori}
R^2=\frac{\kappa_1}{\kappa_1+\kappa_2}~~.
\eeq
Therefore the induced metric becomes
\beq
\textbf{ds}^2=-R_0^2dt^2+R_0^{-2}d\rho^2+\rho^2\left(\frac{\kappa_1}{\kappa_1+\kappa_2}d\Omega_{(\kappa_1)}^2+\frac{\kappa_2}{\kappa_1+\kappa_2}d\Omega_{(\kappa_2)}^2\right)~~,
\eeq
where $R_0$ is given in \eqref{ds:cliff}. We see that in general the two spherical parts of the geometry have unequal radii. In the case where the two radii are equal, that is, when $\kappa_1=\kappa_2$ these geometries are known as higher-dimensional Clifford tori. These configurations give rise to static black hole horizon topologies of the form $\mathbb{S}^{1}\times\mathbb{T}^{(p)}\times\mathbb{S}^{(D-p-3)}$ and are valid, according to App.~\ref{sec:valanal}, in there regime $r_0\ll L$, $r_0\ll r_m$ and $r_+\ll \rho_{\pm}$.

The on-shell free energy \eqref{free:cliff} becomes
\beq\label{f:c22}
\mathcal{F}=\frac{\Omega_{(n+1)}}{16\pi G}r_+^{n}\mathbb{T}_{(p)}\ \int_{\rho_-}^{\rho^+} d\rho ~\rho^{p}\left(1-\frac{r_m^{n+p+1}}{\rho^{n+p+1}}-\frac{\rho^2}{L^2}\right)^{\frac{n}{2}}~~,
\eeq
where $\mathbb{T}_{(p)}$ is the volume of the $p$-dimensional torus, given by
\beq
\mathbb{T}_{(p)}=\Omega_{(\kappa_1)}\Omega_{(\kappa_2)}\left(\frac{\kappa_1}{\kappa_1+\kappa_2}\right)^{\frac{\kappa_1}{2}}\left(\frac{\kappa_2}{\kappa_1+\kappa_2}\right)^{\frac{\kappa_2}{2}}~~.
\eeq
Since we cannot find a closed form for the free energy \eqref{f:c22}, we integrate it numerically for several values of $p$ and $n$, the result is given in Fig.~\ref{fig:freeclifford}.
\begin{figure}[h!] 
\centering
  \includegraphics[width=0.5\linewidth]{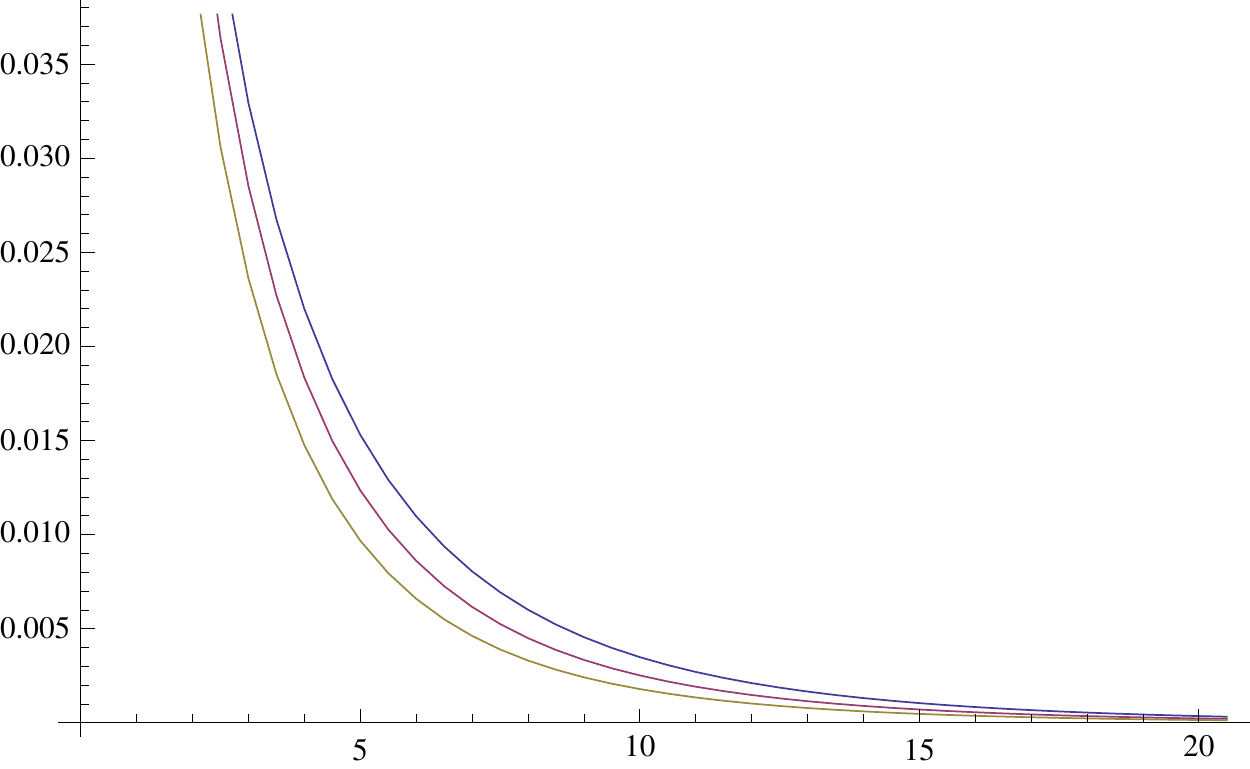}
  \begin{picture}(0,0)(0,0)
\put(-300,113){ $ \frac{16\pi G \mathcal{F}}{r_+^n\Omega_{(n+1)}\mathbb{T}_{(p)}}  $}
\put(-20,-5){ $ n  $}
\end{picture} 
\caption{Free energy for the Clifford Tori with $r_m=\frac{1}{2}$, $L=1$ and $p=2$ (blue curve), $p=3$ (red curve) and $p=4$ (yellow curve).} \label{fig:freeclifford}
\end{figure}
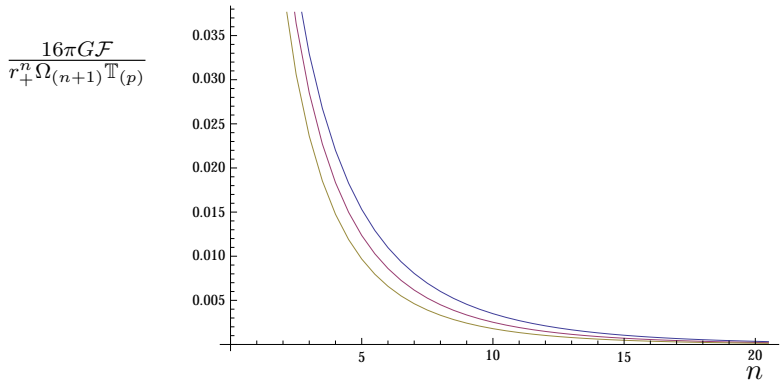
These configurations are also valid the limit $L\to\infty$, however, in that case they are not compact as $\rho_+\to\infty$.

\section{Discussion}\label{sec:conclusions}  
In this paper we have constructed a series of new blackfold
configurations which can give rise to interesting novel black hole
horizon geometries and topologies in asymptotically flat, plane wave and
de Sitter space-times. These blackfold geometries are such that they
intersect limiting surfaces  in the ambient space-time, which are either
inherently present in the space-time or introduced via rotation. The
key ingredient in this work was the recursive use of minimal surfaces in
$\mathbb{R}^{(D)}$ and in $\mathbb{S}^{(D)}$. The presence of limiting
surfaces allowed us to turn several of the non-compact minimal surfaces
in $\mathbb{R}^{(D)}$ into compact minimal surfaces, at least in some
directions, in Lorentzian space-times.

Non-trivial minimal surfaces in $\mathbb{L}^{3}$ are also available in the literature (see e.g.\cite{Mira:2003, LeeS:2008, Wook:2011} for Lorentzian helicoids and catenoids). However, these are embedded in a non-trivial way in the time-like direction and hence most of these surfaces do not preserve a one parameter family of isometries of the ambient space-time. According to the blackfold method explained in Sec.~\ref{sec:validity}, these are not suitable for constructing black hole geometries. 

We have highlighted the fact that there exists the potential for a fruitful interplay between the mathematics of minimal surfaces and black hole geometries. In particular, we were able to show that minimal surfaces such as planes, helicoids, catenoids and Clifford tori give rise to black hole horizons. However, we have only scratched the surface since the mathematics of minimal surfaces is a very active and productive subject of mathematics and, in particular, few examples of higher-dimensional minimal surfaces are known.

We believe that in Sec.~\ref{sec:validity} we have made an important contribution to the blackfold method.
Namely, we have given a prescription for systematically 
analysing the regime of validity of a blackfold configuration based on
the second order effective free energy \eqref{eq:free} obtained in \cite{Armas:2013hsa}. This allows for a
classification of all the length scales associated with the geometric invariants that characterise each
blackfold configuration, and are defined order-by-order in a derivative expansion. Therefore at each order in
the expansion, higher-order invariants must be classified in order to assess the validity of each
configuration. What our analysis has shown is that the blackfold approach to leading order would not be complete without the understanding of higher-order corrections up to second order in the derivative expansion, as this is required in order to understand the regimes of validity of leading order configurations. As a simple example, in Sec.~\ref{sec:blackdisc}, we have applied this method to a disc geometry in flat space-time for which all intrinsic and extrinsic curvature invariants vanish but the length scale associated with local variations of the thickness yields the condition $r_+\Omega\ll1$ near the axis of rotation, which is required in order to capture the ultraspinning limit of Myers-Perry black holes. 

In Sec.~\ref{sec:theorems} and \ref{sec:sclasses} we have proved several assertions regarding solutions to the blackfold equations. For example, we have shown that the only two stationary minimal surfaces embedded into $\mathbb{R}^{3}$ which solve the blackfold equations in flat space-time are the plane and the helicoid. We believe that a systematic study of the blackfold equations using symmetries as a guiding principle might turn out to be a fruitful endeavour for classifying the topologies of black hole horizons in higher-dimensions, at least those which admit regimes with a large separation of scales. 

In Sec.~\ref{sec:helicoid} we constructed rotating black helicoids in $D\ge6$ in asymptotically flat space-time. One interesting feature of these geometries is that, by taking an appropriate limit, one recovers the family of singly-spinning Myers-Perry black holes, described by a rotating black disc. Therefore, this family of rotating black helicoids seems to be connect to the family of Myers-Perry black holes, at least in the ultraspinning regime. It would be interesting to explore this using numerical methods. While this result is true in flat space-time, we have shown that it is not the case in (Anti)-de Sitter space-times as 
it was not possible to construct the corresponding helicoid geometries there. In a later publication \cite{Armas:2015nea}, we show how discs and helicoid geometries can be used to construct new effective theories and other non-trivial black hole geometries such as helicoidal black rings and helicoidal black tori.

In plane wave space-times we also constructed different families of rotating black helicoids and black discs. Though there are no exact analytic black holes in asymptotically plane wave space-times in vacuum, we have given evidence for the existence of several families of rotating black hole solutions with spherical topology (the analogue of Myers-Perry black holes). We have also shown that these geometries can be captured by taking an appropriate limit of black helicoids. In addition, we showed that rotating black catenoids also give rise to black hole horizons in plane wave space-times and in App.~\ref{sec:blackodd} we constructed several p-sphere black holes. It would be interesting to construct approximate metrics for these geometries using the machinery of \cite{Emparan:2007wm, LeWitt:2009qx, Camps:2012hw}.

In Sec.~\ref{sec:discpw} we showed that in plane wave space-times black discs and helicoid geometries can be
static due to the presence of an inherent limiting surface for certain values of the plane wave matrix
components. In the case of the disc this is analogous to what happens in de Sitter space-time, for which such
geometries describe the static ultraspinning regime of Kerr-de Sitter black holes. We also showed the existence
of static black $p$-ball geometries, which analogously to the de Sitter case, capture the intersection of the
horizon of higher-dimensional Kerr-de Sitter black holes with the cosmological horizon when $p$ is an even
number. However, we have noticed that these solutions are valid for all $p$. For odd $p$ these are connected to
another family of black hole solutions which do not have spherical horizon topology. As observed in
\cite{Emparan:2009vd}, it is not possible to construct odd-ball geometries when only centrifugal force and 
tension need to equilibrate each other. When inherent limiting surfaces are present in the background
space-time, however, they act as an internal pressure and its interplay with the tension allows for the existence of static odd-ball geometries. We have observed the same phenomenon in App.~\ref{sec:blackodd} where static black even $p$-spheres are also possible configurations.

In Sec.~\ref{sec:catenoid} we have given the first example of a non-trivial blackfold solution whose worldvolume geometry has self-intersections. This geometry, known as Scherk's second surface, interpolates between the catenoid and the helicoid and hence connects these two families (and also the disc since it can be obtained as limit of the helicoid) in plane wave space-times. However, it is unclear at the present moment whether or not blackfold geometries with self-intersections do give rise to regular black hole solutions since they do not satisfy \eqref{r:d}. However, one might expect that gravitational backreaction might smooth out such intersections.

We conclude this discussion with some \textit{caveats} and corresponding
avenues for future research. 

First of all, the perturbative construction of approximate metrics using
the blackfold approach has not been fully understood when a blackfold
worldvolume intersects a limiting surface in the space-time. It would
be very interesting to understand how this works using the methods
of \cite{Emparan:2007wm, Camps:2012hw} for the simplest example of
Myers-Perry black holes. In fact, our criteria of validity put forth in Sec.~\ref{sec:rvalidity} lead to the conclusion that for geometries with boundaries, the blackfold approximation is expected to break down near the boundary. This is due to the fact that the local variations in the thickness vanish too quickly as one approaches the boundary $r_0=0$. This suggests that either the geometry near the boundary must be replaced by something else than a black brane geometry \eqref{ds:blackp} or that one should require the effective blackfold description to have a smooth limit as $r_0\to0$ and conjecture that this requirement describes the geometry of the corresponding gravitational object. This is indeed the case for Myers-Perry black holes and indicates that the blackfold method is working better than expected. 

We also note that the blackfold method has previously been
successfully applied in plane wave space-times in vacuum by perturbing
\eqref{ds:blackp}, in particular, for perturbatively constructing a black
string geometry in $D=5$ \cite{LeWitt:2009qx}. However, there also some
apparent tension between the joint desiderata of horizon regularity
and plane wave asymptotics was exhibited.  The construction presented
here differs in some respect from that in \cite{LeWitt:2009qx},
in particular in the way the matched asymptotic expansion can be
implemented, and as a result we do not see any obvious obstruction to
constructing solutions with the desired asymptotically plane wave boundary
conditions. Specifically, because of the non-trivial extrinsic and/or
intrinisc geometry of the configurations discussed in this paper, we can
choose the perturbative parameter in the matched asymptotic expansion to
be the length scale associated with the extrinsic or intrinsic scales of
the worldvolume geometry. By contrast, in the black string construction
of \cite{LeWitt:2009qx} (with its trivial extrinsic geometry) the only
other available dimensionful parameters were the inverse length scales
$\sim \sqrt{A_q}$ associated with the plane wave profile. Thus we can
treat the plane wave background exactly, while the matched asymptotic
expansion in \cite{LeWitt:2009qx} required an expansion of the plane
wave metric itself, thought of as a perturbation of Minkowski space,
in inverse powers of the typical mass scale $\mu$ ($ \sim \sqrt{A_1}$,
say).  Such an expansion corresponds to an expansion in positive powers
of some suitably defined radial coordinate and is therefore not suitable
for exploring the asymptotics of the full solution.

In fact, a preliminary study of the perturbative construction of
a black ring, embedded as in theorem \ref{theo:full}, in plane wave
backgrounds in the intermediate region tells us that the perturbations
fall-off rapidly enough at infinity and do not change the asymptotics. It
would be interesting to study this for more non-trivial embeddings of
the black ring as in App.~\ref{sec:oddspheres} using the methods of
\cite{Emparan:2007wm, Caldarelli:2008pz} and to use numerical methods
as in \cite{Kleihaus:2012xh,Dias:2014cia, Kleihaus:2014pha, Figueras:2014dta} in order to
construct the full solution.

As was explained in Sec.~\ref{sec:validity}, the blackfold method
for minimal surfaces can be seen, up to second order, as a purely
hydrodynamic expansion in a curved background. It would be very
interesting to understand how these geometries are modified when second
order corrections given in \eqref{eq:free} are taken into account using
the tools available in \cite{Armas:2013hsa,Armas:2013goa, Armas:2014bia,
Armas:2014rva}. 

Finally, it would be interesting to consider charged blackfolds and
construct analogous geometries in plane wave space-times in string theory.


\section*{Acknowledgements}
We would like to thank Irene Amado and Joan Camps for useful discussions. We are extremely grateful to Roberto Emparan, Troels Harmark and  Niels A. Obers for useful discussions, comments and for a detailed reading of an earlier draft of this paper. We would also like to thank William H. Meeks III and Joaqu\'{i}n P\'{e}rez for useful e-mail correspondence. This work has been supported by the Swiss National Science Foundation and the �Innovations- und Kooperationsprojekt C-13� of the Schweizerische Universit\"{a}tskonferenz SUK/CUS. JA acknowledges the current support of the ERC Starting Grant 335146 \textbf{HoloBHC}.
 

\appendix



 \section{Detailed validity analysis of configurations} \label{sec:valanal}
 In this appendix we provide the specific details regarding the validity analysis of several of the configurations studied in the core of this paper.
 
 \subsection*{Black helicoids in flat space-time}
 Here we study the validity regime of the black helicoids of Sec.~\ref{sec:flathelicoids}. Since the embedding \eqref{emb:helic} is a \textbf{Type I} embedding \eqref{ds:type1} in flat space-time then we only need to check, according to the analysis of Sec.~\ref{sec:validity} and the requirements \eqref{eq:req}, the invariants $|\textbf{k}^{-1}\nabla_a\nabla^{a}\textbf{k}|^{-\frac{1}{2}}$, $|\mathcal{R}|^{-\frac{1}{2}}$ and $|u^{a}u^{b}\mathcal{R}_{ab}|^{-\frac{1}{2}}$. Since this is flat space-time we have $R_0=1$ and hence, according to \eqref{r:type1} we have $\mathcal{R}=\mathcal{R}_{\tilde{\mathbb{E}}}$. Therefore we need to require $r_0\ll |\mathcal{R}_{\tilde{\mathbb{E}}}|^{-\frac{1}{2}}$, explicitly, 
 \beq \label{r0:helic}
r_0\ll \frac{(\lambda^2+a^2\rho^2)}{\sqrt{2}a \lambda}~~.
\eeq
This has a minimum at the origin $\rho=0$ and is maximal at the boundaries $\rho_\pm$. Hence one only needs to satisfy $r_0\ll \lambda/(\sqrt{2}a)$ which is always possible by appropriately tuning the temperature $T$ in \eqref{eq:pressure} and the ratio $\lambda/a$.\footnote{Note that, as explained in the beginning of Sec.~\ref{sec:flathelicoids}, the only physical parameter in the embedding \eqref{emb:helic} is the ratio $\lambda/a$.} Note that if we had taken the limit $\lambda\to0$ in \eqref{r0:helic} we would have obtained a divergent result since the plane $\mathbb{R}^2$ is Ricci-flat. For the invariant $|u^{a}u^{b}\mathcal{R}_{ab}|^{-\frac{1}{2}}$ we instead obtain the requirement
\beq \label{r:helic}
r_0\ll \textbf{k}\frac{\sqrt{\lambda^2+a^2\rho^2}}{\Omega a \lambda}~~,
\eeq
which is minimum at the boundary where $\textbf{k}=0$. Since $r_0$ scales with the same power of $\textbf{k}$ (see Eq.~\eqref{eq:pressure}) as the r.h.s. of \eqref{r:helic} we need to require that $r_+\ll 1/(a\Omega)$ which is again always possible by appropriately tuning the temperature $T$ and the ratio $1/(a\Omega)$. This is still compatible with the requirement $\Omega \lambda<1$ obtained below \eqref{eq:reqh}. We finally check the invariant associated with variations in the local temperature, we find that we must have  
\beq \label{c:w}
r_+\ll \textbf{k}\frac{\sqrt{a^2 \rho ^2+\lambda ^2}}{a \Omega  \sqrt{\lambda ^2+2 a^2 \rho ^2-\Omega ^2 \left(a^2 \rho
   ^2+\lambda ^2\right)^2}}~~.
\eeq
Near $\rho=0$ this implies that we must have $r_+\ll \rho_+$.  
  
\subsection*{Black discs of \textbf{Type I} in plane wave space-times}
Here, the validity of regime of the configurations of Sec.~\ref{sec:disc1} is analysed. Since we are dealing with a \textbf{Type I} embedding in plane wave space-times we need to check if the invariants $\textbf{k}^{-1}\nabla_a\nabla^a\textbf{k}$, $\mathcal{R}$ and $u^{a}u^{b}\mathcal{R}$ as well as \eqref{r:pw} satisfy the requirements \eqref{eq:req}. From the intrinsic curvature invariants we obtain the requirements
\beq \label{req:disc1}
r_0\ll \frac{R_0^2}{\sqrt{2A_1(2+A_1\rho^2)}}~~,~~r_0\ll \textbf{k} \frac{R_0}{\sqrt{(1+R_0^2)A_1+(R_0^2-1)\Omega^2}}~~.
\eeq
As explained in Sec.~\ref{sec:cembed} for static embeddings of \textbf{Type I} in plane wave space-times \eqref{r:type1} tells us that $\mathcal{R}$ diverges at the boundary where $R_0=\textbf{k}=0$. Therefore, when $\Omega=0$ we need to introduce $\epsilon$ and consider the configuration to be valid up to $\rho=\rho_\pm \mp \epsilon$. In this case, it is enough to require $r_0\ll|\textbf{k}^{-1}\nabla_a\nabla^a\textbf{k}|^{-\frac{1}{2}}$ as we will see below. If $\Omega\ne0$ then both requirements \eqref{req:disc1} take their maximum value at the boundary $\textbf{k}=0$ and minimum when $\rho=0$. Both of them reduce to $r_+\ll \sqrt{A_1}^{-1}$. The invariants associated with the curvatures of the background \eqref{r:pw} yield
\beq
r_0\ll \frac{R_0}{2\sqrt{A_1}}~~,~~r_0\ll R_0\frac{\textbf{k}}{\sqrt{A_1}\Omega \rho}~~,
\eeq
and lead again to $r_+\ll \sqrt{A_1}^{-1}$ if $\Omega\ne0$. Finally, the invariant associated to changes in the local temperature leads to the condition
\beq \label{req:discpw}
r_+\ll\frac{\textbf{k}}{\sqrt{A_1-\Omega^2}\sqrt{2+(A_1-\Omega^2)\rho^2}}~~.
\eeq
Near $\rho=0$ this leads to the requirement $r_+\ll\rho_+$ while near the boundary it becomes impossible to satisfy. Therefore one needs to introduce $\epsilon$ and assume a well defined boundary expansion.	  
  
\subsection*{Black discs of \textbf{Type II} in plane wave space-times}  
We now look the configurations of Sec.~\ref{sec:disc2}. Since these are \textbf{Type II} embeddings in plane wave space-times we know from the general analysis of Sec.~\ref{sec:cembed} that $R_{||}=R_{//}=0$. Furthermore, from \eqref{r:type2}, we have that $\mathcal{R}=\mathcal{R}_{\mathbb{E}}=0$. Therefore we only need to check the invariants $\textbf{k}^{-1}\nabla_a\nabla^a\textbf{k}$ and $u^{a}u^{b}\mathcal{R}_{ab}$. From the last invariant we obtain the requirement
\beq
r_0\ll\frac{\textbf{k}}{\sqrt{2A_1}}~~,
\eeq
and hence we need $r_+\ll \sqrt{A_1}^{-1}$ while from the first invariant we obtain again condition \eqref{req:discpw} and hence the same conclusions as in the previous case apply to this configuration.

\subsection*{Black helicoids of \textbf{Type I} in plane wave space-times}
The configurations found in Sec.~\ref{sec:helicoid1} are \textbf{Type I} embeddings and hence we need to analyse the invariants $\textbf{k}^{-1}\nabla_a\nabla^a\textbf{k}$, $\mathcal{R}$, $u^{a}u^{b}\mathcal{R}_{ab}$ as well as $R_{||},R_{//}$. Assuming that the solution is not static we find the first two set of bounds
\beq
r_0\ll \frac{1}{\sqrt{2}} \frac{\lambda}{\sqrt{A_1\lambda^2+a^2}}~~,~~r_0\ll \frac{\textbf{k}}{\sqrt{A_1-a^2\Omega^2}}~~,
\eeq
which are obtained by evaluating $\mathcal{R}$ and $u^{a}u^{b}\mathcal{R}_{ab}$ at the origin $\rho=0$. These imply that we need $r_0\ll \sqrt{A_1+a^2/\lambda^2}^{-1}$ and $r_+\ll\sqrt{A_1-a^2\Omega^2}^{-1}$. The invariant associated with variations of the local temperature implies that $r_+\ll\rho_+$ near $\rho=0$ and near the boundary it forces us to introduce the cutt-off $\epsilon$. From the second set of invariants we need to require that 
\beq
r_0\ll \frac{\sqrt{\lambda^2+a^2\rho^2}}{\sqrt{A_1}a\rho}~~,~~r_0\ll \frac{\textbf{k}}{\sqrt{A_1}\Omega a\rho}~~.
\eeq
This takes its minimum value at the boundaries $\rho_\pm$ when $\Omega\ne0$. If $\Omega=0$ then these diverge at the boundary but this divergence has been taken care of by the introduction of $\epsilon$. We find the same type of requirements as for the invariants $\mathcal{R}$ and $u^{a}u^{b}\mathcal{R}_{ab}$ and hence these configurations are valid in the interval $\rho_- +\epsilon\le\rho\le\rho_+-\epsilon$.

\subsection*{Black helicoids of \textbf{Type II} in plane wave space-times}
For these embeddings of Sec.~\ref{sec:helicoid2}, we have that $R_{||}=R_{//}=0$, therefore we only need to check $\textbf{k}^{-1}\nabla_a\nabla^a\textbf{k}$, $\mathcal{R}$ and $u^{a}u^{b}\mathcal{R}_{ab}$. Using \eqref{r:type1} we find the requirements
\beq
r_0\ll \frac{\lambda^2+a^2\rho^2}{\sqrt{2} a \lambda}~~,~~r_0\ll \frac{\textbf{k}}{\sqrt{A_1}}\frac{\sqrt{\lambda^2+a^2\rho^2}}{\sqrt{\lambda^2+a^2(2\rho^2-\lambda^2\Omega^2)}}~~,
\eeq
associated with the curvatures. The above requirements take their most strict value at the origin $\rho=0$. It is then only sufficient to require $r_0\ll \lambda/a$ and $r_+\ll \sqrt{A_1(1-a^2\Omega^2)}^{-1}$. The analysis of the invariant $\textbf{k}^{-1}\nabla_a\nabla^a\textbf{k}$ as the same as for helicoids of \textbf{Type I} and hence the same conclusions apply here.

\subsection*{Higer-dimensional Clifford tori in de Sitter space-times}
Since this is a \textbf{Type III} embedding found in Sec.~\ref{sec:desitter}, we need to analyse the invariants $\textbf{k}^{-1}\nabla_a\nabla^a\textbf{k}$, $\mathcal{R}$ and $u^{a}u^{b}\mathcal{R}_{ab}$ as well as $R_{||}$ and $R_{//}$. These last two give rise to the conditions $r_0\ll L$ and $r_0\ll r_m$. Computing explicitly the worldvolume Ricci scalar we find
\beq
\mathcal{R}=\frac{b_1+b_2 R_0+b_3 \rho R_0'+b_4\rho^2R_0''}{b_5 \rho^2}~~,
\eeq
for some constants $b_i$ and where the prime denotes a derivative with respect to $\rho$. The Ricci scalar only diverges when $\rho=0$ since $R_0'$ and $R_0''$ do not diverge anywhere except at $\rho=0$. Therefore the singularity is shielded behind the black hole horizon. Next we need to compute the scalar $u^{a}u^{b}\mathcal{R}_{ab}$. We can check that in general one has
\beq
\mathcal{R}_{\tau\tau}=b_1R_0\left(\frac{b_2 R_0'}{\rho}+b_3 R_0''\right)~~,~~\mathcal{R}_{\phi_{\hat a}\phi_{\hat a}}=\sin^2(\theta_{\hat a})\left(b_1+b_2 R_0+b_3 \rho R_0'\right)~~,
\eeq
for some constants $b_i$ and where $\theta_{\hat a}$ are angles which parametrize the $\kappa_1$ and $\kappa_2$-dimensional spheres. Only the component $\mathcal{R}_{\tau\tau}$ diverges and that only happens when $\rho=0$. From the invariant $\textbf{k}^{-1}\nabla_a\nabla^a\textbf{k}$ we find that it is sufficient to require $r_+\ll \rho_{\pm}$ near the origin $\rho=0$, however near the boundaries we need to introduce the cut-off $\epsilon$. Therefore, these configurations are valid blackfold solutions in the interval $\rho_--\epsilon\le\rho\le\rho_+-\epsilon$.
  
  \section{Higher-dimensional black helicoids in flat and plane-wave space-times}  \label{sec:higherhelicoid}
In this section we construct the higher-dimensional analogue of the black helicoids constructed in Sec.~\ref{sec:helicoid} and Sec.~\ref{sec:helicoidpw}. There are two generalisations of the helicoid geometry in $\mathbb{R}^{(D-1)}$ available in the literature, the Barbosa-Dajczer-Jorge helicoids \cite{Jorge:1984} and, recently, the Choe-Hoppe helicoids \cite{Hoppe:2013}. The latter case is not suitable for constructing black hole horizon geometries in flat and plane wave space-times because  these helicoid geometries have a conical singularity at the origin. Therefore we focus on the Barbosa-Dajczer-Jorge helicoids \cite{Jorge:1984}. These helicoids can be of \textbf{Type I} or \textbf{Type II} and we will analyse both cases simultaneously. 

\subsubsection*{Embedding coordinates and geometry}

Explicit coordinate embeddings for the Barbosa-Dajczer-Jorge helicoids into a subset $\mathbb{R}^{2N+1}$, where $N$ is an integer, of $\mathbb{R}^{(D-1)}$ or $\mathbb{R}^{(D-2)}$ are given in \cite{LeeLee:2014}. These can be written in the form
 \beq \label{emb:higherhell}
 \begin{split}
 X^{q}(\rho_q,\phi) &= \rho_q \cos (a_q \phi)~\text{if $q$ is odd and $1\le q\le 2N$}~~,\\
 X^{q}(\rho_q,\phi)&= \rho_{q-1} \sin (a_{q-1} \phi)~\text{if $q$ is even and $1\le q\le 2N$}~~,\\
X^{q}(\rho_q,\phi)&=\lambda ~\phi~\text{if $q=2N+1$}~~,
 \end{split}
 \eeq
 and $t=\tau~,~y=0,~,X^{i}=0~,~i=2N+2,...,D-1$ if the embedding is of \textbf{Type I} and $t=\tau~,~y=z,~,X^{i}=0~,~i=2N+2,...,D-2$ if the embedding is of \textbf{Type II}. Here $a_q,\lambda$ are constants which without generality we assume to be $a_q>0$ and $\lambda\ge0$. Note that $N$ and $p$ are related such that $p=2N$. The coordinates lie within the range $-\infty<\rho_q,\phi<\infty$. For $N=1$ and $\lambda=1$ we obtain the case studied in Sec.~\ref{sec:helicoid} and Sec.~\ref{sec:helicoidpw}. In general we can rescale $\phi$ such that $\phi\to \lambda^{-1}\phi$ and $a_q \to  \lambda^{-1} a_q$ and get rid of $\lambda$. However, we will not do so, since we want to consider later the case $\lambda=0$ which represents a minimal cone. The induced metric on the worldvolume of \textbf{Type I} takes the form
 \beq \label{ds:helicoidh}
 \textbf{ds}^2=-R_0^2d\tau^2+\sum_{\hat{a}=1}^{N}d\rho_{\hat{a}}^2+(\lambda^2+\sum_{\hat{a}=1}^{N}a_{\hat{a}}^2\rho_{\hat{a}}^2)d\phi^2~~,
 \eeq
 while on the worldvolume of \textbf{Type II} it reads
 \beq
\textbf{ds}^2= -R_0^2d\tau^2+2(1-R_0^2)d\tau dz+(2-R_0^2)dz^2+\sum_{\hat{a}=1}^{N}d\rho_{\hat{a}}^2+(\lambda^2+\sum_{\hat{a}=1}^{N}a_{\hat{a}}^2\rho_{\hat{a}}^2)d\phi^2~~,
 \eeq 
 where,
 \beq \label{R0:helix}
 R_0^2=1-\left(\sum_{\hat{a}=1}^{N}A_{\hat{a}}^2\rho_{\hat{a}}^2+A_{N+1}^2\lambda^2\phi^2\right)~~.
 \eeq
 Here we have made a slight modification of notation. For each of the $N$ two-planes $(x_{\hat a},x_{\hat a+1})$ of $\mathbb{R}^{(2N+1)}$ we have set $A_{x_{\hat a},x_{\hat a+1}}=A_{\hat a}$ and also $A_{x_{2N+1}}=A_{N+1}$. These choices will be compatible with the choice of Killing vector field as we will see below.

The helicoid can be boosted along the $\phi$ direction with boost velocity $\Omega$ such that $\textbf{k}^{a}\partial_a=\partial_\tau+\Omega\partial_\phi$ which maps onto the vector field in the ambient space-time
\beq \label{kill:hell}
k^{\mu}\partial_\mu=\partial_t +\Omega\sum_{\hat{a}=1}^{N}a_{\hat{a}}\left(x_{\hat{a}}\partial_{x_{\hat{a}+1}}-x_{\hat a+1}\partial_{x_{\hat{a}}}\right)+\lambda\Omega\partial_{x_{2N+1}}~~.
\eeq
For this to be a Killing vector field of the space-time \eqref{ds:pw} we need to require that for each of the $N$ two-planes $(x_{\hat a},x_{\hat a+1})$ of $\mathbb{R}^{(2N+1)}$, $A_{x_{\hat a},x_{\hat a+1}}=A_{\hat a}$ and $A_{N+1}=0$. Therefore the helicoid is rotating with angular velocity $a_{\hat a}\Omega$ in each of the $(x_{\hat a},x_{\hat a+1})~,~\hat a=1,..,N$ planes and it is boosted along the $x_{2N+1}$ direction with boost velocity $\lambda \Omega$. The modulus of the Killing vector field is
\beq \label{eq:modhigh}
\textbf{k}^2=1+\sum_{\hat a=1}^{N}A_{\hat a}\rho_{\hat a}^2-\Omega^2\left(\lambda^2+\sum_{\hat a=1}^{N}a_{\hat a}^2\rho_{\hat a}^2\right)~~.
\eeq
From this expression we see that a limiting surface appears in general and that for the solution to be valid at the origin $\rho_{\hat a}=0$ we must require that $\Omega^2\lambda^2<1$. If the geometry is static $\Omega=0$ then a limiting surface may also exist provided that at least one of eigenvalues $A_{\hat a}$ is negative. The boundaries of the geometry are given by the ellipsoidal equation
\beq
\sum_{\hat{a}=1}^{N}(A_{\hat{a}}+a_{\hat{a}}^2\Omega^2)\rho_{\hat{a}}^2=\Omega^2\lambda^2-1~~.
\eeq
These higher-dimensional helicoids give rise to black hole horizon topologies $\mathbb{R}\times\mathbb{S}^{(D-3)}$ in the case of \textbf{Type I} and $\mathbb{R}^2\times\mathbb{S}^{(D-4)}$ in the case of \textbf{Type II}. The size of the transverse $(n+1)$-dimensional sphere varies from a maximum size at the origin $\rho_{\hat a}=0$ and vanishes at the boundaries.

 \subsubsection*{Solution to the equations of motion}
 We have shown in theorem \ref{theo:solpw} that helicoids are solutions of the blackfold equations. We will now conclude the same for their higher-dimensional versions. First we note that higher-dimensional helicoids are minimal surfaces in $\mathbb{L}^{(D)}$ by appropriately tuning the components $A_{qr}$. For these configurations to be minimal surfaces in plane wave space-times they need to satisfy \eqref{eq:cpw}. Before writing it explicitly, we need to compute the unit normal vector to the surface embedded in $\mathbb{R}^{(2N+1)}$. This has the form
\beq
n_{\rho}=\frac{1}{\sqrt{\lambda^2+\sum_{\hat{a}=1}^{N}a_{\hat{a}}^2\rho_{\hat{a}}^2}}\left(0,\lambda \sin(a_1\phi),-\lambda\cos(a_1\phi),...,\sum_{\hat{a}=1}^{N}a_{\hat{a}}\rho_{\hat{a}}\right)~~,
\eeq
where we have omitted the transverse index $i$ from ${n^{i}}_{\rho}$ since there is only one normal direction to the surface in $\mathbb{R}^{(2N+1)}$. Therefore, Eq.~\eqref{eq:cpw} demands that
\beq
\sum_{\hat{a}=1}^{N}\sin(a_{\hat{a}}\phi)\cos(a_{\hat{a}}\phi)\left(A_{\hat a}-A_{\hat{a}+1}\right)+A_{x_{2N+1}}\lambda\phi\sum_{\hat{a}=1}^{N}a_{\hat{a}}\rho_{\hat{a}}=0~~.
\eeq
Since the last term is linear in $\phi$ we need to require $A_{x_{2N+1}}=0$ and furthermore that $A_{\hat a}=A_{\hat{a}+1}~,~\hat a=1,...,N$. These choices are compatible with the requirement for the Killing vector field \eqref{kill:hell} to be a Killing vector field of the plane wave space-time \eqref{ds:pw}.

According to corollary \ref{cor:type11} and corollary \ref{cor:type2c} the vanishing of ${K_{\tau\tau}}^{i}$ is enough for the configuration to be a static solution of the blackfold equations. In order for it to be a stationary solution then one must also require that $u^{\hat a}u^{\hat b}{K_{\hat a \hat b}}^{i}=0$. Explicit computation of the extrinsic curvature tensor leads to the result
 \beq
 {K_{\rho_{\hat a}\rho_{\hat b}}}={K_{\phi\phi}}=0~~,~~{K_{\rho_{\hat b}\phi}}=-\frac{a_{\hat b}\lambda}{\sqrt{\lambda^2+\sum_{\hat a=1}^{N}a_{\hat a}^2\rho_{\hat a}^2}}~~. 
 \eeq
 Therefore we since we have that ${K_{\phi\phi}}=0$, the blackfold equations are trivially satisfied for rotating higher-dimensional helicoids. 
 
\subsubsection*{Free energies}
We now present the free energies of these configurations. For \textbf{Type I} helicoids the free energy is given by
\beq 
\mathcal{F}=\frac{\Omega_{(n+1)}}{16\pi G}r_+^{n} \int d\phi\int d\rho~ R_0\sqrt{\lambda^2+\sum_{\hat a=1}^{N}a_{\hat a}^2\rho_{\hat a}^2}\thinspace~ \textbf{k}^{n}~~,~~d\rho=\prod_{{\hat a}=1}^{N}d\rho_{\hat a}~~,
\eeq 
while for \textbf{Type II} helicoids we have that
\beq \label{eq:freehigh}
\mathcal{F}=\frac{\Omega_{(n+1)}}{16\pi G}r_+^{n}\int dz  \int d\phi \int d\rho\sqrt{\lambda^2+\sum_{\hat{a}=1}^{N}a_{\hat a}^2\rho_{\hat a}^2}\thinspace~ \textbf{k}^{n}~~,
\eeq 
where $R_0$ is given in \eqref{R0:helix} and $\textbf{k}$ is given in \eqref{eq:modhigh}. If we set $N=1$ then these free energies reduce to those analysed in Sec.~\ref{sec:helicoidpw} and if we further set $A_{1}=0$ then they reduce to the helicoid of Sec.~\ref{sec:helicoid}. The flat space-time limit of these higher-dimensional helicoids is obtained by setting $A_{\hat a}=0~,~\hat a=1,...,N$. For general $N$, it is not possible to integrate these free energies and obtain closed form expressions. However, this is not a problem using numerics.  In Fig.~\ref{fig:freehelicoiddajczer} we plot the free energy density \eqref{eq:freehigh} as a function of $n$ for a static \textbf{Type II} helicoid embedded in $\mathbb{R}^{5}$. 
  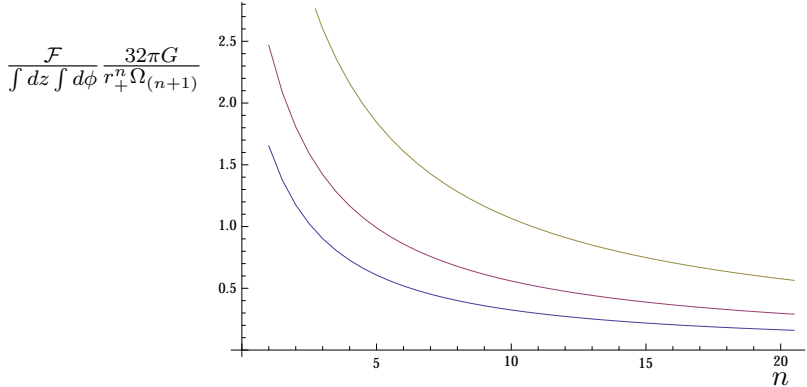
\begin{figure}[h!] 
\centering
  \includegraphics[width=0.5\linewidth]{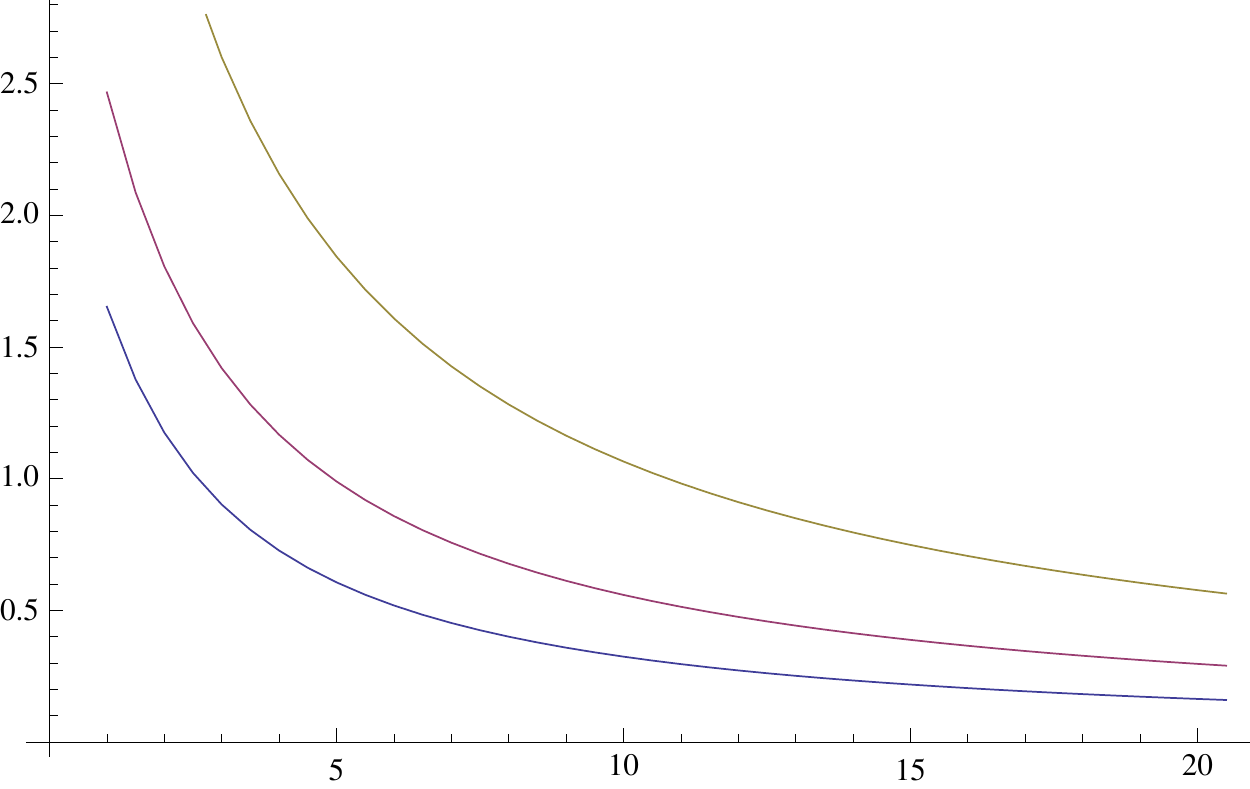}
  \begin{picture}(0,0)(0,0)
\put(-310,113){ $ \frac{ \mathcal{F}}{\int dz \int d\phi}\frac{32\pi G}{r_+^{n}\Omega_{(n+1)}}  $}
\put(-20,-5){ $ n  $}
\end{picture} 
\caption{Free energy density as a function of $n$ for static ($\Omega=0$) black Barbosa-Dajczer-Jorge helicoids of \textbf{Type II} with $N=2$, $a_1=a_2=1$ and $A_{1}=A_{2}=-1$ for $\lambda=\frac{1}{2}$ (blue curve), $\lambda=1$ (red curve) and $\lambda=2$ (yellow curve).} \label{fig:freehelicoiddajczer}
\end{figure}

\subsubsection*{Validity analysis}
The validity analysis follows the same footsteps as in the previous cases. Here we will just look at higher-dimensional helicoids of \textbf{Type II}. The analysis for higher-dimensional helicoids of \textbf{Type I} is very similar, though more cumbersome. For embeddings of \textbf{Type II} it is only necessary to analyse the invariants $\textbf{k}^{-1}\nabla_a\nabla^a\textbf{k}$, $\mathcal{R}$ and $u^{a}u^{b}\mathcal{R}_{ab}$. Explicitly evaluating the worldvolume Ricci scalar leads to
\beq  \label{r:helich}
\mathcal{R}=-\frac{2}{(\lambda^2+\sum_{\hat{a}=1}^{N}a_i^2\rho_i^2)^2}\left(\lambda^2\sum_{\hat{a}=1}^{N}a_{\hat{a}}^2 + \mathbb{P}_{\hat{a}\hat{b}}^{\hat{a}\ne \hat{b}}a_{\hat{a}}^2 a_{\hat{b}}^2(\rho_{\hat{a}}^2+\rho_{\hat{b}}^2) \right)~~,
\eeq
where the last term represents a sum of all the inequivalent permutations of $\hat{a},\hat{b}$. Clearly, this is a geometry that lies within the regime of validity when $\lambda\ne0$ as it does not diverge anywhere. In particular it implies that 
\beq
r_0\ll \frac{\lambda}{\sqrt{\sum_{\hat{a}=1}^{N}a_{\hat{a}}^2}}~~,
\eeq
i..e, the thickness of the blackfold must be much smaller than the pitch of the helicoid. However, when $\lambda=0$ the geometry is that of a higher-dimensional cone and the curvature blows up at the origin $\rho_{\hat{a}}=0$. Therefore it lies outside the validity regime of the method. A similar expression to \eqref{r:helich} can be found for $u^{a}u^{b}\mathcal{R}_{ab}$, which again does not diverge anywhere on the worldvolume geometry. In particular, one finds that the requirement
\beq
r_+\ll \frac{1}{\sqrt{\sum_{\hat{a}=1}^{N}\left(A_{\hat{a}}-\Omega^2a_{\hat{a}}^2\right)}}~~,
\eeq
is sufficient for these configurations to be valid. From the invariant $\textbf{k}^{-1}\nabla_a\nabla^a\textbf{k}$ we find that we need to require $r_+\ll \rho_{\hat a}^{+}$ near the origin $\rho_{\hat a}=0$ where $\rho_{\hat a}^{+}=\sqrt{A_{\hat a}-a_{\hat a}^2\Omega^2}^{-1}$. However, near the boundary we need to introduce the cut-off $\epsilon$ as in all other cases.


\section{Higher-dimensional black catenoids in plane wave space-times}  \label{sec:highercatenoid}
In theorem \ref{theo:solpw2} we have shown that catenoids solve the blackfold equations with an appropriate choice of eigenvalues $A_{q}$. We will now show that this is the case for a specific class of higher-dimensional catenoids of \textbf{Type II}. 

\subsubsection*{Embedding and geometry}
Higher-dimensional catenoids were found in \cite{Barbosa1981} (see also \cite{Hoppe:2013}) and explicit embeddings are given in \cite{2007arXiv0708.3310T}. Generically higher-dimensional catenoids can be embedded in $\mathbb{R}^{(p+1)}$, where $p\ge2$, by choosing a coordinate $\rho$ and a function $z(\rho)$ such that
\beq \label{eq:rho}
\rho=\int_{c}^{z(\rho)}\frac{dr}{\left(c^{-2(p-1)}r^{2(p-1)}-1\right)^{\frac{1}{2}}}~~.
\eeq
In the case $p=2$, $z(\rho)$ is smooth and is well defined on $\mathbb{R}$ while for the cases $p\ge3$, $z(\rho)$ is defined on the interval $[-S,S]$ where
\beq\label{eq:S}
S=S(c)=\int_{c}^{+\infty}\frac{dr}{\left(c^{-2(p-1)}r^{2(p-1)}-1\right)^{\frac{1}{2}}}<\infty~~.
\eeq
By defining the unit vector $\omega$ on $\mathbb{R}^{(p)}$ (i.e., the unit vector on the hyperplane $x_{p+1}=0$) the embedding of the higher-dimensional catenoid of \textbf{Type II} is given by $X^{q}(\rho,\omega)=(z(\rho)\omega,\rho)~,~q=1,...,p+1$ as well as $t=\tau~,~y=z$ and $X^{i}=0~,~i=p+2,..,D-2$. To see precisely how this works we first consider the case $p=2$. Integrating \eqref{eq:rho} we find
\beq
\rho=c \log\left(\frac{z(\rho)+\sqrt{z^2(\rho)-c^2}}{c}\right)~~.
\eeq
Solving it explicitly for $z(\rho)$ we obtain
\beq
z(\rho)=c\cosh\left(\frac{\rho}{c}\right)~~.
\eeq
By defining the unit vector on $\mathbb{R}^{2}$ as $\omega^{q}=(\cos(a\phi),\sin(a\phi))$ we obtain exactly the embedding given in Sec.~\ref{sec:catenoid}. In general, we find for any $p$ that
\beq \label{eq:genrho}
\rho=\frac{\frac{\sqrt{\pi } c^2 \Gamma \left(\frac{4-3 p}{2-2 p}\right)}{\Gamma \left(1+\frac{1}{2-2
   p}\right)}-z(\rho)^2 \left(\frac{c}{z(
   \rho)}\right)^p \, _2F_1\left(\frac{1}{2},\frac{p-2}{2 (p-1)};\frac{4-3
   p}{2-2 p};\left(\frac{f(\rho) }{c}\right)^{2-2 p}\right)}{(p-2)c}~~.
\eeq
For $p\ge3$ we cannot invert this transcendental equation and find $z(\rho)$ explicitly. However this can be done numerically. Similarly, we can also find $S$ by integrating \eqref{eq:S}
\beq
S=\frac{1}{(p-2)c}\frac{\sqrt{\pi } c^2 \Gamma \left(\frac{4-3 p}{2-2p}\right)}{\Gamma \left(1+\frac{1}{2-2
   p}\right)}~~.
\eeq
The induced worldvolume metric takes the simple form
\beq \label{ds:hcat}
\textbf{ds}^2= -R_0^2d\tau^2+2(1-R_0^2)d\tau dz+(2-R_0^2)dz^2+(1+z'(\rho)^2)d\rho^2+z^2(\rho)d\Omega^2_{(p-1)}~~,
\eeq
where we have set $A_{x_qx_q}=A_{1}~,~q=1,...p$ and $A_{x_{p+1}x_{p+1}}=0$, hence
\beq
R_0^2=1+A_{1}z^2(\rho)~~.
\eeq
 This requirement is necessary for setting the catenoid to rotate in order to solve the equations of motion. The first derivative of $z(\rho)$ in \eqref{ds:hcat} with respect to $\rho$ can be determined from \eqref{eq:genrho}, such that
\beq \label{eq:fp}
\frac{z '(\rho ) \left(\frac{c}{z(\rho )}\right)^{p-1}}{\sqrt{1-\left(\frac{z(\rho )}{c}\right)^{2-2 p}}}=1~~.
\eeq

\subsubsection*{Solution to the equations of motion}
In order to solve the equations of motion we will set the catenoid to rotate with angular velocity $\Omega^{\hat {a}}$ on each of the Cartan angles $\phi_{\hat a}$ associated with each of the $[p/2]$ two-planes of the $(p-1)$-dimensional sphere such that $\textbf{k}^{a}\partial_a=\partial_\tau+\sum_{\hat a=1}^{[p/2]}\Omega^{\hat a}\partial_{\phi_{\hat a}}$. For this to correspond to a background Killing vector field \eqref{eq:killing} we require $A_{x_qx_q}=A_{1}~,~q=1,...p$, which will be necessary for solving the equations of motion.

For \textbf{Type II} embeddings, according to corollary \ref{cor:type2c} we need to solve \eqref{eq:c3}. We first need to evaluate the normal vector, which reads
\beq
n_{\rho}=\frac{1}{\sqrt{1+z'(\rho)^2}}\left(0,0,-\omega^{q},z(\rho)\right)~~q=1,...,p-1~.
\eeq
With this we compute the extrinsic curvature components,
\beq
{K_{\tau\tau}}=-\frac{z(\rho)}{\sqrt{1+z'(\rho)^2}}\left(A_1+A_{p+1}\rho\right)~~,
\eeq
while for the other components we find
\beq
{K_{\rho\rho}}=-z''(\rho)~~,~~{K_{\rho\sigma_{\hat{a}}}}=z'(\rho)n_{\hat{b}}\partial_{\sigma_{\hat{a}}}\omega^{\hat b}~~,~~{K_{\sigma_{\hat{a}}\sigma_{\hat{b}}}}=z(\rho)n_{\hat c}\partial_{\sigma_{\hat{a}}}\partial_{\sigma_{\hat{b}}}\omega^{\hat c}~~.
\eeq
In the equation of motion only the $\sigma_{\hat{a}}$ coordinates that correspond to Cartan angles $\phi_{{\hat{a}}}$ play a role. It is easy to show that the mixed components ${K_{\phi_{\hat a}\phi_{\hat{b}}}}~,~\hat a\ne \hat b$ vanish. Therefore, Eq.~\eqref{eq:c3} reduces to
\beq
{K_{\tau\tau}}+(\Omega^{\hat{a}})^2{K_{\phi_{\hat a}\phi_{\hat a}}}=\frac{z(\rho)}{\sqrt{1+z'(\rho)^2}}\left(-A_1-A_{p+1}\rho-\sum_{\hat a}^{[p/2]} (\Omega^{\hat a})^2\omega_{\hat a}\partial_{\phi_{\hat a}}^2\omega^{\hat a}\right)=0~~.
\eeq
Since the second term in the parenthesis is linear in $\rho$ and the other terms do not depend on $\rho$ we need to set $A_{p+1}=0$. The remaining equation is solved if we set $\Omega^{\hat a}=\Omega~,~\hat a=1,..,[p/2]$ and only take odd values of $(p-1)$. \footnote{The same requirement of odd number of dimensions parametrising the sphere is also necessary for the rotating black odd spheres that we construct in App.~\ref{sec:blackodd}.} In this case we have that,
\beq
-\sum_{\hat a}^{[p/2]} (\Omega^{\hat a})^2\omega_{\hat a}\partial_{\phi_{\hat a}}^2\omega^{\hat a}=(p-1)\Omega^2~~,
\eeq
and hence we obtain a solution if $A_1=(p-1)\Omega^2$. Note that we have assumed that all Cartan angles have periodicity $2\pi$. Explicit computation of $\textbf{k}$ leads to the result $\textbf{k}=1$. Therefore, these higher-dimensional catenoids are non-compact in the $z$ and $\rho$ directions and give rise to black hole horizons of the form $\mathbb{R}^{(p)}\times\mathbb{S}^{(D-p-2)}$.

\subsubsection*{Free energy and validity}
The on-shell free energy of the rotating catenoids is given by
\beq
\mathcal{F}=\frac{\Omega_{(n+1)}}{16\pi G}\Omega_{(p-1)}\int dz \int d\rho~ z(\rho)^2\sqrt{1+z'(\rho)^2} r_+^{n}~~,
\eeq
where $r_+=n/(4\pi T)$. In the case $p=2$ this reduces to the free energy of the catenoid \eqref{free:cat2}. This integral can be evaluated numerically. However, because the integrant is everywhere positive then the free energy density will be a positive quantity in general.

We will now turn our attention to the validity of these configurations. Since this is a \textbf{Type II} we only need to analyse the invariants $\textbf{k}^{-1}\nabla_a\nabla^a\textbf{k}$, $\mathcal{R}$ and $u^{a}u^{b}\mathcal{R}_{ab}$. Evaluating explicitly the worldvolume Ricci scalar we find
\beq \label{ricci:cat}
\mathcal{R}=\frac{(p-2)+(p-2)z'(\rho)-2(p-1)z(\rho)z''(\rho)}{z^2(\rho)(1+z'(\rho)^2)^2}~~.
\eeq
One should now look for divergences in this expression. This invariant would diverge if $z(\rho)$ would vanish at some point. Looking at Eq.~\ref{eq:genrho} we see that if $z(\rho)\to0$ then the r.h.s. of \eqref{eq:genrho} becomes imaginary. Therefore we conclude that the Ricci scalar cannot diverge due to $z(\rho)$ since it has no zeros. The Ricci scalar could however diverge due to $z'(\rho)$ and $z''(\rho)$. From \eqref{eq:fp} we can deduce the behaviour of $z'(\rho)$ and $z''(\rho)$ and it is easy to see that $z'(\rho)$ and also $z''(\rho)$ approach $\infty$ only if $\rho\to\infty$. From \eqref{ricci:cat} we deduce that $\mathcal{R}\to0$ as $\rho\to\infty$ and hence constitutes no problem. We also compute the required components of the worldvolume Ricci tensor
\beq
\begin{split}
\mathcal{R}_{\tau\tau}&=-A_1\frac{pz'(\rho)^2+pz'(\rho)^4+z(\rho)z''(\rho)}{(1+z'(\rho)^2)^2}~~,\\
\mathcal{R}_{\phi_{\hat a}\phi_{\hat a}}&=f(\theta_{\hat a})\frac{(p-2)+(p-2)z'(\rho)^2-z(\rho)z''(\rho)}{(1+z'(\rho)^2)^2}~~,
\end{split}
\eeq
where $f(\theta_{\hat a})$ is a function of the form $(\cos\theta_{\hat a})^\alpha (\sin\theta_{\hat a})^{\beta}$ for some constants $\alpha$ and $\beta$. By the same arguments as above, the invariant $u^{a}u^{b}\mathcal{R}_{ab}$ does not diverge anywhere. The invariant $\textbf{k}^{-1}\nabla_a\nabla^a\textbf{k}$ vanishes since $\textbf{k}$ is constant along the worldvolume. Therefore we conclude that higher-dimensional catenoids are valid solutions of the blackfold equations.

\section{Black $p$-spheres in plane wave space-times} \label{sec:blackodd}
In this appendix we construct a series of black hole geometries with constant mean extrinsic curvature in plane wave space-times. The phenomenology of these black holes, when constructed in space-times with a non-trivial gravitational potential, is similar to the phenomenology of soap bubbles: tension must equilibrate with internal pressure. These black holes can be stationary, in which case, there is also an interplay between internal pressure, tension and centrifugal repulsion. These configurations constitute the analogue examples of those found in flat \cite{Emparan:2009vd} and (Anti)-de Sitter space-times \cite{Armas:2010hz}. The latter cases were described in \eqref{eq:dspodd}.


\subsection{Black $p$-spheres} \label{sec:oddspheres}
In order to embed these geometries we consider writing the $(D-2)$ Euclidean part of the metric \eqref{ds:pw} in the form
\be \label{eq:dss}
\sum_{q=1}^{p+1}dx_q^2=dr^2+r^2d\Omega_{(p)}^2+\sum_{q=p+2}^{D-2}dx_q^2~~,
\ee
where we label the coordinates on the $p$-sphere by $\mu^{\hat{a}}~,~\hat a=1,...,p$. We now choose the parametrisation
\beq \label{e:odds}
t=\tau~~,~~y=0~~,~~r=R~~,~~\mu^{\hat{a}}=\sigma^{\hat {a}}~~,~~X^{i}=0~,~i=p+2,...,D-2~~.
\eeq
With this choice, the induced metric on the worldvolume becomes
\beq
\textbf{ds}^2=-R_0^2d\tau^2+R^2d\Omega_{(p)}^2~~,~~R_0^2=1+A(R,\sigma^{\hat a})~~.
\eeq
We now assume that the blackfold is rotating with equal angular velocity along each of the $[(p+1)/2]$ Cartan angles $\phi_{\hat a}$ of the $p$-sphere, such that
\beq
\textbf{k}^{a}\partial_a=\partial_\tau+\Omega\sum_{\hat a=1}^{[(p+1)/2]} \partial_{\phi_{\hat a}}~~,~~\textbf{k}=\sqrt{R_0^2-\Omega^2R^2}~~,
\eeq
where we have assumed that $p$ is odd otherwise the term proportional to $\Omega^2$ in $\textbf{k}^2$ would be $\sigma^{\hat a}$-dependent. However, if $\Omega=0$, this is not required. For this Killing vector to correspond to a Killing vector field of the metric \eqref{ds:pw} we need to impose the same relations between the components $A_{qr}$ as for the higher-dimensional helicoids and catenoids of App.~\ref{sec:higherhelicoid} and App.~\ref{sec:highercatenoid}. For simplicity we focus on the case where $A_{x_q x_q}=A_{1}~,~q=1,...,p+1$. In this case we have that $A(R,\sigma^{\hat a})=A_1R^2$. 

The free energy of these configurations to leading order can be obtained using \eqref{eq:free} and reads
\beq \label{p:p}
\mathcal{F}[R]=-\Omega_{(p)} R_0 R^{p} P~~,~~P=-\frac{\Omega_{(n+1)}}{16\pi G} \left(\frac{n}{4\pi T}\textbf{k}\right)^{n}~~.
\eeq
This in fact takes the same form as the free energy for black odd-spheres in (Anti)-de Sitter space-time \cite{Armas:2010hz} provided we identify $A_1=L^2$. Varying this free energy with respect to $R$ leads to the solution
\beq \label{b:odd}
\Omega^2R^2=R_0^2 \frac{p+\mathbf{R}^2(n+p+1)}{(n+p)+\textbf{R}^2(n+p+1)}~~,~~\textbf{R}^2=A_1 R^2~~,
\eeq
which takes the same form as in (Anti)-de Sitter space \cite{Armas:2010hz} under the same identification. For $p=1$ these represent black rings in asymptotically plane-wave space-times. In general these have horizon topology $\mathbb{S}^{(p)}\times \mathbb{S}^{(n+1)}$. If $A_1<0$ then static solutions exists due to the repulsive gravitational potential as in de Sitter space-time. The balancing condition \eqref{b:odd} becomes \eqref{b:oddds} and $p$ can also take even values. Since the free energy is the same as in (Anti)-de Sitter space so are its thermodynamic properties, given in \cite{Armas:2010hz}. The validity of these configurations will be studied in the next section.


\subsection{Products of $m$-spheres}
In this section we generalise the previous construction to an arbitrary product of $m$-spheres. This will constitute the analogue configurations in plane wave space-times of those constructed in \cite{Emparan:2009vd, Armas:2010hz}. We consider writing the $(D-2)$-dimensional Euclidean metric of \eqref{ds:pw} as product of $m$ balls where the dimension of each sphere is $p_{(\hat a)}$ such that
\be 
\sum_{q=1}^{D-2}dx_q^2=\sum_{\hat a=1}^{m}\left(dr_{\hat a}^2+r_{\hat a}^2d\Omega_{(p_{\hat a})}^2\right)+\sum_{i=p+m+1}^{D-2}dx_i^2~~,
\ee
where $p=\sum_{\hat a=1}^{m}p_{(\hat a)}$. We make a similar choice of matrix components $A_{qr}$ as in the previous section, namely, for each set of coordinates $x^{q}~,~q=1,...,p_{\hat a}$ associated with each ball we set $A_{x_qx_q}=A_{\hat a}$. We further choose the embedding map
\beq 
t=\tau~~,~~y=0~~,~~r_{\hat a}=R_{\hat a}~~,~~\mu^{\hat{a}}=\sigma^{\hat {a}}~~,~~X^{i}=0~,~i=p+m+1,...,D-2~~,
\eeq
where the coordinates $\mu^{\hat{a}}$ now parametrize all the $p$ coordinates on the $m$ spheres.  The induced worldvolume geometry is
\beq
\textbf{ds}^2=-R_0^2d\tau^2+\sum_{\hat a=1}^{m}R_{\hat a}^2d\Omega_{(p_{\hat a})}^2~~,~~R_0^2=\left(1+\sum_{\hat a=1}^{m}A_{\hat a}R_{\hat a}^2\right)~~.
\eeq
Therefore we can see this geometry as a product of odd-spheres being embedded in an inhomogenous (Anti)-de Sitter space-time. We assume the geometry to be rotating with angular velocity $\Omega^{\hat a}$ in each of the Cartan angles associated with each $p_{\hat a}$-dimensional sphere. The Killing vector field is thus of the form
\beq
\textbf{k}^{a}\partial_a=\partial_\tau+\sum_{\hat a=1}^{[(p+1)/2]}\Omega^{\hat a}\partial_{\phi_{\hat a}}~~,~~\textbf{k}^2=R_0^2-\sum_{\hat a=1}^{m}\Omega_{\hat a}^2R_{\hat a}^2~~,
\eeq
where we have assumed each $p_{\hat a}$ to be an odd number. However, if each $\Omega_{\hat a}$ vanishes, this is not necessary and $p_{\hat a}$ must only satisfy $p_{\hat a}\ge1$ for all $\hat a$.

The free energy \eqref{eq:free} to leading order is given by
\beq
\mathcal{F}[R_a]=-V_{(p)}R_0P~~,~~V_{(p)}=\prod^{m}_{\hat a=1}\Omega_{(p_{\hat a})}R_{\hat a}^{p_{\hat a}}~~,
\eeq
where the pressure is given by \eqref{p:p} and $V_{(p)}$ is the volume of the product of m-spheres. Varying this equation with respect to each $R_{\hat a}$ gives rise to a set of $m$ coupled equations. The general solution takes the same form as in (Anti)-de Sitter space-time \cite{Armas:2010hz}
\beq \label{eq:ads}
(\Omega^{\hat a})^2R_{\hat a}^2=R_0^2\frac{p_{\hat a}+\textbf{R}_{\hat a}^{2}\left(n+p+1\right)}{(n+p)+(n+p+1)\textbf{R}^2}~~,~~\textbf{R}_{\hat a}^2=A_{\hat a}^2R_{\hat a}^2~~, ~~\textbf{R}^2=\sum_{\hat a=1}^{m}R_{\hat a}^{2}~~.
\eeq
In particular, if we set $A_{\hat a}=L^2~,~\hat{a}=1,...,m$ we obtain the same result as in \cite{Armas:2010hz}. These configurations give rise to horizon topologies of the form $\prod_{\hat a=1}^{m} \mathbb{S}^{(p_{\hat a})}\times \mathbb{S}^{(D-p-2)}$. Furthermore these configurations also admit a static solution $\Omega_{\hat{a}}=0$ for all $\hat a$ provided we take $A_{\hat a}<0$ for all $\hat {a}$. The thermodynamics also take the same form to leading order as for their (Anti)-de Sitter counterparts. 

\subsubsection*{Validity analysis}
For these configurations we need to analyse all scalars present in \eqref{eq:req} except for the scalar $\textbf{k}^{-1}\nabla_a\nabla^{a}\textbf{k}$ as it vanishes since $\textbf{k}$ is constant over the worldvolume. The scalars $\mathcal{R}$ and $K^{i}K_{i}$ give rise to the same condition, namely, $r_0\ll \textbf{R}$. The scalars $u^{a}u^{b}\mathcal{R}_{ab}$, $R_{//}$ and $R_{||}$ give rise to the condition
\beq
r_+\ll R_{\hat a}\left(1+\sum_{\hat a=1}^{m}A_{\hat a}R_{\hat a}^2\right)^{-\frac{1}{2}}~~,
\eeq
where $r_+=n/(4\pi T)$. These conditions are satisfied by taking $r_0\ll \text{min}(R_{a},\sqrt{A_{a}}^{-1})$.


\subsection{String and branes with a $p$-sphere}
Here we construct boosted strings along the $y$-direction of the plane wave space-time \eqref{ds:pw} and later also branes with a  $p$-spheres. These have no non-trivial analogue in flat or (Anti)-de Sitter space-times. We consider a similar embedding to \eqref{e:odds} describing boosted strings with a $p$-sphere, more precisely,
\beq \label{e:sodd}
t=\tau~~,~~y=z~~,~~r=R~~,~~\mu^{\hat{a}}=\sigma^{\hat {a}}~~,~~X^{i}=0~,~i=p+2,...,D-2~~.
\eeq
In this case the induced metric reads 
\beq \label{ds:sodd}
\textbf{ds}^2=-R_0^2d\tau^2+2(1-R_0^2)d\tau dz^2+(2-R_0^2)dz^2+R^2d\Omega_{(p)}^2~~,~~R_0^2=1+A_1R^2~~,
\eeq
where we have made the same choices for the components $A_{qr}$ as in Sec.~\ref{sec:oddspheres}. These geometries are non-compact along the $z$-direction. The Killing vector field of the boosted string with boost velocity $H$ is given by
\beq
\textbf{k}^{a}\partial_a=\partial_\tau+H \partial_z+\Omega\sum_{\hat {a}=1}^{[(p+1)/2]}\partial_{\phi_{\hat a}}~~,
\eeq
with norm,
\beq
\textbf{k}=\sqrt{(1+\textbf{R}^2)+2\textbf{R}^2 H-(1-\textbf{R}^2)H^2-\Omega^2R^2}~~,
\eeq
where we have assumed that the sphere is rotating with equal angular velocity in all Cartan angles and that $p$ is an odd number. If $\Omega=0$ then it is only necessary to require $p\ge1$. In this case the free energy to leading order is
\beq \label{eq:oddf}
\mathcal{F}[R]=-\Omega_{(p)}\int dz R^{p}P~~.
\eeq
Varying this with respect to $R$ and solving the resulting equation of motion leads to the solution
\beq
\Omega^2R^2=\textbf{R}^2 (1+H)^2+\frac{p \left(1-H^2\right)}{(n+p)}~~.
\eeq
If we take the flat space-time limit $A_1\to0$ and $H=0,p=1$ this geometry describes the uniform black cylinder constructed in \cite{Emparan:2009vd}. Since we must have that $\Omega^2R^2>0$ then this implies that we must have
\beq
H> \frac{2 p}{p-\textbf{R}^2 (n+p)}-1~~.
\eeq
Furthemore we must have that $\textbf{k}>0$ for the solution to be valid which implies that $H<1$. Therefore we obtain the bound on the boost velocity $H$,
\beq
\frac{2 p}{p-\textbf{R}^2 (n+p)}-1<H<1~~.
\eeq
Note also that a static solution with $\Omega=0$ exists provided
\beq
\textbf{R}^2=\frac{p \left(H^2-1\right)}{(n+p) (1+H)^2}~~.
\eeq
Since $H^2<1$ we must have that $A_1<0$ for this configuration to exist. These configurations have horizon topology $\mathbb{R}\times \mathbb{S}^{(p)}\times \mathbb{S}^{(D-p-3)}$. The validity analysis of these configurations results in the same conclusion as in the previous section, namely, one needs to require $r_0\ll \text{min}(R_{a},\sqrt{A_{a}}^{-1})$.

\subsubsection*{Thermodynamic properties}
The thermodynamics properties of these configurations can be obtained from the free energy \eqref{eq:oddf} using \eqref{eq:thermo}. The total mass, angular momentum and entropy read
\beq
M=\frac{\Omega_{(n+1)}V_{(p)}}{16\pi G}\int dz~ \tilde r_+^{n}\frac{H^2- \left(\textbf{R}^2 (H+1) (n+p)-(n+p+1)\right)}{H^2-1}~~,
\eeq
\beq
J=\frac{\Omega_{(n+1)}V_{(p)}}{16\pi G}\int dz~\tilde r_+^{n} \frac{(n+p) R \sqrt{\frac{p(1- H^2)}{n+p}+\textbf{R}^2 (H+1)^2}}{H^2-1}~~,
\eeq
\beq
S=\frac{\Omega_{(n+1)}V_{(p)}}{16\pi G}\int dz~ \tilde r_+^{n}\frac{n}{T}~~,
\eeq
where we have defined $V_{(p)}=\Omega_{(p)}R^{p}$ and 
\beq
\tilde r_+^{n}=\left(\frac{n}{4\pi T}\right)^{n}\left(\frac{n(1- H^2)}{(n + p)}\right)^{\frac{n}{2}}~~,
\eeq
which vanishes when $H=\pm 1$. Furthermore, the configuration also has a momentum $\mathcal{P}$ associated with the boost velocity $H$. One can obtain it from the free energy \eqref{eq:oddf} in a similar manner as for the angular momentum and entropy, that is, by taking the derivative $-\partial\mathcal{F}/\partial H$, leading to
\beq
\mathcal{P}=\frac{\Omega_{(n+1)}V_{(p)}}{16\pi G}\int dz~\tilde r_+^{n}\frac{(n+p)\left(\textbf{R}^2 (H+1)-H\right)}{H^2-1}~~.
\eeq
This is the total momentum along the $z$-direction.

\subsubsection*{Branes of odd-spheres} 
It is possible to generalise the previous configurations by considering instead strings of products of $m$-spheres. Here, however, we will make yet another generalisation by adding extra boosted flat directions to the brane worldvolume \eqref{ds:sodd}. We take the embedding map \eqref{e:sodd} but choose some of the $X^{i}$ functions such that $X^{\hat l}=z^{\hat l}~,~\hat l=p+2,...,p+2+\ell$ and $X^{i}=0~,~i=p+3+\ell,...,D-2$. The worldvolume metric \eqref{ds:sodd} becomes
\beq
\textbf{ds}^2=-R_0^2d\tau^2+2(1-R_0^2)d\tau dz^2+2(1-R_0^2)dz^2+R^2d\Omega_{(p)}^2+\sum_{\hat l=p+2}^{p+2+\ell} dz_{\hat l}^2~~,
\eeq
while the Killing vector field, now also boosted with boost velocity $H^{\hat l}$ along each $z^{\hat l}$ direction, takes the form
\beq
\begin{split}
&\textbf{k}^{a}\partial_a=\partial_\tau+H \partial_z+\Omega\sum_{\hat a=1}^{[(p+1)/2]}\partial_{\phi_{\hat a}}+\sum_{\hat l=p+2}^{p+2+\ell}H_{\hat l}\partial_{z^{\hat l}}~~,\\
&\textbf{k}=\sqrt{(1+\textbf{R}^2)+2\textbf{R}^2 H-(1-\textbf{R}^2)H^2-\bar{H}^2-\Omega^2R^2}~~, ~~\bar H^2=\sum_{\hat l=p+2}^{p+2+\ell}H_{\hat l}^2~~.
\end{split}
\eeq
By explicit evaluating the free energy and solving the resulting equation of motion we find the general solution
\beq
\Omega^2R^2=\textbf{R}^2 (H+1)^2-\frac{p \left(\bar H^2+H^2-1\right)}{(n+p)}~~,
\eeq
giving rise to valid blackfold solutions with horizon topologies $\mathbb{R}^{(\ell+2)}\times \mathbb{S}^{(p)}\times \mathbb{S}^{(n+1)}$.




\addcontentsline{toc}{section}{References}
\footnotesize
\providecommand{\href}[2]{#2}\begingroup\raggedright\endgroup


\begin{thebibliography}{10}

\bibitem{Harmark:2009dh}
T.~Harmark, ``{Domain Structure of Black Hole Space-Times},''
  \href{http://dx.doi.org/10.1103/PhysRevD.80.024019}{{\em Phys. Rev.} {\bf
  D80} (2009)  024019},
\href{http://arxiv.org/abs/0904.4246}{{\tt arXiv:0904.4246 [hep-th]}}.

\bibitem{Armas:2011ed}
J.~Armas, P.~Caputa, and T.~Harmark, ``{Domain Structure of Black Hole
  Space-Times with a Cosmological Constant},''
  \href{http://dx.doi.org/10.1103/PhysRevD.85.084019}{{\em Phys.Rev.} {\bf D85}
  (2012)  084019},
\href{http://arxiv.org/abs/1111.1163}{{\tt arXiv:1111.1163 [hep-th]}}.

\bibitem{Emparan:2009cs}
R.~Emparan, T.~Harmark, V.~Niarchos, and N.~A. Obers, ``{World-Volume Effective
  Theory for Higher-Dimensional Black Holes},''
  \href{http://dx.doi.org/10.1103/PhysRevLett.102.191301}{{\em Phys. Rev.
  Lett.} {\bf 102} (2009)  191301},
\href{http://arxiv.org/abs/0902.0427}{{\tt arXiv:0902.0427 [hep-th]}}.

\bibitem{Emparan:2009at}
R.~Emparan, T.~Harmark, V.~Niarchos, and N.~A. Obers, ``{Essentials of
  Blackfold Dynamics},'' \href{http://dx.doi.org/10.1007/JHEP03(2010)063}{{\em
  JHEP} {\bf 03} (2010)  063},
\href{http://arxiv.org/abs/0910.1601}{{\tt arXiv:0910.1601 [hep-th]}}.

\bibitem{Armas:2012jg}
J.~Armas and N.~A. Obers, ``{Relativistic Elasticity of Stationary Fluid
  Branes},'' \href{http://dx.doi.org/10.1103/PhysRevD.87.044058}{{\em
  Phys.Rev.} {\bf D87} (2013)  044058},
\href{http://arxiv.org/abs/1210.5197}{{\tt arXiv:1210.5197 [hep-th]}}.

\bibitem{Armas:2013hsa}
J.~Armas, ``{How Fluids Bend: the Elastic Expansion for Higher-Dimensional
  Black Holes},'' \href{http://dx.doi.org/10.1007/JHEP09(2013)073}{{\em JHEP}
  {\bf 1309} (2013)  073},
\href{http://arxiv.org/abs/1304.7773}{{\tt arXiv:1304.7773 [hep-th]}}.

\bibitem{Armas:2013goa}
J.~Armas, ``{(Non)-Dissipative Hydrodynamics on Embedded Surfaces},''
  \href{http://dx.doi.org/10.1007/JHEP09(2014)047}{{\em JHEP} {\bf 1409} (2014)
   047},
\href{http://arxiv.org/abs/1312.0597}{{\tt arXiv:1312.0597 [hep-th]}}.

\bibitem{Armas:2014bia}
J.~Armas and T.~Harmark, ``{Black Holes and Biophysical (Mem)-branes},''
  \href{http://dx.doi.org/10.1103/PhysRevD.90.124022}{{\em Phys.Rev.} {\bf D90}
  (2014)  124022},
\href{http://arxiv.org/abs/1402.6330}{{\tt arXiv:1402.6330 [hep-th]}}.

\bibitem{Armas:2014rva}
J.~Armas and T.~Harmark, ``{Constraints on the effective fluid theory of
  stationary branes},'' \href{http://dx.doi.org/10.1007/JHEP10(2014)063}{{\em
  JHEP} {\bf 1410} (2014)  63},
\href{http://arxiv.org/abs/1406.7813}{{\tt arXiv:1406.7813 [hep-th]}}.

\bibitem{Emparan:2009vd}
R.~Emparan, T.~Harmark, V.~Niarchos, and N.~A. Obers, ``{New Horizons for Black
  Holes and Branes},'' \href{http://dx.doi.org/10.1007/JHEP04(2010)046}{{\em
  JHEP} {\bf 04} (2010)  046},
\href{http://arxiv.org/abs/0912.2352}{{\tt arXiv:0912.2352 [hep-th]}}.

\bibitem{Armas:2010hz}
J.~Armas and N.~A. Obers, ``{Blackfolds in (Anti)-de Sitter Backgrounds},''
  \href{http://dx.doi.org/10.1103/PhysRevD.83.084039}{{\em Phys.Rev.} {\bf D83}
  (2011)  084039}, \href{http://arxiv.org/abs/1012.5081}{{\tt arXiv:1012.5081
  [hep-th]}}.

\bibitem{Caldarelli:2010xz}
M.~M. Caldarelli, R.~Emparan, and B.~Van~Pol, ``{Higher-dimensional Rotating
  Charged Black Holes},'' \href{http://dx.doi.org/10.1007/JHEP04(2011)013}{{\em
  JHEP} {\bf 1104} (2011)  013},
\href{http://arxiv.org/abs/1012.4517}{{\tt arXiv:1012.4517 [hep-th]}}.

\bibitem{Emparan:2011hg}
R.~Emparan, T.~Harmark, V.~Niarchos, and N.~A. Obers, ``{Blackfolds in
  Supergravity and String Theory},''
  \href{http://dx.doi.org/10.1007/JHEP08(2011)154}{{\em JHEP} {\bf 1108} (2011)
   154},
\href{http://arxiv.org/abs/1106.4428}{{\tt arXiv:1106.4428 [hep-th]}}.

\bibitem{Kleihaus:2012xh}
B.~Kleihaus, J.~Kunz, and E.~Radu, ``{Black rings in six dimensions},''
  \href{http://dx.doi.org/10.1016/j.physletb.2012.11.015}{{\em Phys.Lett.} {\bf
  B718} (2013)  1073--1077},
\href{http://arxiv.org/abs/1205.5437}{{\tt arXiv:1205.5437 [hep-th]}}.

\bibitem{Dias:2014cia}
O.~J. Dias, J.~E. Santos, and B.~Way, ``{Rings, Ripples, and Rotation:
  Connecting Black Holes to Black Rings},''
  \href{http://dx.doi.org/10.1007/JHEP07(2014)045}{{\em JHEP} {\bf 1407} (2014)
   045},
\href{http://arxiv.org/abs/1402.6345}{{\tt arXiv:1402.6345 [hep-th]}}.

\bibitem{Kleihaus:2014pha}
B.~Kleihaus, J.~Kunz, and E.~Radu, ``{Black ringoids: spinning balanced black
  objects in $d\geq 5$ dimensions -- the codimension-two case},''
\href{http://arxiv.org/abs/1410.0581}{{\tt arXiv:1410.0581 [gr-qc]}}.

\bibitem{Figueras:2014dta}
P.~Figueras and S.~Tunyasuvunakool, ``{Black rings in global anti-de Sitter
  space},''
\href{http://arxiv.org/abs/1412.5680}{{\tt arXiv:1412.5680 [hep-th]}}.

\bibitem{LeWitt:2009qx}
J.~Le~Witt and S.~F. Ross, ``{Black holes and black strings in plane waves},''
  \href{http://dx.doi.org/10.1007/JHEP01(2010)101}{{\em JHEP} {\bf 1001} (2010)
   101},
\href{http://arxiv.org/abs/0910.4332}{{\tt arXiv:0910.4332 [hep-th]}}.

\bibitem{Hubeny:2002nq}
V.~E. Hubeny and M.~Rangamani, ``{Generating asymptotically plane wave
  space-times},'' \href{http://dx.doi.org/10.1088/1126-6708/2003/01/031}{{\em
  JHEP} {\bf 0301} (2003)  031},
\href{http://arxiv.org/abs/hep-th/0211206}{{\tt arXiv:hep-th/0211206
  [hep-th]}}.

\bibitem{Hubeny:2002pj}
V.~E. Hubeny and M.~Rangamani, ``{No horizons in pp waves},''
  \href{http://dx.doi.org/10.1088/1126-6708/2002/11/021}{{\em JHEP} {\bf 0211}
  (2002)  021},
\href{http://arxiv.org/abs/hep-th/0210234}{{\tt arXiv:hep-th/0210234
  [hep-th]}}.

\bibitem{Liu:2003cta}
J.~T. Liu, L.~A. Pando~Zayas, and D.~Vaman, ``{On horizons and plane waves},''
  \href{http://dx.doi.org/10.1088/0264-9381/20/20/302}{{\em Class.Quant.Grav.}
  {\bf 20} (2003)  4343--4374},
\href{http://arxiv.org/abs/hep-th/0301187}{{\tt arXiv:hep-th/0301187
  [hep-th]}}.

\bibitem{Gimon:2003xk}
E.~G. Gimon, A.~Hashimoto, V.~E. Hubeny, O.~Lunin, and M.~Rangamani, ``{Black
  strings in asymptotically plane wave geometries},''
  \href{http://dx.doi.org/10.1088/1126-6708/2003/08/035}{{\em JHEP} {\bf 0308}
  (2003)  035},
\href{http://arxiv.org/abs/hep-th/0306131}{{\tt arXiv:hep-th/0306131
  [hep-th]}}.

\bibitem{Hubeny:2003ug}
V.~E. Hubeny and M.~Rangamani, ``{Horizons and plane waves: A Review},''
  \href{http://dx.doi.org/10.1142/S0217732303012428}{{\em Mod.Phys.Lett.} {\bf
  A18} (2003)  2699--2712},
\href{http://arxiv.org/abs/hep-th/0311053}{{\tt arXiv:hep-th/0311053
  [hep-th]}}.

\bibitem{LeWitt:2008zx}
J.~Le~Witt and S.~F. Ross, ``{Asymptotically Plane Wave Spacetimes and their
  Actions},'' \href{http://dx.doi.org/10.1088/1126-6708/2008/04/084}{{\em JHEP}
  {\bf 0804} (2008)  084},
\href{http://arxiv.org/abs/0801.4412}{{\tt arXiv:0801.4412 [hep-th]}}.

\bibitem{Nitsche:2011}
J.~C.~C. Nitsche, ``{Lectures on Minimal Surfaces},'' {\em Cambridge University
  Press; Reissue edition} {\bf 11} (2011)  .

\bibitem{Emparan:2011ve}
R.~Emparan and N.~Haddad, ``{Self-similar critical geometries at horizon
  intersections and mergers},''
  \href{http://dx.doi.org/10.1007/JHEP10(2011)064}{{\em JHEP} {\bf 1110} (2011)
   064},
\href{http://arxiv.org/abs/1109.1983}{{\tt arXiv:1109.1983 [hep-th]}}.

\bibitem{Caldarelli:2008pz}
M.~M. Caldarelli, R.~Emparan, and M.~J. Rodriguez, ``{Black Rings in
  (Anti)-de{S}itter space},''
  \href{http://dx.doi.org/10.1088/1126-6708/2008/11/011}{{\em JHEP} {\bf 11}
  (2008)  011},
\href{http://arxiv.org/abs/0806.1954}{{\tt arXiv:0806.1954 [hep-th]}}.

\bibitem{MeeksReview}
W.~H. Meeks~III and J.~Perez, ``{The classical theory of minimal surfaces},''
  \href{http://dx.doi.org/10.1090/S0273-0979-2011-01334-9}{{\em Bull. Amer.
  Math. Soc.} {\bf 48} (2011)  325--407}.

\bibitem{Barbosa1981}
J.~Barbosa and M.~do~Carmo, ``{Helicoids, catenoids, and minimal hypersurfaces
  of $R^{n}$ invariant by an $\ell$-parameter group of motions.},'' {\em An.
  Acad. Brasil. Cienc.} {\bf 53} (1981)  403--408.

\bibitem{Jorge:1984}
J.~Barbosa, M.~Dajczer, and L.~Jorge, ``{Minimal ruled sub manifolds in spaces
  of constant curvature},'' {\em Univ.Math.J. 33 (1984), no. 4, 531-547} (1984)
   .

\bibitem{Hoppe:2013}
J.~Choe and J.~Hoppe, ``{Higher dimensional minimal submanifolds generalizing
  the catenoid and helicoid},'' {\em Tohoku Math. J. (2) 65 (2013), no. 1,
  43-55} (2013)  .

\bibitem{Choe:1996}
J.~Choe, ``On the existence of higher dimensional {E}nneper's surface,''
  \href{http://dx.doi.org/10.1007/BF02566436}{{\em Commentarii Mathematici
  Helvetici} {\bf 71} (1996) no.~1, 556--569}.
  \url{http://dx.doi.org/10.1007/BF02566436}.

\bibitem{Kaabachi_riemannminimal}
S.~Kaabachi and F.~Pacard, ``Riemann minimal surfaces in higher dimensions,''
  {\em Journal of the Institute of Mathematics of Jussieu} {\bf 4} (2007)
  613--637.

\bibitem{Brendle2013}
S.~Brendle, ``Minimal surfaces in $s^3$: a survey of recent results,''
  \href{http://dx.doi.org/10.1007/s13373-013-0034-2}{{\em Bulletin of
  Mathematical Sciences} {\bf 3} (2013) no.~1, 133--171}.
  \url{http://dx.doi.org/10.1007/s13373-013-0034-2}.

\bibitem{Mira:2003}
L.~J. Alias, R.~Chaves, and P.~Mira, ``{Bjoring problem for maximal surfaces in
  Lorentz-Minkowski space},'' {\em Math. Proc. Cambridge Philos. Soc.} {\bf
  134} (2003) no.~2, 289--316.

\bibitem{LeeS:2008}
S.~Lee, ``{Weierstrass representation for timelike minimal surfaces in
  Minkowski 3-space},'' {\em Communications in Mathematical Analysis, Conf. 01}
  {\bf 01} (2008)  11--19.

\bibitem{Wook:2011}
Y.-W. Kim, S.-E. Koh, H.-Y. Shin, and S.-D. Yang, ``{Spacelike Maximal
  Surfaces, Timelike Minimal Surfaces, and Bjorling Representation Formulae},''
  \href{http://dx.doi.org/10.4134/JKMS.2011.48.5.1083}{{\em Journal of the
  Korean Mathematical Society Volume 48, Issue ,5, pp.1083-1100} (2011)  }.

\bibitem{Armas:2012bk}
J.~Armas, T.~Harmark, N.~A. Obers, M.~Orselli, and A.~V. Pedersen, ``{Thermal
  Giant Gravitons},'' \href{http://dx.doi.org/10.1007/JHEP11(2012)123}{{\em
  JHEP} {\bf 1211} (2012)  123},
\href{http://arxiv.org/abs/1207.2789}{{\tt arXiv:1207.2789 [hep-th]}}.

\bibitem{Armas:2013ota}
J.~Armas, N.~A. Obers, and A.~V. Pedersen, ``{Null-Wave Giant Gravitons from
  Thermal Spinning Brane Probes},''
\href{http://arxiv.org/abs/1306.2633}{{\tt arXiv:1306.2633 [hep-th]}}.

\bibitem{Armas:2011uf}
J.~Armas, J.~Camps, T.~Harmark, and N.~A. Obers, ``{The Young Modulus of Black
  Strings and the Fine Structure of Blackfolds},''
  \href{http://dx.doi.org/10.1007/JHEP02(2012)110}{{\em JHEP} {\bf 1202} (2012)
   110},
\href{http://arxiv.org/abs/1110.4835}{{\tt arXiv:1110.4835 [hep-th]}}.

\bibitem{Emparan:2007wm}
R.~Emparan, T.~Harmark, V.~Niarchos, N.~A. Obers, and M.~J. Rodriguez, ``{The
  Phase Structure of Higher-Dimensional Black Rings and Black Holes},''
  \href{http://dx.doi.org/10.1088/1126-6708/2007/10/110}{{\em JHEP} {\bf 10}
  (2007)  110},
\href{http://arxiv.org/abs/0708.2181}{{\tt arXiv:0708.2181 [hep-th]}}.

\bibitem{Camps:2012hw}
J.~Camps and R.~Emparan, ``{Derivation of the blackfold effective theory},''
  \href{http://dx.doi.org/10.1007/JHEP03(2012)038}{{\em JHEP} {\bf 1203} (2012)
   038},
\href{http://arxiv.org/abs/1201.3506}{{\tt arXiv:1201.3506 [hep-th]}}.

\bibitem{Grignani:2010xm}
G.~Grignani, T.~Harmark, A.~Marini, N.~A. Obers, and M.~Orselli, ``{Heating up
  the BIon},'' \href{http://dx.doi.org/10.1007/JHEP06(2011)058}{{\em JHEP} {\bf
  1106} (2011)  058},
\href{http://arxiv.org/abs/1012.1494}{{\tt arXiv:1012.1494 [hep-th]}}.

\bibitem{Grignani:2011mr}
G.~Grignani, T.~Harmark, A.~Marini, N.~A. Obers, and M.~Orselli,
  ``{Thermodynamics of the hot BIon},''
  \href{http://dx.doi.org/10.1016/j.nuclphysb.2011.06.002}{{\em Nucl.Phys.}
  {\bf B851} (2011)  462--480},
\href{http://arxiv.org/abs/1101.1297}{{\tt arXiv:1101.1297 [hep-th]}}.

\bibitem{Niarchos:2012pn}
V.~Niarchos and K.~Siampos, ``{M2-M5 blackfold funnels},''
  \href{http://dx.doi.org/10.1007/JHEP06(2012)175}{{\em JHEP} {\bf 1206} (2012)
   175},
\href{http://arxiv.org/abs/1205.1535}{{\tt arXiv:1205.1535 [hep-th]}}.

\bibitem{Niarchos:2012cy}
V.~Niarchos and K.~Siampos, ``{Entropy of the self-dual string soliton},''
\href{http://arxiv.org/abs/1206.2935}{{\tt arXiv:1206.2935 [hep-th]}}.

\bibitem{Niarchos:2013ia}
V.~Niarchos and K.~Siampos, ``{The black M2-M5 ring intersection spins},'' {\em
  PoS} {\bf Corfu2012} (2013)  088,
\href{http://arxiv.org/abs/1302.0854}{{\tt arXiv:1302.0854 [hep-th]}}.

\bibitem{OperaBook}
J.~Oprea, ``{The mathematics of soap films: explorations with Maple},'' {\em
  Amer Mathematical Society} (2000)  .

\bibitem{Sauvigny2010}
U.~Dierkes, S.~Hildebrandt, and F.~Sauvigny, ``{Minimal Surfaces},'' {\em
  Springer} {\bf 339} (2010)  .

\bibitem{Ogawa1992}
A.~Ogawa, ``{Helicatenoid},'' {\em The Mathematica Journal issue 2} {\bf 2}
  (1992)  .

\bibitem{Emparan:2003sy}
R.~Emparan and R.~C. Myers, ``{Instability of ultra-spinning black holes},''
  {\em JHEP} {\bf 0309} (2003)  025,
\href{http://arxiv.org/abs/hep-th/0308056}{{\tt arXiv:hep-th/0308056
  [hep-th]}}.

\bibitem{Armas:2015nea}
J.~Armas and M.~Blau, ``New geometries for black hole horizons,''
  \href{http://dx.doi.org/10.1007/JHEP07(2015)048}{{\em Journal of High Energy
  Physics} {\bf 2015} (2015) no.~7, }.
  \url{http://dx.doi.org/10.1007/JHEP07%282015%29048}.

\bibitem{Gray:2006}
E.~Abbena, S.~Salamon, and A.~Gray, ``{Modern Differential Geometry of Curves
  and Surfaces with Mathematica, Third Edition},'' {\em CRC Press} (2006)  .

\bibitem{LeeLee:2014}
E.~Lee and H.~Lee, ``{Generalizations of the Choe-Hoppe helicoid and Clifford
  cones in Euclidean space},'' \href{http://arxiv.org/abs/1410.3418}{{\tt
  arXiv:1410.3418 [math]}}.

\bibitem{2007arXiv0708.3310T}
L.-F. {Tam} and D.~{Zhou}, ``{Stability properties for the higher dimensional
  catenoid in $\mathbb{R}^{n+1}$},'' {\em Proc. Amer. Math. Soc.} {\bf 137}
  (2009)  3451--3461.

\end{thebibliography}
\end{document}